\def\llncs{0}
\def\fullpage{1}
\def\anonymous{0}
\def\authnote{0}
\def\notxfont{0}
\def\submission{0}
\def\reply{0}
\def\cameraready{0}
\def\noaux{1}
\ifnum\submission=1
\def\anonymous{1}
\def\llncs{1}
\else
\fi

\ifnum\cameraready=1
\def\llncs{1}
\def\anonymous{0}
\def\authnote{0}
\else
\fi

\ifnum\anonymous=1
\def\authnote{0}
\else
\fi

\def\mac{0}

\ifnum\mac=0
\ifnum\llncs=1
	\documentclass[envcountsect,a4paper,runningheads]{llncs}
\else
	\documentclass[letterpaper,hmargin=1.05in,vmargin=1.05in]{article}
			\ifnum\fullpage=1
		\usepackage{fullpage}
		\fi
\fi
\else
\ifnum\llncs=1
	\documentclass[dvipdfmx,envcountsect,a4paper,runningheads]{llncs}
\else
	\documentclass[dvipdfmx,letterpaper,hmargin=1.05in,vmargin=1.05in]{article}
			\ifnum\fullpage=1
		\usepackage{fullpage}
		\fi
\fi
\fi

\ifnum\reply=1
\usepackage{ulem}
\renewcommand{\emph}{\textit}
\else
\fi

\usepackage[%
  colorlinks=true,
  citecolor=darkgreen,
  linkcolor=darkblue,
  urlcolor=darkblue,
  pagebackref=true
]{hyperref}

\usepackage{amsmath, amsfonts, amssymb, mathtools,amscd}

\usepackage{amsthm}

\usepackage{lmodern}
\usepackage[T1]{fontenc}
\usepackage[utf8]{inputenc}

\usepackage{etex}

\usepackage{arydshln} % In order to use \hdashline
\usepackage{url}
\usepackage{ifthen}
\usepackage{bm}
\usepackage{multirow}
\usepackage[dvips]{graphicx}
\usepackage[usenames]{color}
\usepackage{xcolor,colortbl} % Leave out in case the usepackage cannot be found. Not critical
\usepackage{threeparttable}
\usepackage{comment}
\usepackage{paralist,verbatim}
\usepackage{cases}
\usepackage{booktabs}
\usepackage{braket}
\usepackage{cancel} 
\usepackage{framed}
\usepackage{authblk}
\usepackage{pifont}
\usepackage{physics}
\definecolor{darkblue}{rgb}{0,0,0.6}
\definecolor{darkgreen}{rgb}{0,0.5,0}
\definecolor{maroon}{rgb}{0.5,0.1,0.1}
\definecolor{dpurple}{rgb}{0.2,0,0.65}
\definecolor{chocolate}{rgb}{0.8,0.4,0.1}

\usepackage[capitalise,noabbrev]{cleveref}
\usepackage[absolute]{textpos}
\usepackage[final]{microtype}
\usepackage[absolute]{textpos}
\usepackage{autonum}

\usepackage{dsfont}
\DeclareMathAlphabet{\mathpzc}{OT1}{pzc}{m}{it}

\ifnum\submission=1
\renewcommand*{\backref}[1]{}
\def\notxfont{1}
\else
\fi

\ifnum\llncs=1
\renewcommand{\subparagraph}{\paragraph}
\else
\ifnum\notxfont=1
\else
\usepackage{mathpazo}
\usepackage{newtxtext}
\usepackage{helvet}
\fi
\fi

\newtheoremstyle{thicktheorem}%
{\topsep}
{\topsep}
{\itshape}{}%
{\bfseries}%
{.}
{ }%
{\thmname{#1}\thmnumber{ #2}%
		\thmnote{ (#3)}%
}

\newtheoremstyle{remark}%name
{\topsep}
{\topsep}
	{}%body font
	{}%indent amount
	{}%theorem head font
	{.}%punctuation after theorem head
	{ }%space after theorem head
	{\textit{\thmname{#1}}\thmnumber{ #2}%theorem head specs
			\thmnote{ (#3)}%
	}

\ifnum\llncs=0
	\theoremstyle{thicktheorem}
	\newtheorem{theorem}{Theorem}[section]
	\newtheorem{lemma}[theorem]{Lemma}
	\newtheorem{corollary}[theorem]{Corollary}
	
	\newtheorem{definition}[theorem]{Definition}

	\theoremstyle{remark}
	\newtheorem{claim}[theorem]{Claim}
	\newtheorem{remark}[theorem]{Remark}

\fi
	\crefname{theorem}{Theorem}{Theorems}
	\crefname{assumption}{Assumption}{Assumptions}
	\crefname{construction}{Construction}{Constructions}
	\crefname{corollary}{Corollary}{Corollaries}
	\crefname{conjecture}{Conjecture}{Conjectures}
	\crefname{definition}{Definition}{Definitions}
	\crefname{exmaple}{Example}{Examples}
	\crefname{experiment}{Experiment}{Experiments}
	\crefname{counterexample}{Counterexample}{Counterexamples}
	\crefname{lemma}{Lemma}{Lemmata}
	\crefname{observation}{Observation}{Observations}
	\crefname{proposition}{Proposition}{Propositions}
	\crefname{remark}{Remark}{Remarks}
	\crefname{claim}{Claim}{Claims}
	\crefname{fact}{Fact}{Facts}
	\crefname{note}{Note}{Notes}

\ifnum\llncs=1
 \crefname{appendix}{App.}{Appendices}
 \crefname{section}{Sec.}{Sections}
\else
\fi

\ifnum\llncs=1
\renewcommand*{\backref}[1]{}
\else
	\renewcommand*{\backref}[1]{(Cited on page~#1.)}
	\ifnum\notxfont=1
	\else
		\usepackage{newtxtext}
	\fi
\fi

\newcommand*{\keys}[1]{\mathsf{#1}}

\newcommand{\Oracle}[1]{O_{\mathtt{#1}}}

\newcommand*{\algo}[1]{\ensuremath{\mathsf{#1}}}
\newcommand*{\qalgo}[1]{\ensuremath{\mathpzc{#1}}}
\newcommand*{\qstate}[1]{\mathpzc{#1}}
\newcommand*{\qreg}[1]{{\color{gray}{\mathsf{#1}}}}

\newcounter{expitem}

\usepackage{mathtools}
\newcommand{\la}{\leftarrow}
\newcommand{\ra}{\rightarrow}
\renewcommand{\gets}{\leftarrow}

\newcommand{\seteq}{\coloneqq}

\newcommand{\qCh}{\qalgo{Ch}}

\newcommand{\cD}{\mathcal{D}}
\newcommand{\cE}{\mathcal{E}}
\newcommand{\cF}{\mathcal{F}}

\newcommand{\cH}{\mathcal{H}}
\newcommand{\cI}{\mathcal{I}}

\newcommand{\cM}{\mathcal{M}}
\newcommand{\cN}{\mathcal{N}}

\newcommand{\cP}{\mathcal{P}}

\newcommand{\cR}{\mathcal{R}}

\newcommand{\cX}{\mathcal{X}}
\newcommand{\cY}{\mathcal{Y}}
\newcommand{\cZ}{\mathcal{Z}}

\newcommand{\qA}{\qalgo{A}}
\newcommand{\qB}{\qalgo{B}}
\newcommand{\qC}{\qalgo{C}}
\newcommand{\qD}{\qalgo{D}}

\newcommand{\qS}{\qalgo{S}}
\def\makeuppercase#1{
\expandafter\newcommand\csname sf#1\endcsname{\mathsf{#1}}
\expandafter\newcommand\csname frak#1\endcsname{\mathfrak{#1}}
\expandafter\newcommand\csname bb#1\endcsname{\mathbb{#1}}
\expandafter\newcommand\csname bf#1\endcsname{\textbf{#1}}
}

\def\makelowercase#1{
\expandafter\newcommand\csname frak#1\endcsname{\mathfrak{#1}}
\expandafter\newcommand\csname bf#1\endcsname{\textbf{#1}}
}

\newcounter{char}
\loop
   \edef\letter{\alph{char}}
   \edef\Letter{\Alph{char}}
   \expandafter\makelowercase\letter
   \expandafter\makeuppercase\Letter
   \stepcounter{char}
   \unless\ifnum\thechar>26
\repeat

\def\makeuppercase#1{
\expandafter\newcommand\csname tl#1\endcsname{\widetilde{#1}}
}

\def\makelowercase#1{
\expandafter\newcommand\csname tl#1\endcsname{\widetilde{#1}}
}

\newcommand{\bit}{\{0,1\}}

\newcommand{\secp}{\lambda}

\newcommand{\id}{\mathsf{id}}

\newcommand{\cert}{\keys{cert}}

\newcommand{\aux}{\mathsf{aux}}

\newcommand{\expa}[2]{\mathsf{Expt}_{#1}^{\mathsf{#2}}}
\newcommand{\expb}[3]{\mathsf{Exp}_{#1}^{ \mathsf{#2} \mbox{-} \mathsf{#3}}}

\newcommand{\Hyb}{\mathsf{Hyb}}

\newcommand*{\pk}{\keys{pk}}
\newcommand*{\sk}{\keys{sk}}

\newcommand*{\vk}{\keys{vk}}

\newcommand*{\ssk}{\keys{ssk}}
\newcommand*{\svk}{\keys{svk}}

\newcommand*{\key}{\keys{k}}
\newcommand*{\msk}{\keys{msk}}

\newcommand*{\pp}{\keys{pp}}

\newcommand*{\tk}{\keys{tk}}

\newcommand*{\xk}{\keys{xk}}

\newcommand*{\ct}{\keys{ct}}

\newcommand*{\msg}{\keys{m}}

\newcommand{\qsk}{\qstate{sk}}

\newcommand{\SD}{\mathsf{SD}}
\newcommand{\TD}{\mathsf{TD}}

\newenvironment{boxfig}[2]{\begin{figure}[#1]\fbox{\begin{minipage}{0.97\linewidth}
                        \vspace{0.2em}
                        \makebox[0.025\linewidth]{}
                        \begin{minipage}{0.95\linewidth}
            {{
                        #2 }}
                        \end{minipage}
                        \vspace{0.2em}
                        \end{minipage}}
                        }
                        {\end{figure}}

\newcommand{\Mac}{\algo{Mac}}

\newcommand{\win}{\mathtt{Win}}

\newcommand{\Event}{\mathtt{E}}

\newcommand{\Setup}{\algo{Setup}}

\newcommand{\Gen}{\algo{Gen}}

\newcommand{\KG}{\algo{KG}}

\newcommand{\Sign}{\algo{Sign}}
\newcommand{\Vrfy}{\algo{Vrfy}}

\newcommand{\qKG}{\qalgo{KG}}

\newcommand{\qDel}{\qalgo{Del}}

\newcommand{\qR}{\qalgo{R}}
\newcommand{\E}{\algo{E}}

\newcommand{\PRF}{\algo{PRF}}

\newcommand{\Eval}{\algo{Eval}}

\newcommand{\Mark}{\mathsf{Mark}}

\newcommand{\negl}{{\mathsf{negl}}}
\newcommand{\nonnegl}{{\mathsf{non}\textrm{-}\mathsf{negl}}}
\newcommand{\poly}{{\mathrm{poly}}}

\newcommand{\xor}{\oplus}

\newcommand{\tlC}{\widetilde{C}}

\newcommand{\SKL}{\algo{SKL}}

\newcommand{\qEval}{\qalgo{Eval}}

\newcommand{\qP}{\qalgo{P}}

\newcommand{\MLTT}{\algo{MLTT}}

\newcommand{\TMac}{\mathsf{T}\mathsf{\Mac}}
\newcommand{\tmac}{\mathsf{tmac}}
\newcommand{\TG}{\qalgo{TG}}
\newcommand{\qSign}{\qalgo{Sign}}
\newcommand{\qtk}{\qstate{tk}}

\newcommand{\qTrace}{\qalgo{Trace}}
\newcommand{\tTrace}{\algo{Trace}}
\newcommand{\mltt}{\mathsf{mltt}}
\newcommand{\Win}{\mathsf{Win}}
\newcommand{\TPRF}{\algo{TPRF}}
\newcommand{\API}{\qalgo{API}}
\newcommand{\FC}{\algo{FC}}
\newcommand{\tprf}{\mathsf{tprf}}
\newcommand{\fc}{\mathsf{fc}}

\newcommand{\Est}{\qalgo{EST}}
\newcommand{\Repair}{\qalgo{Repair}}

\newcommand{\qF}{\qalgo{F}}
\newcommand{\qE}{\qalgo{E}}
\newcommand{\chec}{\algo{Check}}
\newcommand{\out}{\mathsf{out}}
\newcommand{\WMPRF}{\mathsf{WMPRF}}
\newcommand{\qExtract}{\qalgo{Extract}}
\newcommand{\prfk}{\mathsf{prfk}}
\newcommand{\Dchall}{D_{b}^{\mathsf{chall}}}
\newcommand{\Dwprf}{D^{\mathsf{wprf}}}
\newcommand{\MLTPRF}{\mathsf{MLT}\textrm{-}\mathsf{PRF}}
\newcommand{\ProjImp}{\mathsf{ProjImp}}
\newcommand{\counts}{\mathsf{count}}

\newcommand{\SigVrfy}{\algo{SigVrfy}}
\newcommand{\DSSKL}{\algo{DS}\textrm{-}\algo{SKL}}
\newcommand{\DSig}{\algo{DSig}}
\newcommand{\sig}{\algo{sig}}

\newcommand{\Live}{\mathtt{Live}}
\newcommand{\APILive}{\mathtt{APILive}}
\newcommand{\GoodTrace}{\mathtt{GoodTrace}}
\newcommand{\GoodExt}{\mathtt{GoodExt}}
\newcommand{\BadTrace}{\mathtt{BadTrace}}
\newcommand{\BadCode}{\mathtt{BadCode}}
\newcommand{\BadBit}{\mathtt{BadBit}}
\newcommand{\BadID}{\mathtt{BadID}}
\newcommand{\NoAbort}{\mathtt{NoAbort}}
\newcommand{\QPRF}{\mathsf{QPRF}}
\newcommand{\CorVrfy}{\algo{CorVrfy}}
\newcommand{\qg}{\qstate{g}}

\newcommand{\authornote}[3]{\textcolor{#3}{[\textsc{#1:} {#2}]}}
\ifnum\authnote=1
\newcommand{\fuyuki}[1]{\authornote{Fuyuki}{#1}{chocolate}}
\newcommand{\ryo}[1]{\authornote{Ryo}{#1}{darkblue}}
\newcommand{\nikhil}[1]{\authornote{Nikhil}{#1}{red}}
\else
\newcommand{\fuyuki}[1]{}
\newcommand{\ryo}[1]{}
\newcommand{\nikhil}[1]{}
\fi

\ifnum\llncs=1
\let\oldvec\vec% Store \vec in \oldvec
\let\vec\oldvec% Restore \vec from \oldvec
\makeatletter
\renewcommand*\l@author[2]{}
\renewcommand*\l@title[2]{}
\makeatletter
\fi    
\theoremstyle{remark}

\ifnum\submission=1
\title{
\textbf{Collusion-Resistant Quantum Secure Key Leasing
Beyond Decryption}%\thanks{{\color{red}{\emph{We attached the full version of this paper as a supplementary material}}}.}
}
\else
\title{
\textbf{Collusion-Resistant Quantum Secure Key Leasing Beyond Decryption}
}
\fi

\begin{document}
%\author{}
%\institute{}

\ifnum\anonymous=1 
\ifnum\llncs=1
% for anonymized LNCS version
\author{\empty}\institute{\empty}
\else
% for anonymized full version
\author{}
\fi
\else
%
%  For camera ready version.
%
\ifnum\llncs=1
\author{
	Fuyuki Kitagawa\inst{1,2} \and Ryo Nishimaki\inst{1,2} \and Nikhil Pappu\inst{3}
}
\institute{
	NTT Social Informatics Laboratories, Tokyo, Japan \and NTT Research Center for Theoretical Quantum Information, Atsugi, Japan \and Portland State University, USA
}
\else
%
%   For full/eprint version, etc.
%
\author[$\dagger$ $\diamondsuit$]{\hskip 1em Fuyuki Kitagawa}
\author[$\dagger$ $\diamondsuit$]{\hskip 1em Ryo Nishimaki}
\author[$\star$]{\hskip 1em Nikhil Pappu \thanks{Supported by the US National Science Foundation (NSF) via Fang Song's Career Award (CCF-2054758).}}
\affil[$\dagger$]{{\small NTT Social Informatics Laboratories, Tokyo, Japan}\authorcr{\small \{fuyuki.kitagawa,ryo.nishimaki\}@ntt.com}}
\affil[$\diamondsuit$]{{\small NTT Research Center for Theoretical Quantum Information, Atsugi, Japan}}
\affil[$\star$]{{\small Portland State University, USA}\authorcr{\small nikpappu@pdx.edu}}
\renewcommand\Authands{, }
\fi %%%%% END OF LNCS branch
\fi

\ifnum\llncs=1
\date{}
\else
\ifnum\anonymous=0
\date{\today}
\else
\date{}
\fi
\fi

\maketitle

\begin{abstract}
Secure key leasing (SKL) enables the holder of a secret key for a
cryptographic function to temporarily lease the key using quantum
information. Later, the recipient can produce a \emph{deletion
certificate}—a proof that they no longer have access to the secret
key.  The security guarantee ensures that even a \emph{malicious}
recipient cannot continue to evaluate the function, after producing
a valid deletion certificate.

Most prior work considers an adversarial recipient that obtains a single
leased key, which is insufficient for many applications. In
the more realistic \emph{collusion-resistant} setting, security must
hold even when polynomially many keys are leased (and subsequently
deleted). However, achieving collusion-resistant SKL from standard
assumptions remains poorly understood, especially for
functionalities beyond decryption.

We improve upon this situation by introducing new pathways for
constructing collusion-resistant SKL. Our main contributions
are as follows:

\begin{itemize}
\item A generalization of quantum-secure collusion-resistant traitor
tracing called multi-level traitor tracing (MLTT), and a compiler
that transforms an MLTT scheme for a primitive $X$ into a
collusion-resistant SKL scheme for primitive $X$.

\item The first bounded collusion-resistant SKL scheme for PRFs, assuming
LWE.

\item A compiler that upgrades any single-key secure SKL scheme for
digital signatures into one with unbounded collusion-resistance,
assuming OWFs.

\item A compiler that upgrades collusion-resistant SKL schemes with
classical certificates to ones having verification-query resilience,
assuming OWFs.
\end{itemize}

%Along the way, we generalize an information-theoretic property
%of random two-superposition states, and employ techniques from
%quantum-secure traitor tracing to construct an MLTT scheme for
%PRFs, which might be of independent interest.
\end{abstract}

%\keywords{Secure Key Leasing, Quantum-Secure Traitor Tracing}

\ifnum\llncs=1
\else
\newpage
\setcounter{tocdepth}{2}
\tableofcontents

\newpage
\fi

% !TEX root = main.tex

\section{Introduction}\label{sec:intro}

\textbf{Unclonable cryptography and copy protection.}
Unclonable cryptography is a prominent subfield of quantum
cryptography, which aims to leverage uniquely quantum phenomena to
achieve guarantees that are impossible in the classical world. A
central object of study in this area is quantum copy protection
\cite{Aar09}, which has attracted significant interest in recent years
\cite{C:ALLZZ21,C:CLLZ21,TCC:LLQZ22,STOC:ColGun24,CMP24,TCC:CheHerVu23,TCC:CakGoy24}. At a
high level, copy protection encodes software into quantum states in a way that
preserves functionality, while preventing the creation of functionally
equivalent copies. If realized, this primitive would be particularly
valuable to software distributors seeking to combat piracy.
Despite its appeal, currently known schemes for copy protection in the
plain model have a major limitation: they rely on the assumption of
indistinguishability obfuscation (iO) \cite{C:BGIRSV01} \footnote{The
work of Coladangelo et al. \cite{CMP24} presents a copy protection
scheme in the random oracle model from standard assumptions, without relying on more
structured oracles. However, their results are for 
some evasive functions, and not for cryptographic ones.}. iO is a strong cryptographic primitive
known to imply much of modern cryptography. As a result, constructions
based on iO are generally inefficient, and often considered “overkill”
compared to approaches based on simpler assumptions. Moreover, we do
not yet know how to construct \emph{post-quantum} iO from well-studied
assumptions.

\textbf{Secure software leasing (SSL).} A notion related to
copy protection is that of secure software leasing (SSL)
\cite{EC:AnaLaP21}. In
SSL, an entity called the \emph{lessor} can provide a quantum secret key to
an entity called the \emph{lessee}, which enables the lessee to evaluate some
function. Later, the lessee can be asked to revoke (delete) this key,
which if verified successfully, guarantees it can no longer
evaluate the function. Hence, SSL is a relaxation of copy protection
in a sense, as it does not prevent the creation of equivalent
copies. It merely ensures that an adversary creating
such copies cannot pass verification. Unfortunately,
Ananth and La Placa \cite{EC:AnaLaP21} showed that both SSL and copy protection are not
possible to achieve for some unlearnable functions \footnote{These notions are
trivially impossible for learnable functions.}. However, SSL for
certain cryptographic functionalities has been realized from standard
assumptions \cite{TCC:KitNisYam21}, unlike copy protection.

\textbf{Secure key leasing (SKL).}
Apart from the natural security guarantee of SSL, its tractability
from standard assumptions makes it particularly appealing.
However, a drawback of most SSL notions is that security is only
guaranteed for an adversary that ``honestly'' evaluates the function,
after deleting its key. This is quite unrealistic for
cryptographic settings. Hence, recent works have introduced a
notion called secure key leasing (SKL). This is essentially SSL for
cryptographic functions, except that the adversaries may try to evaluate the function
arbitrarily. Note that existing copy protection schemes 
imply SKL (As discussed in \cite{EC:AKNYY23}), but this pathway
suffers from the aforementioned iO limitation. Hence, previous works
\cite{EC:AKNYY23,TCC:AnaPorVai23,EC:CGJL25,EC:KitMorYam25,TCC:AnaHuHua24} have constructed
SKL for PKE, PRFs and signatures from standard assumptions.

\textbf{Collusion-resistant SKL.}
Even though SKL provides a more realistic guarantee than SSL, most
previous works consider an adversary that can only obtain a single
leased key. This is rather limiting for practical scenarios.
For e.g., in the case of PKE, a
service may want to broadcast encrypted data and lease out several
decryption keys to different users for a fee. The users could be
granted a refund if they return their key within some trial period.
Here, one would expect that a coalition of users (or a user requesting multiple
keys) is unable to cheat the system by retaining access without
spending money. However, these attacks are not captured by the usual
definition. This motivates a stronger notion of SKL where the
adversary is provided with polynomially many secret keys,
and security is guaranteed if all the keys are returned.
Such a setting was considered in the work of Kitagawa et al.
\cite{C:KitNisPap25}, where they constructed unbounded
collusion-resistant SKL for the decryption functionality of a PKE
scheme, based on the LWE assumption. The ``unbounded'' prefix refers to the
parameters of the scheme growing poly-logarithmically with the
collusion-bound, rather than polynomially in the ``bounded'' case. To
the best of our knowledge, this work and the work of Agrawal et al.
\cite{EC:AKNYY23} are the only ones that study collusion-resistant SKL
from weaker than iO assumptions. Even though they obtain
several positive results for variants of encryption,
their techniques do not seem
applicable to other primitives such as PRFs and signatures. Hence, we
are not aware of general techniques to achieve collusion-resistant
SKL (See Section \ref{sec:chall} for more details).

\textbf{Collusion-resistant SKL for PRFs and Signatures.}
A natural problem in the study of SKL concerns the leasing of secret-keys
of PRFs and signatures, as they form the backbone of classical cryptography. For
e.g., PRFs imply primitives such as symmetric-key
encryption and message authentication codes. Hence, SKL for
PRFs (PRF-SKL) and digital signatures (DS-SKL) can serve as
building blocks for ``higher-level'' SKL primitives. While some of
these implications follow easily, we leave the study of composition
of SKL primitives to future work.

Despite previous works obtaining PRF-SKL from standard assumptions
\cite{TCC:AnaPorVai23,EC:KitMorYam25}, we are currently unaware of
schemes that even satisfy bounded collusion-resistance. In contrast,
PKE-SKL is trivial to achieve with bounded collusion-resistance,
assuming a single-key secure PKE-SKL scheme. In-short, this is because
the core challenge in obtaining bounded collusion-resistant SKL can be
avoided for PKE-SKL, due to a workaround. However, it is not clear if such workarounds exist for other primitives such as PRFs.
For a detailed discussion, see Section \ref{sec:chall}.

One can also envision several use cases for bounded
collusion-resistant PRF-SKL, apart from its use in higher-level
primitives. For e.g., one can embed PRF keys in
trial versions of video games that are to be returned.
Security can
then be at-least heuristically argued, by tying the functionality
to the evaluation of the PRF. It is natural to expect
collusion-resistance here, as subsets of users may be capable of
colluding to cheat the system.
We also expect PRF-SKL to find applications in
multi-party protocols where PRFs are ubiquitous. For
e.g., one can imagine an MPC protocol where subsets of users
share common PRF keys, which are used to obtain common
randomness. However, in a dynamic system where some users may
eventually leave, it could be necessary to ensure 
revocation of their keys for the security of future iterations. The other
alternative is to update the entire setup when a user leaves,
which is not ideal. Clearly, a single-key secure PRF-SKL scheme is
insufficient due to the possibility of collusion. While
achieving unbounded collusion-resistance is ideal, 
an improvement from single-key to bounded security is also
substantial for applications, perhaps more than the jump from bounded
to unbounded security. Hence, obtaining any form of
collusion-resistant PRF-SKL is an important step in the study of SKL.

Similar to PRF-SKL, DS-SKL was also constructed from standard
assumptions \cite{EC:KitMorYam25}, but even bounded
collusion-resistance remains unclear. This is also a natural
primitive, as one can consider applications where 
employees need to sign on behalf of a company, but may eventually 
leave the company. In light
of the above discussion, we pose the following questions regarding
collusion-resistant SKL:

\begin{center}
\emph{1) What pathways exist for constructing collusion-resistant
SKL schemes from standard assumptions?}

\emph{2) Is it possible
to construct collusion-resistant SKL from standard assumptions for:
i) PRFs;
ii) Digital Signatures?
}
\end{center}

In this work, we present a new approach to constructing collusion-resistant SKL,
based on the notion of quantum-secure collusion-resistant
traitor-tracing \cite{TCC:Zhandry20,C:Zhandry23}. We believe this
advances the current understanding of the first question.
Traitor tracing is a primitive that enables the generation of secret-keys
meant for different users (identities), where all the keys allow to evaluate
some common functionality. The interesting aspect is the security
notion, which disincentivizes (possibly colluding) users from leaking
their keys. Essentially, this is achieved with the help of a tracing
algorithm that recovers the identity of at least one cheating user
(a traitor), from any pirate program capable of evaluating the
functionality. While this is a classical primitive, its quantum-secure
variant requires that a traitor be identified even if the pirate program
is a quantum state. Hence, the tracing algorithm is a quantum one, with all
other algorithms being classical.

We then utilize this approach to construct bounded
collusion-resistant PRF-SKL based on the LWE
assumption. While this pathway could also be employed for digital
signatures, we identify a simpler approach in this case that
provides unbounded collusion-resistance. Combined with a result of
prior work \cite{EC:KitMorYam25}, this gives us unbounded
collusion-resistant signatures (with the desirable notion of static
signing-keys) from the SIS assumption, answering our second question
in the affirmative. We now describe our contributions in more
detail.

\subsection{Our Contributions}

\begin{enumerate}[(1)]
\item \emph{New Definitions:} First, we define a generic
collusion-resistant SKL scheme that captures SKL for different
primitives as special cases. Similar
generic definitions for variants of copy protection were provided in
the work of Aaronson et al. \cite{C:ALLZZ21}. Furthermore, we abstract a new
primitive called multi-level traitor tracing (MLTT), which generalizes the 
notion of (quantum-secure collusion-resistant) traitor
tracing\footnote{Unlike standard traitor tracing, an MLTT adversary
only receives keys for randomly chosen identities, as it
suffices for the purpose of SKL. On the other hand, we require a
certain deterministic evaluation property (See Section \ref{sec:abs}).
}.
Our definition for MLTT is also a generic one that captures
different cryptographic applications as special cases.

\item \emph{A Compiler for SKL:} Then, we construct a
collusion-resistant SKL scheme from any MLTT scheme, without
other cryptographic assumptions. Based on certain parameters of the MLTT
scheme, the resulting SKL scheme offers either bounded or unbounded
collusion-resistance.

\item \emph{Generalized Hardcore Bit Property:} The main
quantum ingredient of our compiler are quantum states of the form
$\frac{1}{\sqrt2}(\ket{x} + (-1)^{b}\ket{y})$ where $x, y$ are random
strings and $b$ is a random bit. Previous works showed that no
adversary can produce one among $\{x, y\}$ along with a value $d$ such
that $d\cdot(x \xor y) = b$, with probability significantly more than
$1/2$. We generalize this result to the setting where the adversary
receives $q = \poly(\secp)$ many such states. Our result shows that no
adversary can produce \emph{one} of the $2q$ superposition terms,
along with values $d_1, \ldots, d_q$ consistent with the phases $b_1,
\ldots, b_q$ of each of the states, with probability much more than
$1/2$. 

\item \emph{Collusion-Resistant SKL for PRFs:} We construct a PRF-SKL
scheme with bounded collusion-resistance from the post-quantum
security of LWE with sub-exponential modulus. We achieve this by first
constructing an MLTT scheme for PRFs, followed by instantiating our
SKL compiler with this scheme. Constructing this MLTT scheme
presents several challenges due to the requirement of tracing quantum
adversaries. We overcome them by using
a quantum-secure traceable PRF by Kitagawa and Nishimaki
\cite{EC:KitNis22} as a building block, along with quantum 
tracing techniques of Zhandry \cite{TCC:Zhandry20,C:Zhandry23}.

\item \emph{Collusion-Resistant SKL for Signatures:} We present a
generic transformation that upgrades any single-key secure DS-SKL
scheme into one satisfying unbounded collusion-resistance, by
assuming the existence of post-quantum digital signatures. Based on
prior work \cite{EC:KitMorYam25}, this implies unbounded
collusion-resistant DS-SKL with static signing keys, based on the SIS assumption.

\item \emph{A Compiler for Verification-Oracle Security:} The
aforementioned SKL notion does not provide the adversary with oracle access
to the verification algorithm. However, the result of
verification is often leaked, making this impractical for 
applications. Hence, we consider a stronger notion called security
under
verification-oracle aided key-leasing attacks (Definition
\ref{def:vo-kla}), based on a similar notion by Kitagawa et al.
\cite{C:KitNisPap25}. We show that any SKL
scheme can be upgraded to satisfy this stronger security, assuming the
base scheme satisfies classical revocation. Indeed, our SKL compiler
satisfies this property. Hence, we are able to upgrade the results
from bullets (2), (4) and (5) to this stronger notion, 
assuming OWFs exist (this assumption is
redundant for (4) and (5)).

\end{enumerate}

\subsection{Related Work}

\paragraph{Compilers for Copy-Detection and
SSL.}\label{sec:compilers}

In the work of Aaronson et al.  \cite{C:ALLZZ21}, a relaxation of
copy-protection called copy-detection was introduced. They
showed a compiler that transforms any publicly-extractable
watermarking scheme for some application into a copy-detection scheme
for the same application, by assuming public-key quantum
money. Note that watermarking is a classical primitive
similar to traitor tracing. It consists of
an extraction algorithm that is analogous to the tracing algorithm,
and public extraction refers to the fact that the algorithm does not
utilize secret information.
In a concurrent work, Kitagawa,
Nishimaki and Yamakawa \cite{TCC:KitNisYam21} showed a similar compiler for SSL.
Interestingly, they showed that for SSL, the public-key quantum money
assumption can be replaced by a weaker primitive called two-tier
quantum lightning, which is implied by LWE. They also showed that the
watermarking assumption can be weakened to a notion called
relaxed-watermarking. These compilers
allows to obtain collusion-resistant copy-detection/SSL as well, by
relying on collusion-resistant watermarking.

Although our compiler is similar to these in spirit, it is both
conceptually and technically different. We now mention some of the key
differences:

\begin{enumerate}
\item Our compiler achieves SKL, a significantly
challenging task compared to SSL. This stems from the fact that
SKL schemes cannot enforce structure on the adversary's
post-revocation evaluation attempts, unlike SSL/copy-detection.

\item Our compiler only needs a form of
traitor-tracing with private-tracing, which is analogous to
private-extractable watermarking. In contrast, the aforementioned
compilers need public-extractable watermarking. In the
collusion-resistant setting, this makes a difference because many
collusion-resistant tracing/watermarking schemes satisfy only private
tracing/extraction.

\item Our compiler needs a
special kind of traitor-tracing we define called multi-level
traitor-tracing (MLTT) (Section \ref{sec:abs}), while the previous
compilers could rely on the standard notion of watermarking. However,
we believe MLTT to be a natural and feasible generalization of traitor
tracing. We demonstrate this by constructing an MLTT scheme for PRFs
in this work. Our compiler also requires
quantum-secure tracing unlike the SSL/copy-detection ones, as
they can simply enforce the adversary to output classical pirate
programs.

\item Our compiler does not require additional assumptions. In
contrast, the aforementioned compilers \cite{C:ALLZZ21,TCC:KitNisYam21}
utilized either public-key quantum money or two-tier quantum
lightning. We achieve this
using a new information-theoretic guarantee provided by two-superposition
states (Section \ref{sec:two-sup}).
\end{enumerate}

\paragraph{Secure Key Leasing.}
The notion of SKL was introduced by the concurrent works of Agrawal et
al. \cite{EC:AKNYY23} and Ananth et al. \cite{TCC:AnaPorVai23}. The
former constructed SKL for PKE (PKE-SKL) from any PKE scheme,
while the latter constructed PKE-SKL based on LWE. The former work
also presented SKL for other encryption notions (such as ABE
and PKFE), while the latter also presented a PRF-SKL scheme from
LWE. Note that the SKL scheme for PKFE 
of Agrawal et al. relies on PKFE, which implies iO.  The work
of Bartusek et al. \cite{EC:BGKMRR24} showed SKL based on a new primitive they
constructed called differing inputs obfuscation with certified
deletion. Hence, their scheme provides SKL for all differing input
circuit families. This allows them to achieve SKL for PKFE, and also
PRFs, but their SKL schemes inherently rely on iO.  The recent work of
Kitagawa, Morimae and Yamakawa \cite{EC:KitMorYam25} showed a simple
framework for constructing SKL schemes for PKE, PRFs and 
signatures based on standard assumptions and certified deletion properties of BB84 states. The work of Chardouvelis et al. \cite{EC:CGJL25} showed
that PKE-SKL can be realized by using only classical communication,
based on the hardness of LWE.  While the work of Ananth et al.
\cite{TCC:AnaPorVai23} relied on a complexity theoretic conjecture
apart from LWE, this was removed in the work of Ananth, Hu and Huang
\cite{TCC:AnaHuHua24}, thereby obtaining PKE-SKL and PRF-SKL from LWE
alone. Recently, Kitagawa, Nishimaki and Pappu
\cite{C:KitNisPap25} constructed unbounded collusion-resistant
PKE-SKL from LWE. This is currently the only work that obtains this
notion from a weaker than iO (or PKFE) assumption. We are not aware
of any works that study collusion-resistant SKL for primitives
other than encryption.

\subsubsection*{Collusion-Resistant Copy Protection.}
The first collusion resistant copy protection schemes in the plain
model were shown in the work of Liu et al. \cite{TCC:LLQZ22}. They
constructed copy protection schemes for PKE, PRFs and digital
signatures that are $k \ra k+1$ secure, i.e., an adversary receiving
$k = \poly(\secp)$ many copies cannot produce $k+1$ copies.
Importantly, their scheme is bounded collusion-resistant, i.e., the
parameter sizes grow linearly with the collusion-bound $k$. In the work of
{\c C}akan and Goyal~\cite{TCC:CakGoy24}, unbounded collusion-resistant copy protection
schemes were constructed in the plain model for PKE, PKFE, PRFs, and
digital signatures. Although these schemes imply SKL with similar
collusion-resistance guarantees (See the discussion in
\cite{EC:AKNYY23}), they all rely on iO. Since, copy protection is known
to imply public-key quantum money in general, achieving it with weaker than iO
assumptions is a major open problem.

\subsubsection*{Quantum-Secure Traitor Tracing.}

Traitor tracing in the quantum setting was first explored in the work
of Zhandry \cite{TCC:Zhandry20}. The work identifies several
challenges of dealing with quantum adversaries that output quantum
pirate programs. Firstly, it is not possible to know the success
probability of a quantum pirate until it is measured. This is because
the pirate may be in a superposition of ``successful'' and
``unsuccessful'' pirates.  Moreover, a measurement may disturb the
pirate and render it useless. Consequently, the work presented
workarounds for such definitional issues. Additionally, estimating the
success probability is also challenging, as it requires testing the
adversary on several samples from a distribution. The work shows an
efficient quantum procedure for this task by building on the work of
Marriott and Watrous \cite{CC:MarWat05}. Furthermore, the work shows a useful
property: consecutive estimations of the adversary's success
probability on computationally close distributions produce similar
outcomes. These quantum tools were leveraged to extend classical
private linear broadcast encryption (PLBE) based tracing schemes
\cite{EC:BonSahWat06} to the quantum setting. In a followup work by Zhandry
\cite{C:Zhandry23}, a new quantum rewinding technique due to Chiesa et
al. \cite{FOCS:CMSZ21} was utilized to further expand the kind of
probability estimations that can be performed without destroying the
pirate. The work utilized this
technique to extend several more classical tracing
schemes for PKE to the quantum setting, including several collusion-resistant ones.

The work of Kitagawa and Nishimaki \cite{EC:KitNis22} explores the setting of
watermarking in the quantum setting, which is a notion similar to that
of traitor tracing. Specifically, they showed a watermarkable PRF
secure against quantum adversaries (that output quantum pirate
programs), based on the hardness of LWE. We utilize this PRF as a
building block in our MLTT construction for PRFs (Section
\ref{sec:mlt-prf}).  Recently, Kitagawa and Nishimaki
\cite{EPRINT:KitNis25}
constructed a watermarkable digital signature scheme that is secure
against quantum adversaries, in the setting of white-box
traitor-tracing introduced by Zhandry \cite{C:Zhandry21}.

\subsubsection*{Certified Deletion.}

The notion of certified deletion for encryption was introduced in the
work of Broadbent and Islam \cite{TCC:BroIsl20}. This primitive allows
the generation of quantum ciphertexts, which can be provably deleted
by presenting a classical certificate. After deletion, even if the
secret-key is revealed, one cannot learn the contents of the
ciphertext they once held. Observe that the difference between this
notion and SKL is that here, it is access to secret data that is being
``revoked'', rather than the ability to evaluate some function.
Following this work, several other works have studied certified
deletion for different primitives
\cite{AC:HMNY21,ITCS:Poremba23,C:BarKhu23,EC:HKMNPY24,EC:BGKMRR24} and
with publicly-verifiable deletion
\cite{AC:HMNY21,EC:BGKMRR24,TCC:KitNisYam23,TCC:BKMPW23}. Recently,
the work of Ananth, Mutreja and Poremba \cite{EPRINT:AnaMutPor24}
introduced multi-copy revocable encryption.  This is a notion similar
to certified deletion but guarantees security in a setting where
multiple copies of the quantum ciphertext are provided to the
adversary. Note that this is in contrast to the collusion-resistant
setting we consider, where multiple i.i.d leased-keys are provided,
instead of identical copies of the same quantum state.

% !TEX root = main.tex
\section{Technical Overview}\label{sec:to}

\subsection{Collusion-Resistant SKL}

We will begin by defining the notion of SKL in the collusion setting. A
collusion-resistant SKL scheme $\SKL$ for a cryptographic application
$(\qF, \qE, t)$ (Definition \ref{def:app}) consists of five
algorithms $(\Setup, \qKG, \qEval, \qDel\allowbreak, \Vrfy)$. The
setup algorithm takes a
collusion-bound $q$ as input,\footnote{$q = \bot$ is a valid input, which is
meant for the unbounded collusion setting.} and
outputs a tuple $(\msk, f, \aux_f)$.  Here, $\msk$ is a master
secret-key, $f \in \cF$ is a function that is to be leased, and
$\aux_f$ is some auxiliary information that is to be made public. For
instance, in the case of PKE, $f$ is a decryption function described
by a PKE decryption key, while $\aux_f$ contains the public encryption
key. Now, the quantum key generation algorithm $\qKG$ takes as input
$\msk$ (we will assume $\msk$ implicitly includes $f$ and $\aux_f$),
and outputs a quantum secret-key $\qsk$ along with a classical
verification-key $\vk$.

Consider now a setting where an entity called the \emph{lessor} samples
$(\msk, f, \aux_f) \allowbreak \la \Setup(1^\secp, q)$ and $(\qsk, \vk) \la
\qKG(\msk)$. The lessor then provides an entity called the
\emph{lessee} with $\qsk$ along with $\aux_f$. It should be feasible
for the lessee to evaluate $f$ using $\qsk$ using the algorithm
$\qEval$, i.e., $\qEval(\qsk, x)$ should output $y = f(x)$ with
overwhelming probability. This correctness guarantee is specified by
the quantum predicate $\qF$ (Definition \ref{def:cor-pred}).
At a later point, the lessee can be asked to ``revoke'' its
secret-key. Then, the lessee can use the deletion algorithm to compute
a classical certificate $\cert \la \qDel(\qsk)$, which can be verified
by the lessor by evaluating the verification algorithm $\Vrfy(\vk,
\cert)$ with the key $\vk$. We require,
verification correctness, i.e., $\cert$ produced as
above should be accepted.

The crucial part is the security guarantee where we consider a QPT
adversary $\qA$ that receives $q = \poly(\secp)$ many leased
secret-keys $(\qsk_1, \ldots, \qsk_q)$ generated as $(\qsk_i, \vk_i)
\la \qKG(\msk)$ for each $i \in [q]$. In the bounded collusion
setting, $q$ is determined by the scheme, while it is
specified by the adversary in the unbounded case. Then, the adversary
produces (possibly malformed) certificates $(\cert_1, \ldots,
\cert_q)$. If $\Vrfy(\vk_i, \cert_i) = \top$ for each $i \in [q]$,
i.e., if the adversary successfully revokes all the leased keys, then
it should lose the ability to evaluate $f$.
Care must be taken in formalizing the inability to evaluate $f$
because in several applications like PKE, PRFs etc, we need a stronger
requirement than the inability to produce outputs of $f$. In the PKE
case, we would hope that the lessor cannot distinguish ciphertexts of
different messages. In the PRF case, the approach used by
prior works \cite{TCC:AnaPorVai23,EC:KitMorYam25} is that the lessor
cannot distinguish a random function from the PRF, when
given access to either function on random inputs. We formalize the
security using the quantum predicate $\qE$ (Definition
\ref{def:sec-pred}).

In more detail, we require $\qA$ to output a quantum program $\qP^* =
(U^*, \rho^*)$ (Definition \ref{def:program}) described by unitary
$U^*$ and quantum state $\rho^*$, along with the certificates
$(\cert_1, \ldots, \cert_q)$. If all the certificates are valid, then
the challenger of the security game tests whether $\qP^*$ is
$\epsilon$-good wrt $(f, \qE, t)$ or not (Definition \ref{def:good}),
where $\epsilon$ is a parameter of the experiment. Intuitively, this
test performs a measurement on the adversary, and if the measurement
accepts, it is guaranteed that the residual state of $\qP^*$ can
evaluate $f$ (as defined by $\qE$) with probability greater than $t +
\epsilon$. If the certificates are all valid and this test also
passes, $\qA$ is said to win the game. The security requirement is
that $\qA$ wins with at-most $\negl(\secp)$ probability for every
parameter $\epsilon$ such that $\epsilon = 1/\poly(\secp)$. We 
call this notion standard key leasing attack (standard-KLA)
security (Definition \ref{def:std-kla}).

The reason we introduced the
``$\epsilon$-good test'' is because estimating the success
probability of quantum programs is tricky (See \cite{TCC:Zhandry20} for a
detailed discussion). At a high level, the issues stem from the fact
that a quantum adversary may be in a superposition of ``successful'' and
``unsuccessful'' adversaries, and measuring the adversary in an ad-hoc
way may render it useless. However, these issues
were resolved in the work of Zhandry \cite{TCC:Zhandry20} by utilizing
a measurement procedure called projective implementation (Definition
\ref{def:pi}). Hence, we abstract
the details of performing such a measurement as part of our
$\epsilon$-good test (Definition \ref{def:good}). In the next subsection, we discuss the challenges in achieving this
security notion.

\subsection{Challenges in Achieving Collusion-Resistance}\label{sec:chall}

The difficulty in obtaining collusion-resistance
arises from the fact that an adversary can try to correlate all its
leased keys, before ``deleting'' any of them.
Sometimes the adversary can even learn a classical description of the
function $f$ without disturbing the states noticeably, due
to the gentle measurement lemma. In some cases, one
can circumvent the problem by making sure that all the
leased keys are uncorrelated, even if they provide a common
functionality. For instance, in the case of bounded
collusion-resistant PKE-SKL, one can use multiple instances of a
single-key PKE-SKL scheme to generate public-keys $\{\pk_i\}_{i\in[q]}$
and corresponding leased decryption keys $\{\qsk_i\}_{i\in[q]}$.
Then, the ciphertexts can include ciphertexts $\{\ct_i\}_{i\in[q]}$
under each of the public-keys $\{\pk_i\}_{i\in[q]}$ so that
each of the leased keys can decrypt. However, such an approach does
not seem to generalize to other primitives. For instance, consider the case of SKL for PRFs,
where each of the $q$ leased keys $\{\qsk_i\}_{i\in[q]}$ must 
evaluate a common PRF $f(\cdot) =
F_k(\cdot)$. Clearly, the leased keys must be correlated. Still, it could
be possible that the adversary cannot exploit these correlations.
Despite this hope, we find that it is not clear how to extend previous works on
PRF-SKL \cite{TCC:AnaPorVai23,EC:KitMorYam25}, even to the
setting of bounded-collusions. We now discuss some of these issues:

\paragraph{Constructions based on BB84 States.}

The work of Kitagawa, Morimae and Yamakawa \cite{EC:KitMorYam25}
constructed SKL schemes for PKE, PRFs and signatures. Their approach
is modular and makes use of a certified deletion property of BB84
states~\cite{C:BarKhu23}. Using this approach, they
built PRF-SKL from a primitive called two-equivocal PRFs (TEPRFs), which are
known from OWFs \cite{C:HJOSW16}. This approach requires first
sampling a BB84 state $\ket{x}_\theta$ where $x, \theta \la
\bit^\ell$. Recall that a BB84 state $\ket{x}_\theta$ is the state
$\big(H^{\theta_1} \otimes \ldots \otimes H^{\theta_\ell}\big) \big(\ket{x[1]}
\otimes \ldots \otimes \ket{x[\ell]}\big)$, where $H$ denotes the
Hadamard transform. Then, for each $i \in
[\ell]$, they compute $\rho_i$ as:

$$ \rho_i \seteq
\begin{cases}
        \ket{x[i]}\ket{\sk_{i,x[i]}} & \textrm{if}~~ \theta[i]=0\\
        \frac{1}{\sqrt2}\Big(\ket{0}\ket{\sk_{i,0}}+(-1)^{x[i]}\ket{1}\ket{\sk_{i,1}}\Big) & \textrm{if}~~ \theta[i]=1,
\end{cases}
$$

Here, $\sk_{i,0}$ and $\sk_{i,1}$ are correlated secret keys of a
TEPRF.
The overall leased key $\qsk$ is computed as $\qsk \seteq
(\rho_i)_{i\in[\ell]}$.
We leave out the details of what the keys exactly are, and what a TEPRF is.
The crucial point is that the construction
exploits a certain security property of TEPRF (called differing point
hiding) which is only guaranteed if one of the keys $\sk_{i,0}$ and
$\sk_{i,1}$ is hidden from the adversary. Consequently, the work
invokes this security guarantee only for the computational basis
positions ($i : \theta[i] = 0$) where the adversary receives
information of only one
of the two TEPRF keys. This security guarantee enables them to extract
the values associated with the computational basis positions
(given $\theta$),
from an adversary that is able to evaluate the PRF. On the other hand,
the deletion certificate requires the adversary to measure all the
qubits in the Hadamard basis, from which one can extract the values
of the Hadamard basis positions ($i: \theta[i]=1$). 
Hence, they are able to reduce to the certified deletion
property of BB84 states \cite{C:BarKhu23}. This property ensures that
if the adversary produces correct values wrt the Hadamard basis positions, then
even if $\theta$ is later revealed, it cannot output all the
values corresponding to the computational basis positions.

Observe now that the PRF is described by the TEPRF keys
$\{\sk_{i,j}\}_{i\in[\ell], j \in \bit}$. However, it is not clear how
to use the same TEPRF keys with multiple leased keys, to ensure the
different leased keys can evaluate the same PRF. Even if independently
sampled BB84 states are used each time, for an adversary with
polynomially many leased keys, it is extremely unlikely that there
exists $i \in [\ell]$ such that only one of $\sk_{i,0}, \sk_{i,1}$ was
obtained. This means we cannot invoke the security of TEPRF to hope to
extract the computational basis values of one of the BB84 states.
Notice that restricting one of the positions $i \in [\ell]$ to be
$\ket{0}$ across all the BB84 states (likewise, $\ket{1}$) does not help.
This is because we cannot hope that finding the value at this position
is hard and then reduce to it, as a simple gentle measurement attack
will reveal it. We remark that their DS-SKL scheme
follows a similar template, and hence runs into the same
issue.

\paragraph{Constructions based on Guassian Superpositions.}

In the work of Ananth, Poremba and Vaikuntanathan
\cite{TCC:AnaPorVai23} and the followup work of Ananth, Hu and Huang
\cite{TCC:AnaHuHua24}, Gaussian superposition states \cite{ITCS:Poremba23} were
utilized to construct SKL schemes based on the LWE and SIS
assumptions. In more detail, for parameters $n, m, q$, a matrix
$\mathbf{A} \in \mathbb{Z}_q^{n\times m}$, and a vector $\mathbf{y}
\in \mathbb{Z}_q^n$, the following Gaussian superposition state is
considered where $\rho_{\sigma}(\mathbf{x}) = \text{exp}(-\pi
\norm{\mathbf{x}}^2/\sigma^2)$:
$$\ket{\psi_{\mathbf{y}}} \seteq \sum_{\substack{x \in \mathbb{Z}_q^m
:\; \mathbf{A}\mathbf{x} =
\mathbf{y}\;(\text{mod}\;q)}}\rho_\sigma(\mathbf{x})\ket{\mathbf{x}}$$
In other words, $\ket{\psi_\mathbf{y}}$ consists of a superposition of
short vectors mapping $\mathbf{A}$ to $\mathbf{y}$. The vector
$\mathbf{y}$ is part of the PRF key $k$ and the state
$\ket{\psi_\mathbf{y}}$ is part of the leased secret key. In their
reduction to LWE, the reduction must be able to produce a Gaussian
vector $\mathbf{x}_0$ such that $\mathbf{A}\cdot\mathbf{x}_0 =
\mathbf{y} \;(\text{mod}\;q)$ along with an auxiliary input
$\mathsf{AUX}$ (to feed to the SKL adversary) that depends on 
$\ket{\psi_{\mathbf{y}}}$. Note that this reduction does not
have access to a trapdoor of $\mathbf{A}$. Hence, it cannot obtain
some other $\mathbf{x}_1 \neq \mathbf{x}_0$ such that $\mathbf{A}\cdot
\mathbf{x}_1 = \mathbf{y} \;(\text{mod}\;q)$, without breaking SIS. To
circumvent this, the works make use of a
Gaussian collapsing property due to Poremba \cite{ITCS:Poremba23}.
The property intuitively ensures that the collapsed form of
$\ket{\psi_{\mathbf{y}}}$ is indistinguishable from the state itself, and hence 
$\ket{\mathbf{x}_0}$ can be used to compute
$\mathsf{AUX}$, without the need for $\ket{\psi_{\mathbf{y}}}$. Observe now that if the
adversary obtains two different leased keys that are correlated with
$\mathbf{y}$, the proof breaks down. If one were to collapse two
such $\ket{\psi_{\mathbf{y}}}$, one would obtain different
$\ket{\mathbf{x}_0}, \ket{\mathbf{x}_1}$. Hence, it is unclear how a single
$\ket{\mathbf{x}_0}$ can be used to simulate $\mathsf{AUX}$ in this
case, which results from adversarial computation performed on both the leased
keys.

\paragraph{Unbounded Collusion-Resistant PKE-SKL.}

In the case of PKE, the work of Kitagawa, Nishimaki and Pappu
\cite{C:KitNisPap25} showed a scheme based on LWE in the setting
of unbounded collusions, which is non-trivial unlike the bounded
setting. At a high level, they switch the challenge
ciphertext distribution in their proof such that an adversary that
deletes its leased keys cannot distinguish the switch. Then, it is
argued using the security of ABE that the
adversary cannot decrypt ciphertexts sampled from this altered
distribution. Unfortunately, such techniques do not seem applicable to
primitives other than encryption. 
As a result, new techniques are needed to provide provable guarantees
in the setting of bounded collusion-resistant PRF-SKL, and
collusion-resistant SKL in general. Starting from the next
subsection, we will discuss the details of our new approaches.

\subsection{Collusion-Resistance of Two-Superposition States}

Recall that the work of Kitagawa, Morimae and Yamakawa
\cite{EC:KitMorYam25} showed
a framework for SKL based on BB84 states. In
essence, they generate secret keys in superposition
of a BB84 state, and finally reduce to the certified-deletion
property of the BB84 state. In this work, we will construct SKL schemes
by generating secret keys in superposition of a random
two-superposition state, which is well-suited to the
collusion setting. Then, we reduce to a new certified-deletion
property of such a state, which we show by generalizing existing
results. We now proceed to give an
overview about known guarantees of these states, and the ones we
require. A random two-superposition state refers to a state of the
following form:

$$\sigma \seteq \frac{1}{\sqrt2}\big(\ket{v} + (-1)^b\ket{w}\big)$$

where $v, w \la \bit^\secp$ and $b \la \bit$.  Such states were
utilized in prior works on publicly-verifiable deletion
\cite{TCC:BKMPW23} and revocable cryptography
\cite{TQC:MorPorYam24}. Based on a theorem of Bartusek et al.
\cite{TCC:BKMPW23}, the work of Morimae et al.
\cite{TQC:MorPorYam24} showed that a QPT adversary given $\sigma$
and $f(v), f(w)$ for an OWF $f$, cannot produce both of the following
simultaneously with probability greater than $1/2 + \negl(\secp)$: 

\begin{itemize}
\item A pre-image $\msg$ such that $f(\msg) \in \{f(v), f(w)\}$
\item A value $d$ such that $d \cdot (v \xor w) = b$.
\end{itemize}

Due to our focus on the collusion setting, we need to consider an
adversary that obtains $q = \poly(\secp)$ many i.i.d states $\sigma_1,
\ldots, \sigma_q$ where for each $i \in [q]$, $\sigma_i \seteq
\frac{1}{\sqrt2}\big(\ket{v_i} + (-1)^{b_i}\ket{w_i}\big)$. For
each $i \in [q]$, let $Q_i \seteq \{v_i, w_i\}$.
Intuitively, each $\sigma_i$ will be part of a leased secret key in
our SKL construction. We leave the details and intuition of the SKL
construction to Section \ref{sec:leverage-tracing}. Here, we mention that from a
successful SKL pirate program $\qP^*$, we will be able to extract a value
$\msg$ such that $\msg \in \bigcup_{i \in [q]} Q_i$. Recall that the
corresponding QPT adversary $\qA$ also outputs certificates $(\cert_1,
\ldots, \cert_q)$. From these certificates, we can hope to extract
$d_1, \ldots, d_q$ such that for each $i \in [q]$, $d_i \cdot (v_i
\xor w_i) = b_i$. Then, if we can show that it is difficult for an
adversary to produce such values $\msg$ and $d_1, \ldots, d_q$, we can
meaningfully reduce the security of SKL to this property.

In actuality, we will not be able to extract such values $d_1, \ldots,
d_q$ from the certificates. Hence, we place a stronger requirement
that the adversary can output arbitrary functions $g_1, \ldots, g_q$.
Note that we do not provide the adversary with OWF evaluations
of the values $\{v_i, w_i\}_{i\in[q]}$, although our approach should
easily generalize to this case. Consequently, we prove that no
\emph{unbounded} adversary can output both of these simultaneously
with probability greater than $1/2 + \negl(\secp)$:

\begin{itemize}
\item A pre-image $\msg$ such that $\msg \in \bigcup_{i\in[q]}Q_i$.
\item Functions $g_1, \ldots, g_q$ such that for each $i \in [q]$, it
holds that $g_i(v_i, w_i) = b_i$.
\end{itemize}

We are able to prove this along the lines of the proof of the main
theorem of Bartusek et al. \cite{TCC:BKMPW23}. We also consider an important case
where the values $\{v_i, w_i\}_{i \in [q]}$ are drawn from a domain
$[N]$ which may only be polynomially large. Specifically, we show (by
the same proof) that for each $q, t \in \mathbb{N}$, there exists $N =
O(q^2t^2)$ such that an
adversary cannot succeed in the above task with probability greater
than $1/2 + 1/t$. This allows us to utilize a wider range
of traitor tracing schemes (such as ones known for collusion-resistant
PRFs), which we discuss in Section \ref{sec:tracing-quantum}. Note that the
$1/t$ distinguishing advantage for $t = \poly(\secp)$ is not a problem. This is
because we anyway have to consider a parallel repetition of this game,
to reduce the success probability from around $1/2$ to $\negl(\secp)$.
We are able to prove this parallel-repetition variant in Theorem
\ref{thm:two-sup-par}, by utilizing a general quantum
parallel-repetition result 
from prior work \cite{STOC:BQSY24}.

\subsection{Leveraging Traitor Tracing}\label{sec:leverage-tracing}

In the parallel repetition version of the game from the previous subsection, we
consider an adversary that receives the following states
$\{\sigma_i^j\}_{(i,j)\in[\ell]\times[q]}$. Here $\ell =
\poly(\secp)$ is the number of parallel repetitions, $q$ is the
number of states obtained in each repetition, and $b(i,j)$ is a placeholder
for $b_i^j$:

$$\sigma_i^j \seteq \frac{1}{\sqrt2}\Big(\ket{v_i^j} +
(-1)^{b(i,j)}\ket{w_i^j}\Big)$$

For each $(i, j) \in [\ell] \times [q]$, let $Q_i^j \seteq \{v_i^j,
w_i^j\}$. Also, for each $i \in [\ell]$, let $Q_i \seteq
\bigcup_{j\in[q]} Q_i^j$. We have the guarantee that no QPT adversary
can output both of the following, except with probability
$\negl(\secp)$:

\begin{itemize}
\item Values $(\msg_1, \ldots, \msg_\ell) \in Q_1 \times \ldots \times Q_\ell$.
\item Functions $\{g_i^j\}_{(i,j) \in [\ell] \times [q]}$ such that
$g_i^j(v_i^j, w_i^j) = b_i^j$ for each $(i, j) \in [\ell] \times [q]$.
\end{itemize}

Let $h$ denote some efficiently computable and deterministic
key-generation algorithm. We will get into the specifics of $h$
shortly. Now, for each $(i, j) \in [\ell] \times [q]$, consider secret
keys $\sk_{i,v}^j \seteq h(i, v_i^j)$ and $\sk_{i,w}^j \seteq h(i,
w_i^j)$.  At a high level, for each $j \in [q]$, the leased key of our
SKL scheme will be of the form $\qsk^j \seteq
\{\rho_i^j\}_{i\in[\ell]}$, where each $\rho_i^j$ is a state of the
following form, where $b(i,j)$ is a placeholder for $b_i^j$:

$$\rho_i^j \seteq \frac{1}{\sqrt2}\big(\ket{v_i^j}\ket{\sk_{i,v}^j} +
(-1)^{b(i,j)} \ket{w_i^j}\ket{\sk_{i,w}^j}\big)$$

For each key $\qsk^j$, the deletion algorithm of our SKL scheme
requires the adversary to output a certificate with Hadamard basis
measurements $\cert^j \seteq (c_i^j, d_i^j)_{i \in [\ell]}$. Note that
the measurement $c_i^j$ corresponds to the register depicted with values
$v_i^j/w_i^j$, and $d_i^j$ to the register with values
$\sk_{i,v}^j/\sk_{i,w}^j$.
Hence, the verification algorithm checks whether $b_i^j = c_i^j \cdot
(v_i^j \xor w_i^j) \xor d_i^j \cdot \big(h(i, v_i^j) \xor h(i,
w_i^j)\big)$ holds for each $i \in [\ell]$. Assume for now that the
secret keys $\{\sk_{i,v}^j, \sk_{i,w}^j\}_{i\in[\ell]}$ corresponding
to the leased key $\qsk^j$ are sufficient to evaluate the required
functionality, and that this can be done without disturbing $\qsk^j$.

Our goal now is to reduce the security of SKL to the aforementioned
guarantee of two-superposition states. Consider such a reduction
$\qR$ that obtains the
states $\{\sigma_{i}^j\}_{(i,j) \in [\ell] \times [q]}$, samples the
function $h$ appropriately, and computes the states
$\{\rho_{i}^j\}_{(i,j) \in [\ell] \times [q]}$. Thereby, it can
simulate the view of the SKL adversary $\qA$ in a straightforward way.
Then, $\qA$ outputs certificates $(\cert^1, \ldots, \cert^q)$ and a
quantum pirate program $\qP^*$. Observe that based on the
certificates, $\qR$ can prepare the functions $g_i^j(x, y) \seteq
c_i^j \cdot (x \xor y) \xor d_i^j \cdot (h(i, x) \xor h(i, y))$ and
send them to the challenger. Now, we would like $\qR$ to be able to
obtain values $(\msg_1, \ldots, \msg_\ell) \in Q_1 \times \ldots
\times Q_\ell$ from the pirate program $\qP^*$ to complete the
reduction.

For this purpose, we will leverage the powerful guarantee
of a notion called multi-level traitor tracing (MLTT) (Section
\ref{sec:abs}). This is a
generalization of quantum-secure collusion-resistant traitor tracing.
In more detail, an MLTT scheme for an application
$(\qF, \qE, t)$ consists of algorithms $(\Setup, \KG,
\Eval, \qTrace)$. The algorithm $\Setup$ takes as input 
$N, \ell$ denoting the identity-space size and number of ``levels''
respectively. It outputs $(\msk, f, \aux_f)$ where
$\msk$ is a master secret-key, $f$ is a function and $\aux_f$ is some
public information. The algorithm $\KG$
takes as input the key $\msk$, a `level' $i \in [\ell]$ and an identity $\id$ and
produces a secret key $\sk_i$. The algorithm $\Eval$ takes
secret keys $\sk_1, \ldots, \sk_\ell$ (one for each level) and an
input $x$ and produces output $y$. The evaluation correctness guarantee
(specified by $\qF$) requires that $\Eval$ is consistent with $f$.
Additionally, we require a \emph{deterministic evaluation} property.
Intuitively, this requires that with overwhelming probability over
the choice of $(\msk, f, \aux_f)$ and $x \la \cX_f$ ($\cX_f$ is an
application-specific distribution specified by $\qF$), there
is a fixed $y$ such that $\Eval(\sk_1, \ldots, \sk_\ell, x) = y$ with
overwhelming probability, regardless of which identities were used to
derive $\sk_1, \ldots, \sk_\ell$. The interesting notion is that of
traceability, which we describe next:

Consider an adversary $\qA$ that specifies $q \in [N-1]$ and receives
values $(\id_i^j, \sk_i^j)_{(i,j) \in [\ell] \times [q]}$, where each
$\id_i^j$ is sampled as $\id_i^j \la [N]$, and each $\sk_i^j$ is
computed as $\sk_i^j \seteq \KG(\msk, i, \id_i^j)$. Let $Q_i' \seteq
\{\id_i^j\}_{j\in[q]}$ for each $i \in [\ell]$. Then, if $\qA$ outputs a
quantum program $\qP^*$ that is $\epsilon$-good wrt $(f, \qE, t)$, the
quantum tracing algorithm $\qTrace(\msk, \qP^*, \epsilon)$ outputs
values $(\id_1^*, \ldots, \id_\ell^*) \in Q'_1 \times \ldots \times
Q'_\ell$. Observe that with the help of this primitive, the
aforementioned SKL reduction $\qR$ is straightforward. It uses the
function $h$ defined as $h(i, \cdot) \seteq \KG(\msk, i,
\cdot)$\footnote{We assume $\KG$ is deterministic. This is wlog,
assuming post-quantum secure PRFs.} and runs $\qTrace$ to obtain
values $(\msg_1, \ldots, \msg_\ell) \in Q_1 \times \ldots \times Q_\ell$.
If this condition does not hold, we can show a reduction that breaks the
traceability of MLTT. Also notice that the correctness and
deterministic evaluation properties of
MLTT ensure that the leased keys $\qsk^j$ allow to evaluate $f$ using
$\Eval$ in superposition. Importantly, the gentle measurement
lemma ensures that this can be done without disturbing the quantum
state. This gives us an SKL scheme for the same application $(\qF, \qE, t)$.

Our next goal is to construct an MLTT scheme for PRFs. For this, we
rely on a construction similar to the classical collusion-resistant traceable PRF of
Maitra and Wu. \cite{PKC:MaiWu22}. Even though the upgrade from
standard traitor tracing to MLTT is straightforward, the upgrade from classical
tracing security to its quantum counterpart introduces 
challenges. We discuss these issues in the next subsection.

\subsection{Tracing Quantum Adversaries}\label{sec:tracing-quantum}

To begin with, we describe the structure of the traceable PRF by
Maitra and Wu \cite{PKC:MaiWu22}. The scheme utilizes two primitives:
a fingerprinting code (FC) and a traceable PRF with identity space $\bit$
(TPRF). The former is
an information-theoretic notion, the details of which we omit in
this overview. The latter notion of TPRF
consists of algorithms $(\Setup, \KG, \Eval, \tTrace)$. Here,
$\Setup$ outputs a master secret-key $\msk$ and
$\KG(\msk,\id)$ outputs a secret-key $\sk_{\id}$ for any $\id \in
\bit$.  Evaluation correctness requires that for $x$ sampled uniformly, $\Eval(\msk, x) = \Eval(\sk_{\id}, x)$ with overwhelming
probability.  The pseudo-randomness guarantee is straightforward and
requires that query access to $\Eval(\msk, \cdot)$ is
indistinguishable from query access to a truly random function.  The
other security guarantee called traceability ensures that a PPT
adversary that receives $\sk_{\id} = \KG(\msk, \id)$ for some $\id \in
\bit$ cannot produce a successful weak-pseudorandomness distinguisher $D$\footnote{The
definition considers $D$ that can distinguish the PRF from a random
function when given oracle access to either on uniform inputs.
This captures that the distinguisher
has some ability to evaluate the PRF, and doesn't simply hold 
hard-coded PRF evaluations. Similar definitions were utilized in prior
traitor tracing works \cite{EC:KitNis22,AC:GKWW21}.}
such that $\tTrace(\msk, D) \neq \id$. 
Consider now the collusion-resistant traceable PRF
$(\Setup',\KG',\Eval',\tTrace')$, which has the same syntax as the
TPRF, except that $\KG'$ admits identities in some larger identity
space $[N]$. Moreover, the security guarantee is stronger: the PPT
adversary gets to query $\KG'(\msk', \cdot)$ arbitrarily. Let $S$ be
the set of identities queried. Then, it cannot produce a successful
distinguisher $D$ such that $\tTrace'(\msk', D) \notin S$.

At a high level, the algorithm $\KG'(\msk', \id)$ for $\id \in [N]$
outputs $\sk'_{\id} \seteq \{\sk^i\}_{i\in[\ell]}$ such that $\sk^i \la
\KG(\msk^i, w_\id[i])$ where $\msk^i$ corresponds to an independent
TPRF instance, and $w_{\id}$ is a codeword of length $\ell$ that the 
FC scheme maps $\id$ to. Recall that $S$ denotes the
identity queries made by the adversary in the traceability security
game. The important point is that to guarantee collusion-resistance,
for every successful PPT distinguisher $D$, the algorithm
$\tTrace'(\msk', D)$ computes (internally) a string $w^*$ satisfying
the following:

For any $i \in [\ell]$, if there exists $b \in \bit$ such that for each
$\id \in S$, it holds that $w_{\id}[i] = b$, then $w^*[i] = b$. 

To ensure this property, the PRF is defined as
$\Eval(\msk', x) = \bigoplus_{i\in[\ell]}\Eval\allowbreak(\msk^i, x)$.
The idea is that given a distinguisher $D$,
the tracing algorithm can construct distinguishers
$\{D_i\}_{i\in[\ell]}$ corresponding to each of the $\ell$ TPRF
instances. Then, $\tTrace'$ will run $\tTrace(\msk^i, D_i)$ for each
$i \in [\ell]$ to compute $w^*[i]$. The above property of $w^*$ is then
satisfied by the traceability of TPRF.

In the quantum setting however, the above algorithm $\tTrace'$ doesn't
work. The problem is that in-order to construct the distinguishers
$\{D_i\}_{i \in [\ell]}$, we implicitly rely on the fact that multiple
copies of $D$ can be made. Observe that sequentially making use of a
quantum distinguisher $\qD$ to construct $w^*[1]$ followed by $w^*[2]$
and so on is insufficient. This is because even if the TPRF supports
quantum-secure tracing (via a quantum algorithm $\qTrace$), and we
construct $\qD_1$ from a quantum program $\qD$ and execute
$w^*[1] \la \qTrace(\msk^1, \qD_1)$, this may destroy $\qD$. Since
$\qD$ may no longer be useable,
$\qTrace(\msk^2, \qD_2)$ may not produce the desired outcome.

To overcome this, we first make use of a quantum-secure
TPRF due to Kitagawa and Nishimaki \cite{EC:KitNis22}\footnote{They
consider the setting of watermarking, but we show their scheme implies
TPRFs.}. \begin{comment}Note that we place an
additional requirement on the TPRF called Projective Tracing
(Definition \ref{}) which their construction satisfies. This property
will be useful because it is compatible with our next step.
\end{comment}
Now, our main
idea is to rely on a rewinding technique due to Chiesa et al.
\cite{FOCS:CMSZ21}, as
utilized by Zhandry \cite{C:Zhandry23} for quantum-secure tracing.  Intuitively,
the technique ensures that after estimating the pirate's success
probability on some distribution $\cD_1$, and then on another
distribution $\cD_2$, one can ``rewind'' the adversary. This rewinding
ensures that an estimation on $\cD_1$ right after the rewinding
produces a similar outcome as the first estimation wrt $\cD_1$.
Hence, our quantum-tracing algorithm $\qTrace'$ has the
following structure:

\begin{enumerate}
\item First, it estimates the success probability of $\qD$ on the honest
weak pseudo-randomness distribution. This utilizes an efficient
probability estimation procedure for quantum states from Theorem
\ref{thm:api}.

\item Then, it constructs $\qD_1$ using $\qD$ followed by running
$w^*[1] \la \qTrace(\msk^1, \qD_1)$.

\item Next, the aforementioned quantum rewinding procedure is applied. This ensures that $\qD$ continues to have
high success probability on the honest weak-pseudorandomness distribution.

\item The above steps are then repeated for $i = 2, \ldots, \ell$ to obtain
$w^*[2], \ldots, w^*[\ell]$.
\end{enumerate}

We are able to show that if this procedure does not work as expected,
we can break the security of the underlying TPRF.
Note that the traceable PRF of Maitra and Wu only admits a
polynomial size identity space $[N]$, a limitation which we inherit.
Consequently, our PRF-SKL scheme only satisfies bounded
collusion-resistance. As a result, the case of unbounded collusions
remains open for PRFs, both in the traitor-tracing/watermarking and
secure key leasing regimes.

%\begin{remark}
%We remark that in the work of Zhandry \cite{C:Zhandry23}, an elegant compiler was
%shown for achieving quantum-secure tracing while abstracting away all
%the quantum details. However, we make the quantum details explicit in
%this work. This is because the work requires additional formalism
%called `quantum hidden partition problems' which is well suited to
%the PLBE setting, but a bit complex in our opinion for the simple PRF
%application. Additionally, we need to introduce the quantum details
%anyway for our key-leasing definitions, so we might as well use them
%here to make things more transparent.
%\end{remark}

Until now, we only considered a ``single-level'' quantum-secure and
collusion-resistant tracing scheme for PRFs. However, the multi-level
variant (MLTT) can be achieved by the same approach. We simply 
use $k=\poly(\secp)$ many independent FC instances and $k\ell$
many TPRFs, and rely on a similar tracing algorithm. Note that
we consider a weaker traceability guarantee, where the adversary only
receives keys for randomly chosen identities. In summary, using the
fact that the TPRF of Kitagawa and Nishimaki \cite{EC:KitNis22} is
known from LWE with sub-exponential modulus, we are able to obtain our MLTT scheme for PRFs from
LWE. Together with our SKL compiler, this gives us a bounded
collusion-resistant PRF-SKL scheme from LWE with sub-exponential
modulus.

\subsection{Unbounded Collusion-Resistant Signatures}

We now discuss our compiler that upgrades a single-key secure DS-SKL
scheme into one with unbounded collusion-resistance. Recall that a
DS-SKL scheme allows the lessee to sign messages. Once the lessee
revokes its key, it is guaranteed that it can no longer sign randomly
chosen messages, except with negligible probability. Let
$\widetilde{\DSSKL}$ be a single-key secure DS-SKL scheme. We
construct a DS-SKL scheme $\DSSKL$ using $\widetilde{\DSSKL}$ and a
post-quantum signature scheme $\DSig$. First, the $\Setup$
algorithm of $\DSSKL$ samples a signing and verification key pair
$(\sig.\sk, \sig.\vk)$ of $\DSig$, which will also be the signing and
verification keys of $\DSSKL$. Then, to generate a leased key, the
key-generation algorithm $\qKG(\msk)$ first samples $\widetilde{\svk}$
using the setup algorithm of $\widetilde{\DSSKL}$, where
$\widetilde{\svk}$ is a signature verification key. Then, it samples a
leased-key and verification-key pair
$(\widetilde{\qsk}, \widetilde{\vk})$ using the key-generation
algorithm of $\widetilde{\DSSKL}$. Note that $\widetilde{\qsk}$ allows
to sign messages that can be verified by $\widetilde{\svk}$. Then,
$\qKG$ computes $\sig.\sigma$, a signature of $\DSig$
for the message $\widetilde{\svk}$. The leased key is then set to be
$\qsk \seteq (\widetilde{\qsk}, \widetilde{\svk}, \sig.\sigma)$ and the
verification-key as $\vk \seteq \widetilde{\vk}$. To sign a message
$\msg$ using $\qsk$, one first computes $\widetilde{\sigma} \la
\widetilde{\qEval}(\widetilde{\qsk}, \msg)$ where $\widetilde{\qEval}$
is the evaluation algorithm of $\widetilde{\DSSKL}$. Then, 
$\sigma' \seteq (\widetilde{\sigma}, \widetilde{\svk}, \sig.\sigma)$
is output as the signature. To verify such a signature, one first checks
if $\sig.\sigma$ is a valid signature for the message
$\widetilde{\svk}$ wrt $\DSig$ and $\sig.\vk$. If so, one
checks if $\widetilde{\sigma}$ is a valid signature wrt
$\widetilde{\svk}$ for the actual message $\msg$.

Consider now a QPT adversary $\qA$ that receives unbounded
polynomially many leased keys and then revokes all of them. If it is
still able to sign a random message $\msg$, then it must produce some
valid signature $\sigma'_0 = (\widetilde{\sigma}_0, \widetilde{\svk}_0,
\sig.\sigma_0)$. Now, by the security of $\DSig$, $(\widetilde{\svk}_0,
\sig.\sigma_0)$ can only be a valid message-signature pair of $\DSig$
if the pair was received by $\qA$. This means
that if $\qA$ received $\{\widetilde{\svk}_i\}_{i\in[q]}$ for $q =
\poly(\secp)$ as part of its leased keys, $\widetilde{\svk}_0$ must
satisfy $\widetilde{\svk}_0  = \widetilde{\svk}_j$ for some $j \in
[q]$. Observe now that $\widetilde{\sigma}_0$ must also be a valid
signature wrt $\widetilde{\svk}_j$ for $\sigma_0'$ to be a valid
signature of $\DSSKL$. However, this breaks the security of the $j$-th
instance of the single-key secure DS-SKL scheme $\widetilde{\DSSKL}$.
We remark that an analogous construction provides unbounded
collusion-resistant copy protection for signatures, greatly
simplifying this task in comparison to prior works
\cite{TCC:LLQZ22,TCC:CakGoy24}.

\subsection{Verification Oracle Security}

Finally, we consider a stronger security model where the adversary is
provided with classical oracle access to the verification algorithm. Such a
notion was considered in the work of Kitagawa et al.
\cite{C:KitNisPap25}.
Specifically, the notion allows the adversary to make arbitrarily many
classical queries to the algorithms $\Vrfy(\vk_i, \cdot)$ for each $i
\in [q]$. Then, as long as the adversary generates an accept response
at-least once for each $i \in [q]$, it must lose the ability to
evaluate the leased function. One might think that a stateful verifier
easily achieves this notion, as it can penalize the adversary for
producing incorrect certificates. However, it is preferable to have a
verifier that is a stateless machine. Even if a verifier is stateful,
an adversary with even black-box access to the verifier might be able
to rewind it, for example with a hard reset. This allows the adversary
to learn the outcomes of verification, and provides it with multiple
attempts at breaking the scheme, without any penalty. We formalize
this security notion, called verification-oracle aided
key-leasing-attacks (VO-KLA) security in Definition \ref{def:vo-kla}.

Unfortunately, our aforementioned SKL compiler (Theorem \ref{thm:const}) is
completely broken in this model, in the setting of unbounded
collusions. Even in the setting of bounded collusions, the security
breaks down whenever the collusion-bound $q$ is more than the number
of parallel-repetitions $\ell$. This is not ideal, as it would require
large quantum leased keys. The attack is as follows. Recall
that the adversary receives states of the following form for each
$(i,j) \in [\ell] \times [q]$:

$$\rho_i^j \seteq \frac{1}{\sqrt2}\big(\ket{v_i^j}\ket{\sk_{i,v}^j} +
(-1)^{b(i,j)} \ket{w_i^j}\ket{\sk_{i,w}^j}\big)$$

Now, instead of producing a legitimate certificate for the first
leased key $\qsk^1 \seteq (\rho_i^1)_{i \in [\ell]}$, the adversary
constructs a malformed certificate $\widetilde{\cert}$ by measuring
the positions $i \in [\ell] \setminus \{1\}$ honestly to get values
$(c_i, d_i)$. For the position $i = 1$, the adversary uses random
values $(c_1, d_1)$ and measures the state in the computational basis
to learn one among $\{\sk_{1,v}^1, \sk_{1,w}^1\}$. Observe that the
adversary fails with probability $1/2$ in producing an accept,
but can easily succeed in subsequent tries by
altering $(c_1, d_1)$. Hence, it is able to produce an
accept wrt $\qsk^1$, while
retaining one among $\{\sk_{1,v}^1, \sk_{1,w}^1\}$. It can then repeat
this attack with $\ell$-many leased keys to obtain a classical
secret-key for each position $i \in [\ell]$. Then, from the
correctness of MLTT, it
can continue to evaluate the function $f$. 
We provide an elegant solution to this problem, which works for any
SKL scheme with classical revocation and standard-KLA security.
Particularly, we show the following:

\begin{theorem}[VO-Resilience (Informal)]\label{thm:vo-inf}
Let there be an SKL scheme for application $(\qF, \qE, t)$
with standard-KLA security and classical revocation. Then, there
is an SKL scheme for $(\qF, \qE, t)$ with VO-KLA
security and classical revocation, assuming OWFs.
\end{theorem}

Observe that our SKL scheme (Section \ref{sec:const}) satisfies classical
revocation, as the deletion certificates are
simply Hadamard basis measurements. Our DS-SKL scheme also satisfies
classical revocation, assuming the single-key secure DS-SKL scheme
does, as the scheme of \cite{EC:KitMorYam25}. As a result, we are able to
upgrade all our results to the stronger notion of VO-KLA
security. We now give an overview of the VO-KLA secure construction
$\SKL$. This uses a
standard-KLA secure scheme $\widetilde{\SKL}$, along with a
primitive called tokenized MAC (TMAC) \cite{BSS21}. TMAC is a uniquely
quantum primitive consisting of algorithms $\TMac.(\KG, \TG, \qSign,
\Vrfy)$. The key-generation algorithm $\TMac.\KG$ outputs a secret-key
$\sk$. The token-generation algorithm $\TMac.\TG(\sk)$
outputs a quantum ``token'' state $\qtk$. The quantum
signing algorithm $\qSign$ takes as input $\qtk$ and a message $\msg$
and produces a signature $\sigma$. The classical
verification algorithm $\TMac.\Vrfy(\sk,\sigma,\msg)$
outputs $\top$ or $\bot$. The correctness notion is straightforward.
The
important aspect is the security guarantee, which ensures that no
QPT adversary $\qA$ given $\qtk$ and classical oracle access to 
$\TMac.\Vrfy(\sk, \cdot, \cdot)$ can succeed in producing two pairs
$(\msg_0, \sigma_0), (\msg_1, \sigma_1)$ such that $\msg_0 \neq
\msg_1$ and $\TMac.\Vrfy(\sk, \sigma_0, \msg_0) = \TMac.\Vrfy(\sk, \sigma_1,
\msg_1) = \top$. In other words, $\qA$ can sign at most one message
with a single token $\qtk$. We leverage this security guarantee
called unforgeability as follows.

The setup algorithm of $\SKL$ is the same as that of
$\widetilde{\SKL}$. The algorithm $\SKL.\qKG(\msk)$ computes the
leased-key $\qsk$ as $\qsk \seteq (\widetilde{\qsk}, \qtk)$ and the
verification-key $\vk$ as $\vk \seteq (\widetilde{\vk}, \sk)$.
Here, the pair $(\widetilde{\qsk}, \widetilde{\vk})$ is generated using
$\widetilde{\qKG}(\msk)$ and $(\qtk, \sk)$ is a
token and secret-key pair of $\TMac$, which is generated
independently for each invocation of $\SKL.\qKG$. The evaluation algorithm
simply runs the evaluation of $\widetilde{\SKL}$. The deletion
algorithm first generates a certificate $\widetilde{\cert} \la
\widetilde{\qDel}(\widetilde{\qsk})$ and then signs this certificate
using $\qtk$ to get $\sigma$. The final certificate is set as $\cert
\seteq (\widetilde{\cert}, \sigma)$. Given $\cert =
(\widetilde{\cert}, \sigma)$ as input, the verification algorithm
$\Vrfy$ first checks if the signature is valid using
$\TMac.\Vrfy(\sk, \widetilde{\cert}, \sigma)$. If it is valid, then it
accepts if $\widetilde{\Vrfy}(\widetilde{\vk},
\widetilde{\cert})$ accepts, and rejects otherwise.
The main idea is that this scheme essentially renders multiple queries
of $\qA$ to be useless. This is because the unforgeability of $\TMac$
restricts $\qA$ to commit to one query, for which it provides a valid
$\TMac$ signature. We remark that tokenized MACs were previously utilized
to handle decryption queries in the context of quantum PKE \cite{C:KMNY24}.
Finally, Theorem \ref{thm:vo-inf} follows
from the fact that tokenized MACs are known from OWFs \cite{BSS21}.

% !TEX root = main.tex
\section{Preliminaries}\label{sec:prelim}

\subsection{Notation}

The security parameter is denoted by $\secp$, a polynomial in
$\secp$ by $\poly(\secp)$, and a negligible function by
$\negl(\secp)$. The notation $y \seteq z$ denotes that
variable $y$ is assigned to, or replaced with value $z$. We use
calligraphic font to denote mixed quantum states and quantum
algorithms (Eg. $\qstate{q}$ and $\qalgo{A}$). We
use sans-serif font to denote classical values and algorithms (Eg.
$\mathsf{\sk}$ and $\algo{B}$). For a finite set $S$ and a
distribution $D$, $x \la S$ denotes sampling uniformly randomly from
$S$ and $x \la D$ denotes sampling according to $D$. We denote
sampling an output $y$ by running quantum algorithm $\qA$ on input
$1^\secp$ as $y \la \qA(1^\secp)$. 

\subsection{Quantum Information}
A pure quantum state is a vector $\ket{\psi}$ in some Hilbert space
 $\cH$ with $\norm{\ket{\psi}} = 1$. A Hermitian operator is any
 operator $P$ satisfying $P^\dagger = P$.
 Let $S(\cH)$ denote the set of Hermitian operators on $\cH$. A density
 matrix $\qstate{q} \in S(\cH)$ is a positive semi-definite operator
 with $\Tr(\qstate{q}) = 1$. Density matrices represent a probabilistic
 mixture over pure states, i.e., a mixed state. A pure state
 $\ket{\psi}$ has density matrix $\ketbra{\psi}$. A Hilbert space $\cH$
 can be divided into registers $\cH \seteq \cH^{\qreg{R_1}} \otimes
 \cH^{\qreg{R_2}}$. We use the notation $\cX^{\qreg{R_1}}$ to denote
 that operator $\cX$ acts on register $\qreg{R_1}$. This is equivalent
 to applying the operation $\cX^{\qreg{R_1}} \otimes
 \mathbb{I}^{\qreg{R_2}}$ to $\cH^{\qreg{R_1}} \otimes
 \cH^{\qreg{R_2}}$. A unitary operation is a complex matrix $U$
 satisfying $UU^\dagger = U^\dagger U = \mathbb{I}$.  It transforms a
 pure state $\ket{\psi}$ into $U\ket{\psi}$ and a mixed state
 $\qstate{q}$ into $U \qstate{q} U^\dagger$. A projector $\Pi$ is a
 Hermitian operator satisfying $\Pi^2 = \Pi$. The trace distance
 between mixed states $\qstate{q_0}$ and $\qstate{q_1}$ is defined as
 $\TD(\qstate{q_0}, \qstate{q_1}) \seteq \frac12 \Tr(\sqrt{(\qstate{q_0} -
 \qstate{q_1})^{\dagger}(\qstate{q_0} - \qstate{q_1})})$.

In this work, we consider quantum algorithms which are quantum
circuits built from some universal gate set, and possibly consist of
an initial state $\ket{\psi}$. By QPT algorithms, we mean ones with
polynomial circuit size. Often, we refer to certain algorithms as
``quantum programs'', which are essentially quantum algorithms with
classical inputs and outputs. These are described as follows.

\begin{definition}[Quantum Programs (with classical inputs and
outputs) \cite{C:ALLZZ21}]\label{def:program} A quantum program with classical inputs and outputs $\qP =
(U, \qstate{q})$ consists of a unitary $U$ and a quantum state
$\qstate{q}$. Let $\{U_x\}_x$ denote a set of unitaries indexed by $x
\in \bit^*$. Evaluating $\qP$ on input $x$ corresponds to the
following:
\begin{itemize}
\item Compute $U_x$ by applying $U$ to the input state $\ket{x}$.
\item Apply $U_x$ to $\qstate{q}$, producing the state $U_x\qstate{q} U_x^\dagger$.
\item Measure the first register of the resulting state to obtain the output.
\end{itemize}
\end{definition}

\begin{definition}[Positive Operator-Valued Measure (POVM)]
A POVM is a set of Hermitian positive semi-definite matrices
$\cM = \{M_i\}_{i\in\cI}$
s.t. $\sum_{i\in\cI}M_i =
\mathbb{I}$ holds. Applying $\cM$ to a state $\qstate{q}$ produces
outcome $i \in \cI$ with probability $\Trace(\qstate{q}M_i)$. Let
$\cM(\ket{\psi})$ denote the distribution of the output of applying
$\cM$ to $\ket{\psi}$.
\end{definition}

\begin{definition}[Quantum Measurement] A quantum measurement is a set
of matrices $\cE = \{E_i\}_{i\in\cI}$ satisfying
$\sum_{i\in\cI}E^\dagger_iE_i = \mathbb{I}$. Applying $\cE$ to a state
$\qstate{q}$ produces outcome $i$ with probability
$p_i \seteq \Tr(\qstate{q}E_i^\dagger E_i)$ with the corresponding
post-measurement state being $E_i \qstate{q} E_i^\dagger / p_i$.
For any quantum state $\qstate{q}$, let $\cE(\qstate{q})$ denote the
distribution of the outcome of applying $\cE$ to $\qstate{q}$. For any
states $\qstate{q_0}$ and $\qstate{q_1}$, the statistical distance
between $\cE(\qstate{q_0})$ and $\cE(\qstate{q_1})$ is bounded above
by $\TD(\qstate{q_0}, \qstate{q_1})$.
\end{definition}

\begin{definition}[Projective Measurement/POVM]
A quantum measurement $\cE = \{E_i\}_{i\in\cI}$ is a projective
measurement if for each $i \in \cI$, $E_i$ is a projector. A binary
projective measurement is of the form $\cE = \{\Pi, \mathbb{I} -
\Pi\}$ where $\Pi$ corresponds to the outcome $0$ and $\mathbb{I} -
\Pi$ to the outcome $1$. Likewise, a POVM $\cM = \{M_i\}_{i\in\cI}$ is
projective if for each $i \in \cI$, $M_i$ is a projector.
\end{definition}

\begin{definition}[$(\epsilon, \delta)$-Almost Projective Measurement
\cite{TCC:Zhandry20}]
A quantum measurement $\cE$ is $(\epsilon,\delta)$-almost projective
if applying $\cE$ twice in a row to any quantum state $\qstate{q}$
produces outcomes $p, p'$ such that $\Pr[\lvert p - p'\rvert \ge
\epsilon] \le \delta$.
\end{definition}

\begin{definition}[Mixture of Projective Measurements
\cite{TCC:Zhandry20}]
For a set of projective measurements $\cE = \{\cE_i\}_{i \in \cI}$
where $\cE_i = (\Pi_i, \mathbb{I} - \Pi_i)$, consider the following
POVM $\cP_D$ corresponding to distribution $D$:
\begin{itemize}
\item Sample $i \la D$.
\item Apply $\cE_i$ and output the resulting bit.
\end{itemize}
We refer to $\cP_D$ as a mixture of projective measurements. It is
specified by matrices $(P_D, Q_D)$, where $P_D = \sum_{i\in\cI}\Pr[i
\la D]\Pi_i$ and $Q_D = \sum_{i\in\cI}\Pr[i \la D](\mathbb{I} -
\Pi_i)$.
\end{definition}

\begin{lemma}[Gentle Measurement \cite{Winter99}]
\label{lma:gentle} Let $\qstate{q}$ be a
quantum state and $(\Pi, \mathbb{I} - \Pi)$ be a binary projective
measurement such that $\Tr(\Pi\qstate{q}) \ge 1- \delta$. Let the
post-measurement state on applying
$(\Pi, \mathbb{I} - \Pi)$ and conditioning on the first outcome be:
$\qstate{q}' = \Pi \qstate{q} \Pi/\Trace(\Pi\qstate{q})$. Then, we
have that
$\TD(\qstate{q}, \qstate{q}') \le 2\sqrt{\delta}$.
\end{lemma}

\begin{definition}[Projective Implementation]\label{def:pi}
Consider the following:
\begin{itemize}
\item $\cP = \{M_i\}_{i\in \cI}:$ A POVM, where $\cI$ is some index set.
\item $\cD:$ A finite set of distributions over $\cI$.
\item $\cE = \{\E_{D}\}_{D \in \cD}:$ A projective measurement indexed
by distributions $D \in \cD$.
\end{itemize}
Now, consider the following procedure: (1) The projective measurement
$\cE$ is applied to obtain outcome $D$. (2) A sample $d$ is drawn from $D$ and then output.

If the above procedure is equivalent to applying the POVM $\cP$, then
$\cE$ is called the projective implementation of $\cP$, and is denoted
by $\ProjImp(\cP)$.
\end{definition}

\begin{theorem}[\cite{TCC:Zhandry20}, Lemma 1] Any binary outcome POVM $\cP =
(P, \mathbb{I} - P)$ has a unique projective implementation
$\ProjImp(\cP)$.
\end{theorem}

\begin{definition}[Shift Distance]
For distributions $D_0, D_1$, the $\epsilon$-shift distance 
$\Delta_\epsilon(D_0, D_1)$ is the
smallest value $\delta$ such that for all $x \in \mathbb{R}:$
$$\Pr[D_0 \le x] \le \Pr[D_1 \le x + \epsilon] + \delta
\;\;\;\;\;\;\;\;
\Pr[D_0 \ge x] \le \Pr[D_1 \ge x - \epsilon] + \delta$$
$$\Pr[D_1 \le x] \le \Pr[D_0 \le x + \epsilon] + \delta
\;\;\;\;\;\;\;\;
\Pr[D_1 \ge x] \le \Pr[D_0 \ge x - \epsilon] + \delta$$

For two quantum measurements $\cM, \cN$,
the shift distance between $\cM$ and $\cN$ with parameter $\epsilon$
is defined as
$\Delta_\epsilon(\cM, \cN) \seteq \sup_{\ket{\psi}}
\Delta_\epsilon\big(\cM(\ket{\psi}), \cN(\ket{\psi})\big)$.
\end{definition}

\begin{theorem}[Indistinguishability of Projective Implementation \cite{TCC:Zhandry20}]\label{thm:pi-ind}
Let $\qstate{q}$ be an efficiently constructible and possibly mixed
state, and $D_0, D_1$ be computationally indistinguishable
distributions. Then, for any inverse polynomial $\epsilon$ and any
function $\delta$, the following holds:
$$\Delta_{\epsilon}\big(\ProjImp(\cP_{D_0})(\qstate{q}),
\ProjImp(\cP_{D_1})(\qstate{q})\big) \le \delta$$
\end{theorem}

\begin{remark}
As noted in prior work \cite{EC:KitNis22}, for Theorem
\ref{thm:pi-ind} to hold, the indistinguishability of $D_0, D_1$
needs to hold for distinguishers that can efficiently construct
$\qstate{q}$.
\end{remark}

\begin{theorem}[Approximate Projective Implementation
    \cite{TCC:Zhandry20}]\label{thm:api}
Let $\cP = \{\cE_i\}_{i\in\cI}$ be a set of projective measurements
and $D$ be a distribution over $\cI$. For any $\epsilon, \delta \in
(0, 1)$, there exists an algorithm $\API_{\epsilon, \delta}^{\cP, D}$
with the following properties:
\begin{itemize}
    \item $\API_{\epsilon, \delta}^{\cP, D}$ is $(\epsilon,
        \delta)$-almost projective.
    \item $\Delta_{\epsilon}\big(\ProjImp(\cP_D), \API_{\epsilon,
        \delta}^{\cP, D}\big) \le \delta$.
    \item The expected runtime of $\API_{\epsilon, \delta}^{\cP, D}$
        is $T_{\cP, D} \cdot \poly(\epsilon^{-1}, \log(\delta^{-1}))$
        where $T_{\cP, D}$ is the combined runtime of sampling
        $i \la D$, mapping $i \ra \cE_i$, and applying $\cE_i$.
\end{itemize}
\end{theorem}

\begin{theorem}[State Repair \cite{FOCS:CMSZ21}]\label{thm:repair}
Let $\cE$ be a projective measurement on register $\cH$ with outcomes
in set $S$. Let $\cM$ be an $(\epsilon, \delta)$-almost projective
measurement on $\cH$. Consider parameters $T \in \mathbb{N}, s \in S,
p \in [0,1]$. Then, there exists a procedure
$\Repair_{T,p,s}^{\cM,\cE}$ which is to be applied on $\cH$ as follows:

\begin{itemize}
\item Apply $\cM$ to obtain $p \in [0,1]$.
\item Next, apply $\cE$ to obtain $s \in S$.
\item Then, apply $\Repair_{T,p,s}^{\cM,\cE}$.
\item Finally, apply $\cM$ again to obtain $p' \in [0,1]$.
\end{itemize}

Then, it is guaranteed that:

$$\Pr[\lvert p - p'\rvert > 2\epsilon] \le |S|\cdot\delta + |S|/T + 4\sqrt{\delta}$$

\end{theorem}

\subsection{Generic Cryptographic Primitives}

We now present the definitions required to generalize our SKL and
MLTT notions. We first define a quantum security predicate, 
which is essentially a binary projective measurement.

\begin{definition}[Quantum Security Predicate]\label{def:sec-pred}
A quantum security predicate $\qE$ gets as input a classical function
$f$, a quantum program $\qP$, and random tape $r$. It performs a binary
projective measurement on $\qP$ and outputs the result.
\end{definition}

\begin{definition}[Quantum Correctness Predicate]\label{def:cor-pred}
A quantum correctness predicate $\qF$ takes as input a quantum program
$\qg$, a classical function $f$, and a random tape $r$. It is a binary
projective measurement with the following structure:
\begin{description}
\item $\underline{\qF(\qg, f, r)}:$
\begin{itemize}
\item Sample $x \la \cX_f$ from an application specific distribution
$\cX_f$, based on random tape $r$.
\item Output $1$ if $\CorVrfy\big(\qg(x), f, x, r\big) = 1$, and $0$
otherwise, where $\CorVrfy$ is a deterministic classical algorithm that is
application specific.
\end{itemize}
\end{description}
\end{definition}

\begin{remark}
As a special case, $\qg$ could also be a classical function.
\end{remark}

Next, we define a cryptographic application as a tuple consisting of
predicates:

\begin{definition}[Cryptographic Application]\label{def:app}
A cryptographic application is a 3-tuple $(\qF, \qE, t)$, where
$\qF, \qE$ are quantum correctness and quantum security predicates
respectively, and $t \in [0, 1)$.
\end{definition}

% \fuyuki{We should explain this abstraction in more detail. It is hard to understand unless one is familiar with ALL+21. For example, how about describing what $(\qR,\qE,\gamma)$ are in PRF and signature applications?}
% \fuyuki{I agree with Ryo that the current correctness requirements for SKL and TT might look weird at first. I believe more explanations here are useful to understand the subsequent contents.}
% \nikhil{Done.}

The parameter $t$ indicates a probability threshold that is $0$ for applications
with search security (Eg. signatures), and $1/2$ for those with
decisional security (Eg. PKE, PRFs). Let us explain this formalism
for the case of PRFs and signatures in the context of SKL.
The application $(\qF, \qE, t)$ for PRFs is as follows:

\begin{description}
\item $\underline{\qF(\qg, f, r)}:$
\begin{itemize}
\item Sample $x \la \cX$ where $\cX$ is the domain of $f$, using
the random tape $r$.
\item If $f(x) = \qg(x)$, output $1$. Else, output $0$.
\end{itemize}

\item $\underline{\qE(f, \qP, r)}$:
\begin{itemize}
\item Sample $x \la \cX$ where $\cX$ is the domain of the PRF $f$,
using the random tape $r$. Sample $b \la \bit$ also using $r$.
\item If $b=0$, compute $y \seteq f(x)$. Otherwise, sample $y \la
\cY$ using $r$, where $\cY$ is the range of the PRF $f$.

\item Run $b' \la \qP(x, y)$.
\item Output $1$ if $b = b'$ and $0$ otherwise.
\end{itemize}

\item $\underline{t} \seteq \frac12$
\end{description}

Observe that the correctness predicate $\qF$ obtains a quantum program
$\qg$ and a classical function $f$ as inputs, along with a random tape
$r$. Since in the context of SKL, the lessee is provided with a leased
key $\qsk$, we want $\qg$ to capture a program which 
evaluates $\qEval(\qsk, \cdot)$ where $\qEval$ is the SKL evaluation
algorithm. We then want the second input $f$ to
capture the classical PRF associated with the leased key
$\qsk$. Note that $f$ is determined by the master 
secret-key $\msk$ sampled by the SKL $\Setup$ algorithm.
Observe that the predicate $\qF$ simply checks whether $\qg$ and $f$
produce the same output for a uniformly random input $x$, which is
deterministically chosen based on the random tape $r$ (which would be
chosen uniformly at random). For the correctness of SKL, we
want that the predicate $\qF(\qEval(\qsk, \cdot), f, r)$ outputs
$1$ with probability $1 - \negl(\secp)$ over the randomness $r$ (and the
randomness implicit in $\qEval$), which captures that $f$ and
$\qEval(\qsk, \cdot)$ have the same evaluations except on a
negligible fraction of inputs.

Consider now the security predicate $\qE$, which is given
a classical PRF $f$ along with a quantum pirate program
$\qP$. Intuitively, $\qE$ tests whether $\qP$ is able to
distinguish between samples drawn according to a random function, or
ones drawn according to the PRF $f$. The reason is that for the
security of SKL, we want that an adversary that provides valid
certificates for all its keys, should not be able to make $\qE$ output
$1$ with probability $t + 1/p(\secp) = 1/2 +
p(\secp)$ for some polynomial $p(\secp)$. To perform the
aforementioned distinguishing test, $\qE$ first samples a uniformly random PRF
input $x$ and a random bit $b$. Note that these are also
deterministically determined by the random tape $r$. Then, based on
$b$, it either chooses $y$ to be the PRF evaluation at $x$ or a truly
random value. Finally, it runs $\qP$ on input $(x, y)$ and outputs $1$
only if $\qP$ guesses $b$ correctly.

Although we discussed only the SKL context above, it is easy to see
that the same formalism captures MLTT for different applications. The
predicate $\qF$ enforces correctness wrt $f$ on the
classical $\Eval$ algorithm of $\MLTT$, while $\qE$ performs the same
test on a quantum pirate program of the MLTT game.

\begin{remark}
Note that the input $x$ is chosen
uniformly at random as in prior SKL works that construct PRF-SKL \cite{TCC:AnaPorVai23,EC:KitMorYam25}. The intuition is that this captures the inability of
a pirate to evaluate $f$. It is problematic if we allow $\qP$ to
choose $x$ by itself as in standard pseudo-randomness definitions.
This is because the SKL adversary $\qA$ gets $\qsk$, so it can
hard-code $\qP$ with input-output pairs of the PRF.
\end{remark}

We can also capture other primitives such as PKE and signatures.
For the case of signatures, $f$ is defined as $f
\seteq \Sign(\ssk, \cdot) \| \svk$ where $\Sign$ is the signing
algorithm and $\ssk, \svk$ are the signing
and verification keys of the signature scheme respectively. In the
context of DS-SKL, the lessee is given a key $\qsk$
that allows them to sign using $\qEval(\qsk, \cdot)$. We want $\qg$ to
capture $\qEval(\qsk, \cdot)$, so we define
$\qF(\qg, f, r)$ to obtain $\svk$ from $f$ and
check whether $\SigVrfy(\svk, \msg, \qg(\msg)) = 1$ for uniformly
chosen $\msg$ (chosen based on $r$), where $\SigVrfy$ is the signature
verification algorithm. The security
predicate $\qE$ is essentially identical to $\qF$, as it is
supposed to check whether a pirate $\qP$ is able to sign random
messages or not, and $t \seteq 0$.

In the above examples, we focussed on designing the security
predicate $\qE$ to perform a test on the program $\qP$.
However, simply running the test once is not useful, as we are
interested in estimating the success probability of $\qP$ in passing
this test. Recall that this cannot be achieved by simply running $\qP$
many times, due to the nature of quantum
states. Hence, we define the following procedure for this purpose,
which makes use of quantum tools introduced in prior work
\cite{TCC:Zhandry20}:

\begin{definition}[$\epsilon$-good test wrt $(f, \qE, t)$]
\label{def:good} A quantum program
$\qP = (U, \rho)$ is said to be $\epsilon$-good wrt $(f, \qE, t)$ if and
only if the following procedure outputs $1$:
\begin{itemize}
\item Sample $\key \la \bit^\secp$ for a quantum-accessible PRF $\QPRF$.
\item Let $\cP = \{\cE_r^\key\}_{r}$ be a set of projective measurements such
that each $\cE_{r}^\key$ corresponds to running $\qE(f, \qP, (r,
\key))$. Note
that we dilate the measurement using sufficient ancillas so that
it is projective.
\item Let $D$ denote the uniform distribution over $\bit^{\poly(\secp)}$.
\item Apply $p \la \API_{\epsilon', \delta}^{\cP, D}$ to $\qP^*$, where
$\delta = 2^{-\secp}$ and $\epsilon' = 0.1\epsilon$. If $p \ge t +
0.9\epsilon$, output $1$. Else, output $0$.
\end{itemize}
\end{definition}

\begin{remark}
One might prefer to define the $\epsilon$-good test using the
projective implementation $\ProjImp(\cP_D)$. We utilize the
$\API$ procedure instead, as it is efficient and makes some of our
proofs simpler. Importantly, the two variants are equivalent in the
context of the security of SKL (Definition \ref{def:std-kla}). This is
because by Theorem \ref{thm:api}, both measurements produce similar outcomes.
\end{remark}

\nikhil{I directly changed this test to use API, instead of only the
MLTT definition. This changes the SKL definition too but makes things
consistent. Also, this SKL notion is clearly equivalent to the one
where ProjImp is used.}

% !TEX root = main.tex
\section{Collusion-Resistant SKL}\label{sec:def-skl}

For the sake of our compiler, we will define secure key leasing in a
generic way. Different cryptographic applications such as encryption,
signatures, and pseudo-random functions can then be cast as special
instances. 

\begin{definition}[Generic SKL Scheme]\label{def:skl}
A generic SKL scheme for application $(\qF, \qE, t)$ is a tuple of
five algorithms $(\Setup, \qKG, \qEval, \qDel, \Vrfy)$, described as follows:

\begin{description}
\item $\Setup(1^\secp, q) \ra (\msk, f, \aux_f)$: The setup algorithm
takes a security parameter as input along with a collusion-bound $q$
(which can be $\bot$ in the unbounded case). It outputs a master secret-key
$\msk$ and a function $f \in \cF$, along with auxiliary information
$\aux_f$. 

\item $\qKG(\msk) \ra (\qsk, \vk):$ The key-generation algorithm
takes a master secret-key $\msk$ as input.
It outputs a quantum leased-key $\qsk$ and a corresponding
verification-key $\vk$.

\item $\qEval(\qsk, x) \ra y:$ The evaluation algorithm takes as input a 
leased-key $\qsk$ and an input $x$, and outputs a value $y$.

\item $\qDel(\qsk) \ra \cert:$ The deletion algorithm takes a
leased-key $\qsk$ as input and outputs a classical certificate
$\cert$.

\item $\Vrfy(\vk, \cert) \ra \top / \bot:$ The verification algorithm takes as input a
verification-key $\vk$ and a certificate $\cert$. It outputs $\top$
(accept) or $\bot$ (reject).
\end{description}

\paragraph{Evaluation Correctness:}

We require that the function $f$ and the $\qEval$ algorithm
are consistent with each other. This is captured by the quantum
predicate $\qF$ as follows. For all $q = \poly(\secp)$ and $q=\bot$,
the following holds:

\[
\Pr\left[
\qF\big(\qEval(\qsk, \cdot), f, r\big) \ra 0
 \ :
\begin{array}{rl}
 &(\msk, f, \aux_f) \la \Setup(1^\secp, q)\\
 & (\qsk, \vk) \la \qKG(\msk) \\
 &r \la \bit^{\poly(\secp)}
\end{array}
\right] \le \negl(\secp).
\]

Additionally, we require that the post-measurement state $\qsk'$
resulting on executing $\qEval(\qsk, \cdot)$ satisfies
$\TD(\qsk, \qsk') \le \negl(\secp)$.

\begin{remark}
The trace distance requirement ensures $\qsk$ can be re-used
polynomially many times. We needed to specify
it explicitly as $\qF$ may have been defined in a way that almost
always outputs $0$ even if $\qEval$ was far from being deterministic.
\end{remark}

\begin{remark}
Due to the focus on key-leasing security,
we do not capture application specific
security in this formalism. This includes
pseudo-randomness for PRFs, unforgeability for signatures etc, which
are to be specified separately.
\end{remark}

\paragraph{Verification Correctness:} For all $q = \poly(\secp)$ and
$q = \bot$, the following holds:

\[
\Pr\left[
\Vrfy(\vk, \cert) \ra \bot
 \ :
\begin{array}{rl}
 &(\msk, f, \aux_f)\la \Setup(1^\secp, q)\\
 & (\qsk, \vk) \la \qKG(\msk) \\
 & \cert \la \qDel(\qsk) \\
\end{array}
\right] \le \negl(\secp).
\]
\end{definition}

Next, we define verification-oracle aided key
leasing attack (VO-KLA) security:

\begin{definition}[VO-KLA Security]\label{def:vo-kla}
We formalize the experiment
$\expb{\SKL,\qA}{vo}{kla}(\allowbreak1^\secp, \qF, \qE, t, \epsilon)$ between an adversary
$\qA$ and a challenger $\qCh$ for an SKL scheme $\SKL$ corresponding
to application $(\qF, \qE, t)$:

\begin{description}
\item $\underline{\expb{\SKL,\qA}{vo}{kla}(1^\secp, \qF, \qE, t,
\epsilon)}:$
\begin{enumerate}
\item $\qCh$ samples $(\msk, f, \aux_f) \la \Setup(1^\secp, \bot)$ and
sends $\aux_f$ to $\qA$.

\item $\qA$ sends $q = \poly(\secp)$ to $\qCh$. 
\item For all $i \in [q]$, $\qCh$ samples $(\qsk_i, \vk_i) \la
\qKG(\msk)$. It sends $(\qsk_i)_{i\in[q]}$ to $\qA$.

\item For each $i\in[q]$, let $V_i \seteq \bot$. Throughout,
$\qA$ is given oracles access to:

\begin{description}
\item $\underline{\Oracle{\Vrfy}(i, \cert)}:$
\begin{itemize}
\item Compute $d \la \Vrfy(\vk_i, \cert)$.
\item If $V_i = \bot$, set $V_i
\seteq d$. Return $d$.
\end{itemize}
\end{description}

\item $\qA$ outputs a quantum program $\qP^* = (U, \rho)$ to $\qCh$.
\item If $V_i = \top$ for each $i \in [q]$ and $\qP^*$ is tested to be
$\epsilon$-good wrt $(f, \qE, t)$, then $\qCh$ outputs $\top$. Else, it
outputs $\bot$.
\end{enumerate}
\end{description}

We say that an SKL scheme $\SKL$ for application $(\qF, \qE, t)$ satisfies
VO-KLA security if the following
holds for every QPT $\qA$ and every $\epsilon = 1/\poly(\secp)$:

$$\Pr[\expb{\SKL,\qA}{vo}{kla}(1^\secp, \qF, \qE, t, \epsilon) \ra \top] \le \negl(\secp)$$
\end{definition}

Next, we present the analogous definition without the verification oracle:

\begin{definition}[Standard-KLA Security]\label{def:std-kla}
This notion is formalized by the experiment
$\expb{\SKL,\qA}{std}{kla}(1^\secp,\qF,\qE,t,\epsilon)$ between an adversary
$\qA$ and a challenger $\qCh$ for an SKL scheme corresponding to
$(\qF, \qE, t)$. The experiment is defined as:

\begin{description}
\item $\underline{\expb{\SKL,\qA}{std}{kla}(1^\secp, \qF, \qE, t,
    \epsilon)}:$
\begin{enumerate}
\item $\qCh$ and $\qA$ interact as in Steps 1-3. of
$\expb{\SKL,\qA}{vo}{kla}(1^\secp, \qF, \qE, t, \epsilon)$, with
$\qCh$ and $\qA$ obtaining $q=\poly(\secp)$ and $(\aux_f,
(\qsk_i)_{i\in[q]})$ respectively.
\item $\qA$ sends $(\cert_1, \ldots, \cert_q)$ and a quantum program $\qP^* = (U,
    \rho)$ to $\qCh$.
\item If for each
$i \in [q]$, it holds that $\Vrfy(\vk_i, \cert_i) = \top$ and
$\qP^*$ is
tested to be $\epsilon$-good wrt $(f, \qE, t)$, then $\qCh$ outputs $\top$. 
Else, it outputs $\bot$.
\end{enumerate}
\end{description}

We say an SKL scheme $\SKL$ for application $(\qF, \qE, t)$ satisfies
standard-KLA security if the following holds for every QPT $\qA$ and
every $\epsilon = 1/\poly(\secp)$:

$$\Pr[\expb{\SKL,\qA}{std}{kla}(1^\secp, \qF, \qE, t, \epsilon) \ra \top] \le \negl(\secp)$$
\end{definition}

Additionally, we say a scheme satisfies $q$-bounded
standard-KLA/VO-KLA
security if $q$ is a fixed polynomial in $\secp$ given as input to
$\Setup$,
instead of being provided by $\qA$ in the experiment. In this case, we denote
the experiments with additional input $q$ as
$\expb{\SKL,\qA}{std}{kla}(1^\secp, q, \qF, \qE, t, \epsilon)$/
$\expb{\SKL,\qA}{vo}{kla}(1^\secp, q, \qF, \qE, t, \epsilon)$.

% !TEX root = main.tex
\section{Multi-Level Traitor Tracing}\label{sec:abs}

We now define the notion of a multi-level traitor tracing
scheme. The definition also captures collusion-resistance and the
ability to trace quantum pirate programs. 

\begin{definition}[Generic Multi-Level Traitor Tracing Scheme]\label{def:mltt-cor}
A multi-level traitor tracing (MLTT) scheme $\MLTT$ for application
$(\qF, \qE, t)$ is a tuple of four algorithms $(\Setup, \allowbreak
\KG, \Eval, \qTrace)$.  The algorithms are described as follows:

\begin{description}
\item $\Setup(1^\secp, N, k) \ra (\msk, f, \aux_f):$ The setup algorithm
takes a security parameter, a parameter $N$ (identity space size)
and a parameter $k$ (number of ``levels'') as input. It outputs a
master secret-key $\msk$ and a function $f \in \cF$, along with
auxiliary information $\aux_f$.

\item $\KG(\msk, i, \id) \ra \sk_i:$ The key-generation
algorithm takes as input a master secret-key $\msk$, a ``level''
$i \in [k]$, and an identity $\id \in [N]$. It outputs a secret key
$\sk_i$. We assume that $\KG$ is a deterministic
algorithm\footnote{This is wlog, assuming post-quantum PRFs, as the
randomness can be derived from the input $\id$ using a PRF.}.

\item $\Eval(\sk_1, \ldots, \sk_k, x) \ra
y:$ The evaluation algorithm takes as input $k$ secret keys
$\sk_1, \ldots, \sk_k$ and an input $x$.  It
outputs a value $y$.

\item $\qTrace(\msk, \qP, \epsilon^*) \ra (\id_1, \ldots,
\id_k):$ The quantum tracing algorithm takes as input a master
secret-key $\msk$, a quantum program $\qP$, and a parameter
$\epsilon^*$ (a lower bound on the advantage of $\qP$). It outputs a $k$-tuple of identities $(\id_1, \ldots,
\id_k)$.
\end{description}

\paragraph{Evaluation Correctness:} For all $N =
\poly(\secp)$, $k = \poly(\secp)$ and $(\id_1, \ldots, \id_k) \in
[N]^k$, there exists $n(\secp) = \negl(\secp)$ such that:

\[
\Pr\left[
\qF\big(\Eval(\sk_1, \ldots, \sk_k, \cdot), f, r\big) \ra 0
 :
\begin{array}{rl}
 &(\msk, f, \aux_f) \la \Setup(1^\secp, N, k)\\
 &\forall i \in [k]: \sk_i \la \KG(\msk, i, \id_i)\\
 &r \la \bit^{\poly(\secp)}
\end{array}
\right]\le n(\secp)
\]

\paragraph{Deterministic Evaluation:}

Let $\cX_f$ denote the distribution that the correctness predicate $\qF$ samples inputs from. For all $N, k =
\poly(\secp)$, with
probability $1-\negl(\secp)$ over choice of $(\msk, f, \aux_f) \la
\Setup(1^\secp, N, k)$ and $x \la \cX_f$, there
exists some $y$ such that for all
$(\id_1, \ldots, \id_k) \in [N]^k$, the following holds:

\[
\Pr\left[
\Eval(\sk_1, \ldots, \sk_k, x) = y
 :
\begin{array}{rl}
 &\forall i \in [k]: \sk_i \la \KG(\msk, i, \id_i)
\end{array}
\right]\ge 1 - \negl(\secp)
\]

\begin{remark}
Note that for fixed $\sk_1, \ldots, \sk_k$, $\Eval$ can be
deterministic wrt input $x$ wlog, assuming post-quantum PRFs exist. We
require above that the outputs are consistent across different
sets of secret-keys.
\end{remark}

\begin{definition}[Traceability]\label{def:mltt-tr}This notion is formalized
by the experiment $\expb{\MLTT,\qA}{multi}{trace}(1^\secp, \qF, \qE,
t, \epsilon, N, k)$
between a challenger $\qCh$ and an adversary $\qA$:

\begin{description}
    \item $\underline{\expb{\MLTT,\qA}{multi}{trace}(1^\secp, \qF,
        \qE, t, \epsilon, N, k)}:$
\begin{enumerate}
\item $\qCh$ samples $(\msk, f, \aux_f) \la \Setup(1^\secp, N, k)$ and
sends $\aux_f$ to $\qA$.
\item $\qA$ sends $q \in [N-1]$ to $\qCh$.\ryo{Is there any problem if
$q\in[N]$?}\nikhil{It seems like Claim \ref{claim:noabort} is
problematic (or at-least not so clean) as we need to condition on the fact that $\qA$ does not set
$q = N$, to break traceability of $\FC$.}
\item For each $i, j \in [k] \times [q]$, $\qCh$ samples $\id_i^j \la
[N]$ and computes $\sk_i^j \la \KG(\msk, i, \id_i^j)$. It sends
$\{\id_i^j, \sk_i^j\}_{(i,j)\in[k]\times[q]}$ to $\qA$. Define the
multi-set $Q_i \seteq
\{\id_i^j\}_{j\in[q]}$ for each $i \in [k]$.

\item $\qA$ outputs a quantum program $\qP^* = (U, \rho)$.
\item $\qCh$ tests if $\qP^*$ is $\epsilon$-good wrt $(f, \qE, t)$. If
not, it outputs $\bot$.
\item $\qCh$ runs $(\id_1^*, \ldots, \id_k^*) \la \qTrace(\msk,
\qP^*, 0.9\epsilon)$.
\item If $(\id_1^*, \ldots, \id_k^*) \in Q_1 \times \ldots
\times Q_k$, it outputs $\bot$. Else, it outputs $\top$.
\end{enumerate}
\end{description}

An MLTT scheme $\MLTT$ satisfies traceability if the following holds
for every QPT $\qA$, every $N = \poly(\secp)$, every $k =
\poly(\secp)$, and every $\epsilon = 1/\poly(\secp)$:

$$
\Pr[\expb{\MLTT,\qA}{multi}{trace}(1^\secp, \qF,
\qE, t, \epsilon, N, k) \ra \top] \le \negl(\secp)
$$

\begin{remark}
Traitor tracing definitions in the literature allow the adversary to
query the key-generation algorithm on identities of its choice. In our
definition, the challenger itself chooses identities at
random and generates keys for them. However, this weaker
notion is sufficient for the purposes of our SKL compiler.
\end{remark}

\begin{remark}
Whenever we refer to an MLTT scheme, we mean one that is
collusion-resistant as per our traceability notion above. Note
that our notions require correctness and traceability to
hold for every $N = \poly(\secp)$. In some of our discussions, we consider
MLTT with an
exponential size identity space $N = 2^\secp$, where it is assumed that these
notions hold for the specific identity space $[2^\secp]$.
\end{remark}

\end{definition}
\end{definition}

% !TEX root = main.tex
\section{Two-Superposition States in the Collusion Setting}\label{sec:two-sup}

Let $\expb{\qA}{two}{sup}(1^\secp, q, N)$ be an experiment that is
defined as follows between a quantum challenger $\qCh$ and an
unbounded quantum adversary $\qA$.

\begin{description}
\item $\underline{\expb{\qA}{two}{sup}(1^\secp, q, N)}:$
\begin{enumerate}
\item For each $i \in [q]$, $\qCh$ performs the following:
\begin{itemize}
    \item Sample $x_i^0, x_i^1 \la [N]$ such that $x_i^0 \neq x_i^1$ 
and $b_i \la \bit$. Define $Q_i$ as $Q_i \seteq \{x_i^0, 
x_i^1\}$.
\item Let $b$ be a placeholder for $b_i$. Construct the following on register $\qreg{A_i}$:
    $$\sigma_i \seteq \frac{1}{\sqrt{2}}\ket{x_i^0} +
    (-1)^{b}\frac{1}{\sqrt{2}}\ket{x_i^1}$$
\end{itemize}
\item $\qCh$ sends the registers $(\qreg{A_1}, \ldots, \qreg{A_q})$
to $\qA$.
\item $\qA$ sends $(g_1, \ldots, g_q)$ and a value $\msg$ to $\qCh$.
\item $\qCh$ checks if there exists $i \in [q]$ such that $\msg \in
Q_i$. If not, it outputs $\bot$. Let $s$ be such an index.
\item If $g_s(x_s^0, x_s^1) = b_s$ holds, $\qCh$ outputs
$\top$. Else, it outputs $\bot$.
\end{enumerate}
\end{description}

\begin{theorem}\label{thm:two-sup}
For every $q, t \in \mathbb{N}$, there exists $N = O(q^2t^2)$ such
that for every unbounded quantum adversary $\qA$, the following holds:

$$\Pr[\expb{\qA}{two}{sup}(1^\secp, q, N) \ra \top] \le \frac12 +
\frac1{t}$$

Particularly, for all $q\in \mathbb{N}$, $N=128q^2$, 
$\Pr[\expb{\qA}{two}{sup}(1^\secp, q, N) \ra \top] \le 3/4$.
\end{theorem}
\fuyuki{I thought `` $\le \frac12 +
\frac1{\poly(\secp)}$'' for unspecified polynomial does not make sense ($\poly(\secp)$ could be $1$).}
\nikhil{Changed the theorem statement.}

\begin{proof}
We will consider the following sequence of hybrids:

\begin{description}
\item $\underline{\Hyb_0(1^\secp)}:$ This is the same as the
experiment $\expb{\qA}{two}{sup}(1^\secp, q, N)$:

\begin{enumerate}
\item For each $i \in [q]$, $\qCh$ performs the following:
\begin{itemize}
    \item Sample $x_i^0, x_i^1 \la [N]$ such that $x_i^0 \neq x_i^1$ 
and $b_i \la \bit$. Define $Q_i$ as $Q_i \seteq \{x_i^0, 
x_i^1\}$.
\item Let $b$ be a placeholder for $b_i$. Construct the following on register $\qreg{A_i}$:
    $$\sigma_i \seteq \frac{1}{\sqrt{2}}\ket{x_i^0} +
    (-1)^{b}\frac{1}{\sqrt{2}}\ket{x_i^1}$$
\end{itemize}
\item $\qCh$ sends the registers $(\qreg{A_1}, \ldots, \qreg{A_q})$
to $\qA$.
\item $\qA$ sends $(g_1, \ldots, g_q)$ and a value $\msg$ to $\qCh$.
\item $\qCh$ checks if there exists $i \in [q]$ such that $\msg \in
Q_i$. If not, it outputs $\bot$. Let $s$ be such an index.
\item If $g_s(x_s^0, x_s^1) = b_s$ holds, $\qCh$ outputs
$\top$. Else, it outputs $\bot$.
\end{enumerate}

\item $\underline{\Hyb_1(1^\secp)}:$ This is similar to
$\Hyb_0(1^\secp)$, with the following differences colored in red:

\begin{enumerate}
\item For each $i \in [q]$, $\qCh$ performs the following:
\begin{itemize}
\item \textcolor{red}{Sample $x_i^0, x_i^1 \la [N]$ such that $x_i^0 \neq x_i^1$. Define $Q_i$ as $Q_i \seteq \{x_i^0, 
        x_i^1\}$.}

    \item \textcolor{red}{$\qCh$ constructs the following state $\sigma_i$ on registers
$\qreg{C_i}$ and $\qreg{A_i}$.
$$\sigma_i \seteq \frac12
\sum_{c \in
\bit}\ket{c}_{\qreg{C_i}}\otimes \big(\ket{x_i^0} +
(-1)^{c}\ket{x_i^1}\big)_{\qreg{A_i}}$$}
\end{itemize}

\item $\qCh$ sends the registers $(\qreg{A_1}, \ldots, \qreg{A_q})$
to $\qA$.
\item $\qA$ sends $(g_1, \ldots, g_q)$ and a value $\msg$ to $\qCh$.
\item $\qCh$ checks if there exists $i \in [q]$ such that $\msg \in
Q_i$. If not, it outputs $\bot$. Let $s$ be such an index.
\item \textcolor{red}{$\qCh$ measures register $\qreg{C_s}$ in the
computational basis to obtain outcome $c_s$.}
\item If $g_s(x_s^0, x_s^1) = \textcolor{red}{c_s}$ holds, $\qCh$ outputs
$\top$. Else, it outputs $\bot$.
\end{enumerate}

\item $\underline{\Hyb_2(1^\secp)}:$ This is similar to
$\Hyb_1(1^\secp)$, with the following differences colored in red:

\begin{enumerate}
\item For each $i \in [q]$, $\qCh$ performs the following:
\begin{itemize}
\item \textcolor{black}{Sample $x_i^0, x_i^1 \la [N]$ such that $x_i^0 \neq x_i^1$. Define $Q_i$ as $Q_i \seteq \{x_i^0, 
        x_i^1\}$.}

    \item \textcolor{black}{$\qCh$ constructs the following state $\sigma_i$ on registers
$\qreg{C_i}$ and $\qreg{A_i}$.
$$\sigma_i \seteq \frac12
\sum_{c \in
\bit}\ket{c}_{\qreg{C_i}}\otimes \big(\ket{x_i^0} +
(-1)^{c}\ket{x_i^1}\big)_{\qreg{A_i}}$$}
\end{itemize}

\item $\qCh$ sends the registers $(\qreg{A_1}, \ldots, \qreg{A_q})$
to $\qA$.
\item $\qA$ sends $(g_1, \ldots, g_q)$ and a value $\msg$ to $\qCh$.
\item $\qCh$ checks if there exists $i \in [q]$ such that $\msg \in
Q_i$. If not, it outputs $\bot$. Let $s$ be such an index.

\item \textcolor{red}{Let $\chec[x_s^0, x_s^1, \msg]$ be a function such that
$\chec[x_s^0, x_s^1, \msg](u) = \top$ if $\msg = x_s^u$ and $\bot$ otherwise.
Apply the following map to the register $\qreg{C_s}$ in the Hadamard
basis, and an ancilla register $\qreg{OUT}$ initialized to $\ket{0}$:
$$
\ket{u}_{\qreg{C_s}}\ket{w}_{\qreg{OUT}} \mapsto
\ket{u}_{\qreg{C_s}}\ket{w \xor \chec[x_s^0, x_s^1, \msg](u)}_{\qreg{OUT}}$$
}

\item \textcolor{red}{$\qCh$ measures register $\qreg{OUT}$ in the computational
basis to get outcome $\out$. If $\out = \bot$, output $\bot$.}

\item \textcolor{black}{$\qCh$ measures register $\qreg{C_s}$ in the
computational basis to obtain outcome $c_s$.}
\item If $g_s(x_s^0, x_s^1) = \textcolor{black}{c_s}$ holds, $\qCh$ outputs
$\top$. Else, it outputs $\bot$.
\end{enumerate}

\item $\underline{\Hyb_3(1^\secp)}:$ This is similar to
$\Hyb_2(1^\secp)$, with the following differences colored in red:

\begin{enumerate}
\item For each $i \in [q]$, $\qCh$ performs the following:
\begin{itemize}
\item \textcolor{black}{Sample $x_i^0, x_i^1 \la [N]$ such that $x_i^0 \neq x_i^1$. Define $Q_i$ as $Q_i \seteq \{x_i^0, 
        x_i^1\}$.}

    \item \textcolor{black}{$\qCh$ constructs the following state $\sigma_i$ on registers
$\qreg{C_i}$ and $\qreg{A_i}$.
$$\sigma_i \seteq \frac12
\sum_{c \in
\bit}\ket{c}_{\qreg{C_i}}\otimes \big(\ket{x_i^0} +
(-1)^{c}\ket{x_i^1}\big)_{\qreg{A_i}}$$}
\end{itemize}

\item $\qCh$ sends the registers $(\qreg{A_1}, \ldots, \qreg{A_q})$
to $\qA$.
\item $\qA$ sends $(g_1, \ldots, g_q)$ and a value $\msg$ to $\qCh$.
\item $\qCh$ checks if there exists $i \in [q]$ such that $\msg \in
Q_i$. If not, it outputs $\bot$. Let $s$ be such an index.

\item \textcolor{red}{$\qCh$ measures registers $\qreg{C_1}, \ldots,
\qreg{C_q}$ in the Hadamard basis to get outcomes $c_1' \ldots, c_s'$
respectively. For $u \seteq c_s'$, if $\msg \neq x_s^u$, $\qCh$
outputs $\bot$.}

\item \textcolor{black}{$\qCh$ measures register $\qreg{C_s}$ in the
computational basis to obtain outcome $c_s$.}
\item If $g_s(x_s^0, x_s^1) = \textcolor{black}{c_s}$ holds, $\qCh$ outputs
$\top$. Else, it outputs $\bot$.
\end{enumerate}

\item $\underline{\Hyb_4(1^\secp)}:$ This is similar to
$\Hyb_3(1^\secp)$, with the following differences colored in red:

\begin{enumerate}
\item For each $i \in [q]$, $\qCh$ performs the following:
\begin{itemize}
\item \textcolor{black}{Sample $x_i^0, x_i^1 \la [N]$ such that $x_i^0 \neq x_i^1$. Define $Q_i$ as $Q_i \seteq \{x_i^0, 
        x_i^1\}$.}

    \item \textcolor{black}{$\qCh$ constructs the following state $\sigma_i$ on registers
$\qreg{C_i}$ and $\qreg{A_i}$.
$$\sigma_i \seteq \frac12
\sum_{c \in
\bit}\ket{c}_{\qreg{C_i}}\otimes \big(\ket{x_i^0} +
(-1)^{c}\ket{x_i^1}\big)_{\qreg{A_i}}$$}
\end{itemize}

\item \textcolor{red}{$\qCh$ measures registers $\qreg{C_1}, \ldots,
        \qreg{C_q}$ in
the Hadamard basis to get outcomes $c_1', \ldots, c_q'$.}

\item $\qCh$ sends the registers $(\qreg{A_1}, \ldots, \qreg{A_q})$
to $\qA$.
\item $\qA$ sends $(g_1, \ldots, g_q)$ and a value $\msg$ to $\qCh$.
\item $\qCh$ checks if there exists $i \in [q]$ such that $\msg \in
Q_i$. If not, it outputs $\bot$. Let $s$ be such an index.

\item \textcolor{red}{For $u \seteq c_s'$, if $\msg \neq x_s^u$, $\qCh$
outputs $\bot$.}

\item \textcolor{black}{$\qCh$ measures register $\qreg{C_s}$ in the
computational basis to obtain outcome $c_s$.}
\item If $g_s(x_s^0, x_s^1) = \textcolor{black}{c_s}$ holds, $\qCh$ outputs
$\top$. Else, it outputs $\bot$.
\end{enumerate}

\end{description}

\begin{remark}
The statements of the following claims are meant for arbitrary
$t\in\mathbb{N}$ and appropriately chosen $N = O(t^2q^2)$.
\end{remark}

\begin{claim}
The probability that $\qCh$ outputs $\bot$ in Step 6. of $\Hyb_4$ is
at most $1/4t^2$.
\end{claim}
\begin{proof}
Notice that for $i \in [q]$, the states $\sigma_i$ are of the
following form:

$$\sigma_i = \frac12
\sum_{c\in\bit}\ket{c}_{\qreg{C_i}}\otimes\big(\ket{x_i^0} +
    (-1)^c\ket{x_i^1}\big) =
    \frac{1}{\sqrt2}\big(\ket{+}_{\qreg{C_i}}\ket{x_i^0}_{\qreg{A_i}} +
    \ket{-}_{\qreg{C_i}}\ket{x_i^1}_{\qreg{A_i}}\big)$$

Since $\qCh$ measures the registers $\qreg{C_1}, \ldots, \qreg{C_q}$
in the Hadamard basis to get $c_1', \ldots, c_q'$, when $\qA$
receives the registers $\qreg{A_1}, \ldots, \qreg{A_q}$, the states
are of the following form:

$$\ket{c_1'}_{\qreg{C_1}}\ket{x_1^{c_1'}}_{\qreg{A_1}}, \ldots,
\ket{c_q'}_{\qreg{C_q}}\ket{x_q^{c_q'}}_{\qreg{A_q}}$$

Now, let $Q$ be a multi-set such that $Q \seteq Q_1 \cup \ldots \cup
Q_q$. Let $\Event$ be the event that all the elements of $Q$ are
distinct. It is easy to see from a birthday bound analysis that if
$N = O(q^2 t^2)$, we can
ensure that the probability $\Event$ does \emph{not} occur is bounded by
$1/8t^2$. Consider now the set $S \seteq
Q \setminus \{x_1^{c_1'}, \ldots, x_q^{c_q'}\}$. Conditioned on
$\Event$ occurring, $S$ is a uniformly random subset of $[N]$ of size
$q$. Notice that for $\qA$ to cause an abort, it must produce $\msg$
such that $\msg \in S$. However, the only information $\qA$ has are
the values $x_1^{c_1'}, \ldots, x_q^{c_q'}$ which are
independent of the values in $S$.
Therefore, we have that $\Pr[\msg \in S | \Event]
\le q/N$, which is less than $1/8t^2$ for appropriate $N =
O(q^2t^2)$. Consequently, the probability that $\qCh$ outputs $\bot$
in Step 6. of $\Hyb_4$ is bounded by $1/8t^2 + 1/8t^2 = 1/4t^2$.
\fuyuki{If I understand correctly, this claim want to bound the probability that the experiment aborts as the last sentence of the proof says, not the probability that the experiment outputs $\bot$ as the statement of the claims says. This is confusing. We should not use the same symbol to indicate the adversary's failure and the experiment's (artificial) abort.}
\nikhil{Changed last sentence of the proof to match the claim.}
\end{proof}

\begin{claim}
The probability that $\qCh$ outputs $\bot$ in Step 5. of $\Hyb_3$ is
at most $1/4t^2$.
\end{claim}
\begin{proof}
This follows from the fact that operations performed on different
registers commute. Hence, the Hadamard measurements on $\qreg{C_1},
\ldots, \qreg{C_q}$ can be performed after $\qA$ submits its response,
without altering the probability of $\bot$ from $\Hyb_4$.
\end{proof}

\begin{claim}
The probability that $\qCh$ outputs $\bot$ in Step 6. of $\Hyb_2$ is
at most $1/4t^2$.
\end{claim}
\begin{proof}
This follows directly because the same check as the one introduced in
$\Hyb_3$ is applied coherently in $\Hyb_2$.
\end{proof}

\begin{claim}
Probability $\qCh$ outputs $\top$ in $\Hyb_2$ is at most $1/2$.
\end{claim}
\begin{proof}
Notice that the measurement on register $\qreg{OUT}$ necessarily collapses
$\qreg{C_s}$ in the Hadamard basis. Consequently, the following
computational basis measurement on $\qreg{C_s}$ yields a truly random
bit $c_s$ that is independent of $x_s^0, x_s^1$ and $g_s$.
\end{proof}

\begin{claim}
Probability $\qCh$ outputs $\top$ in $\Hyb_1$ is at most $1/2 + 1/t$.
\end{claim}
\begin{proof}
The only difference between $\Hyb_1$ and $\Hyb_2$ is the measurement
on register $\qreg{OUT}$ in $\Hyb_2$. Since we argued that this
measurement outputs $\bot$ with probability at most $1/4t^2$,
it follows from the gentle measurement lemma (Lemma
\ref{lma:gentle}) that these hybrids are $1/t$ close in
trace distance.
\end{proof}

\fuyuki{I did not see the reason why we need $\Hyb_2$. Why can't we go from $\Hyb_1$ to $\Hyb_3$ directly? (If I remember correctly, there isn't something like $\Hyb_2$ in the proof of BKM+.)}
\nikhil{I thought it was easier to understand this way. In $\Hyb_2$,
since the measurement is performed coherently, it would only end up
collapsing $\qreg{C_s}$. Also, I think the gentle measurement argument
relating $\Hyb_1$ and $\Hyb_2$ makes more sense due to the coherent
measurement.}

Notice that in $\Hyb_1$, if the measurement on register $\qreg{C_s}$
were performed at the beginning itself, then it is equivalent to
$\Hyb_0$. Since operations on different registers commute, this
measurement can be performed after the response of $\qA$. This
proves that for every $q, t \in \mathbb{N}$,
$\Pr[\expb{\qA}{two}{sup}(1^\secp, q, N) \ra \top] \le 1/2 +
1/t$ for some $N = O(t^2q^2)$. It is also easy to see that for every
$q \in \mathbb{N}$ and $N=128q^2$,
$\Pr[\expb{\qA}{two}{sup}(1^\secp, q, N) \ra \top] \le 3/4$.
\end{proof}

Next, consider the experiment $\expb{\qA}{two}{sup}(1^\secp, k, q, N)$
that corresponds to the $k$-fold parallel repetition of
$\expb{\qA}{two}{sup}(1^\secp, q, N)$. That is, $\qCh$ executes $k$
independent instances of $\expb{\qA}{two}{sup}(1^\secp, q, N)$ with
the adversary $\qA$, and outputs $\top$ only if each of the $k$
experiments outputs $\top$. We utilize the following theorem
(paraphrased) of Bostanci et al. \cite{STOC:BQSY24} to argue
security of the parallel repetition experiment:

\begin{theorem}[Quantum Parallel Repetition
\cite{STOC:BQSY24}]\label{thm:||}
Let $\Pi$ be a quantum interactive game with at-most 3-messages
and $\epsilon$-soundness against QPT adversaries. Then, the $k$-fold
parallel repetition $\Pi^k$ has $(\epsilon^k +
\negl(\secp))$-soundness against QPT adversaries.
\end{theorem}

Note that in this theorem, $\epsilon$-soundness means that
the challenger outputs $\top$ with probability at most $\epsilon$, in an
interactive game with an adversary. As a consequence of Theorem
\ref{thm:||}, \cref{thm:two-sup} immediately gives us the following:

\begin{theorem}\label{thm:two-sup-par}
For every $q \in \mathbb{N}$, $N=128q^2$ and $k=\secp$, the following
holds for every QPT adversary $\qA$:

$$\Pr[\expb{\qA}{two}{sup}(1^\secp, k, q, N) \ra \top] \le \negl(\secp)$$
\end{theorem}

Note that these parameter choices are not tight, and are chosen to
demonstrate feasibility. We also consider the experiment
$\expb{\qA}{two}{sup}(1^\secp, k, \bot, N)$ that is defined similarly,
except $\qA$ specifies $q =
\poly(\secp)$ at the start of the experiment (i.e., $N$ doesn't depend
on $q$), which is used for each
of the $k$ repetitions. It is easy to see that the following
is implied by the proof of Theorem \ref{thm:two-sup}:

\begin{corollary}\label{cor:two-sup}
For $N = 2^\secp$, $k=\secp$, and every QPT adversary $\qA$, it holds
that:

$$\Pr[\expb{\qA}{two}{sup}(1^\secp, k, \bot, N) \ra \top] \le \negl(\secp)$$
\end{corollary}

% !TEX root = main.tex

\section{SKL from Multi-Level Traitor Tracing}\label{sec:const}
We now construct an SKL scheme $\SKL$ for application $(\qF, \qE,
t)$ using an MLTT scheme $\MLTT$ for $(\qF, \qE, t)$. To
guarantee $q$-bounded standard-KLA security of $\SKL$, we assume
that $\MLTT$ admits an identity space $[N]$ for any $N =
\poly(\secp)$. If $\MLTT$ admits identity space $[2^\secp]$,
then $\SKL$ satisfies unbounded standard-KLA security. The
MLTT scheme consists of algorithms $\MLTT.(\Setup, \KG, \Eval,
\qTrace)$ and the SKL scheme's algorithms $(\Setup, \qKG, \qEval,
\qDel, \Vrfy)$ are as follows:

\begin{description}
\item $\underline{\Setup(1^\secp, q)}:$
\begin{enumerate}
\item Define $k \seteq \secp$. If $q = \bot$, define $N \seteq 2^\secp$
(represented succinctly). Otherwise, define $N \seteq 128q^2$.
\item Compute $(\mltt.\msk, f, \aux_f) \la \MLTT.\Setup(1^\secp, N, k)$.
\item Output $(\msk \seteq (\mltt.\msk, N, k), f, \aux_f)$.

\end{enumerate}

\item $\underline{\qKG(\msk)}:$
\begin{enumerate}
\item Parse $\msk = (\mltt.\msk, N, k)$.
\item For each $i \in [k]$, do the following:
\begin{itemize}
\item Sample $v_i, w_i \gets [N]$ and $b_i \gets \bit$.
\item Compute $\sk_{i,v} \gets \MLTT.\KG(\mltt.\msk, i, v_i)$.
\item Compute $\sk_{i,w} \gets \MLTT.\KG(\mltt.\msk, i, w_i)$.
\item Compute $\vk_i = (v_i, w_i, \sk_{i, v}, \sk_{i, w}, b_i)$.
\item Compute the following state on registers $\qreg{C_i, D_i}$:
    $$\rho_i \seteq
    \frac{1}{\sqrt2}\ket{v_i}_{\qreg{C_i}}\ket{\sk_{i,v}}_{\qreg{D_i}} +
    (-1)^{b_i}\cdot\frac{1}{\sqrt2}\ket{w_i}_{\qreg{C_i}}\ket{\sk_{i,w}}_{\qreg{D_i}}$$
\end{itemize}
\item Output $\qsk \seteq (\rho_i)_{i \in [k]}$ and $\vk \seteq
(\vk_i)_{i \in [k]}$.
\end{enumerate}

\item $\underline{\qEval(\qsk, x)}:$
\begin{enumerate}
\item Parse $\qsk$ as $\qsk = (\rho_i)_{i\in[k]}$. For each
$i\in[k]$, parse $\rho_i$ as a state on registers $\qreg{C_i}$ and
$\qreg{D_i}$.
Define the registers $\qreg{C} \seteq \qreg{C_1} \otimes \ldots \otimes
\qreg{C_k}$ and $\qreg{D} \seteq \qreg{D_1} \otimes \ldots \otimes
\qreg{D_k}$.

\item Apply the following map, where $\qreg{OUT}$ is a register
    initialized to $\ket{0\ldots0}$.
$$\ket{v}_{\qreg{C}}\ket{w}_{\qreg{D}}\ket{z}_{\qreg{OUT}} \mapsto
\ket{v}_{\qreg{C}}\ket{w}_{\qreg{D}}\ket{z \xor \MLTT.\Eval(w,x)}_{\qreg{OUT}}
$$

\item Measure the register $\qreg{OUT}$ to obtain an outcome $y$, and output it.
\end{enumerate}

\item $\underline{\qDel(\qsk)}:$
\begin{enumerate}
\item Parse $\qsk$ as $\qsk = (\rho_i)_{i\in[k]}$.
\item For each $i\in[k]$, parse $\rho_i$ as a state on registers
$\qreg{C_i, D_i}$ and measure $\qreg{C_i, D_i}$ in the Hadamard basis to get outcomes $c_i, d_i$.
\item Output $\cert \seteq (c_i, d_i)_{i\in[k]}$.
\end{enumerate}
\item $\underline{\Vrfy(\vk, \cert)}:$
\begin{enumerate}
\item Parse $\cert$ as $\cert \seteq (c_i, d_i)_{i\in[k]}$ and
$\vk$ as $\vk = (\vk_i)_{i\in[k]}$.
\item For each $i\in[k]$, execute the following:
\begin{itemize} 
\item Parse $\vk_i$ as $\vk_i = (v_i, w_i, \sk_{i,v}, \sk_{i,w},
b_i)$.
\item If $\langle(v_i \| \sk_{i,v}) \xor (w_i \|
\sk_{i,w}), c_i \| d_i\rangle \neq b_i$, output $\bot$.
\end{itemize} 
\item Output $\top$.
\end{enumerate}
\end{description}

\paragraph{Evaluation Correctness:}
Observe that $\qEval$ applies the following map:
$$\ket{v}_{\qreg{C}}\ket{w}_{\qreg{D}}\ket{z}_{\qreg{OUT}} \mapsto
\ket{v}_{\qreg{C}}\ket{w}_{\qreg{D}}\ket{z \xor \MLTT.\Eval(w,x)}_{\qreg{OUT}}$$

Since the register $\qreg{D}$ contains $k$ MLTT secret keys (one for
each level), evaluation correctness of $\MLTT$ allows to
compute valid outputs on $\qreg{OUT}$, except for a negligible
fraction of superposition terms. Moreover, the deterministic
evaluation property of $\MLTT$ ensures that some $1-\negl(\secp)$
fraction of outputs are identical. By the gentle measurement lemma, it
follows that the state can be re-used for arbitrary polynomially many
evaluations.

\paragraph{Verification Correctness:}

This follows directly from the fact that measuring
$\frac{1}{\sqrt2}(\ket{x} + (-1)^b\ket{y})$ in the Hadamard basis
yields a value $d$ such that $d\cdot(x \xor y) = b$.

\begin{theorem}\label{thm:const}
The SKL scheme $\SKL$ satisfies $q$-bounded
standard-KLA security for any $q = \poly(\secp)$, assuming $\MLTT$ is
an MLTT scheme satisfying traceability.
\end{theorem}

\begin{proof}
Let $\qA$ be a QPT adversary in
$\expb{\SKL,\qA}{std}{kla}(1^\secp, q, \qF,\qE,t,\epsilon)$ for some
$\epsilon = 1/\poly(\secp)$ and let $\win$
denote the event that $\qA$ wins the SKL experiment. Assume $\win$
occurs with $\nonnegl(\secp)$ probability. Let $(\qsk^1, \vk^1),
\ldots, (\qsk^q, \vk^q)$ denote the leased secret-key and
verification-key pairs sampled by $\qKG$ in the experiment. For each
$j \in [q]$, we have $\qsk^j = (\qsk^j_i)_{i\in[k]}$ and $\vk^j =
(\vk_i^j)_{i\in[k]}$ by construction, where $\vk_i^j = (v_i^j,
w_i^j, \sk_{i,v}^j, \sk_{i,w}^j, b_i^j)$. For each $(i, j) \in [k]
\times [q]$, define the set $Q_i^j = \{v_i^j, w_i^j\}$. For each $i
\in [k]$, let $Q_i \seteq Q_i^1 \cup \ldots \cup Q_i^q$. Recall that
$\qA$ outputs certificates $\cert^1, \ldots, \cert^q$ and a quantum
program $\qP^*$.

Let $(\id_1^*, \ldots, \id_k^*)$ be computed as $(\id_1^*, \ldots,
\id_k^*) \la \MLTT.\qTrace(\mltt.\msk, \qP^*, \allowbreak0.9\epsilon)$ and
$\GoodExt_\epsilon$ denote the event that
$(\id_1^*, \ldots, \id_k^*) \in Q_1 \times \ldots \times Q_k$. Let
$\APILive_\epsilon$ denote the event that the $\epsilon$-good test wrt
$(f, \qE, t)$ outputs 1 when applied to $\qP^*$.

Assume
for contradiction that $\Pr[\lnot \GoodExt_\epsilon \mid
\APILive_\epsilon] =
\nonnegl(\secp)$. We will then construct the following reduction
algorithm $\qR^\qA$ that breaks the traceability of the MLTT scheme
$\MLTT$ for $N = 128q^2$ and $k = \secp$:

\begin{description}
\item \underline{Execution of $\qR^\qA$ in
    $\expb{\MLTT,\qR}{multi}{trace}(1^\secp,\qF,\qE,t,\epsilon, N, k)$}:

\begin{enumerate}
\item $\qCh$ samples $(\mltt.\msk, f, \aux_f) \la \MLTT.\Setup(1^\secp, N, k)$ and
sends $\aux_f$ to $\qR$. $\qR$ sends $\aux_f$ to $\qA$.
\item $\qR$ sends $2q$ to $\qCh$.
\item For each $(i, j) \in [k] \times [2q]$, $\qCh$ samples $\id_i^j
\la [N]$ and computes $\sk_i^j \la \KG(\mltt.\msk, i, \id_i^j)$. It sends
$\{\id_i^j, \sk_i^j\}_{(i,j)\in [k] \times [2q]}$ to $\qR$.
Define $Q_i$ to be the multi-set $Q_i \seteq
\{\id_i^j\}_{j \in [2q]}$.

\item For each $(i, j) \in [k] \times [q]$, $\qR$ does the
following:
\begin{itemize}
    \item Sample $b_i^j \la \bit$. Set $v_i^j \seteq \id_i^j$ and
        $w_i^j \seteq \id_i^{q+j}$.
\item Compute $\sk_{i,v}^j \seteq \sk_i^j$ and
    $\sk_{i,w}^j \seteq \sk_i^{q+j}$.
\item Compute the state $\rho_i^{j}$ on registers $\qreg{C_i^j,
D_i^j}$ similar to that of $\SKL.\qKG$.
\end{itemize}
\item For each $j \in [q]$, $\qR$ computes $\qsk^j \seteq
(\rho^j_i)_{i\in[k]}$ and $\vk^j \seteq (\vk^j_i)_{i\in[k]}$ where
$\vk_i^j \seteq (v_i^j, w_i^j, \sk_{i,v}^j, \sk_{i,w}^j, b_i^j)$.
Then, it sends $\qsk^1, \ldots, \qsk^q$ to $\qA$.
\item $\qA$ sends $(\cert^1, \ldots, \cert^q)$ and a quantum program $\qP^* =
(U, \rho)$ to $\qR$.\ryo{Shouldn't $\qR$ check
$(\cert^1,\ldots,\cert^q)$ are valid or not?}\nikhil{I believe its not
needed. As long as $\Pr[\lnot \GoodExt \land \APILive] = \nonnegl(\secp)$,
security of $\MLTT$ is broken.}
\item $\qR$ outputs the quantum program $\qP^*$. \nikhil{$\qR$ doesn't
need to check if $\qP^*$ is good either.}
\item $\qCh$ tests if $\qP^*$ is $\epsilon$-good wrt $(f, \qE, t)$. If not,
    it outputs $\bot$.
\item $\qCh$ runs $(\id_1^*, \ldots, \id_k^*) \la
    \MLTT.\qTrace(\mltt.\msk,
\qP^*, 0.9\epsilon)$.

\item If $(\id_1^*, \ldots,
\id_k^*) \in Q_1 \times \ldots \times Q_k$, $\qCh$ outputs $\bot$.
Else, it outputs $\top$.
\end{enumerate}
\end{description}

Observe that the view of $\qA$ in $\qR$ is indistinguishable from its
view in the standard-KLA experiment for $\SKL$.
Since we have $\Pr[\Win] = \nonnegl(\secp)$ and that
$\APILive_\epsilon$ occurs when $\Win$ occurs, we have
$\Pr[\APILive_\epsilon] = \nonnegl(\secp)$. This means $\Pr[\lnot
\GoodExt_\epsilon \land \APILive_\epsilon] = \nonnegl(\secp)$, which breaks
the security of $\MLTT$.

Now, assume that $\Pr[\lnot \GoodExt_\epsilon \mid
\APILive_\epsilon] \le \negl(\secp)$, which means that
$\Pr[\GoodExt_\epsilon \mid \APILive_\epsilon] \ge 1 -
\negl(\secp)$. We have that
$\Pr[\GoodExt_\epsilon \land \Win] =
\Pr[\Win]\cdot\Pr[\GoodExt_\epsilon \mid \Win] =
\nonnegl(\secp)\cdot \nonnegl(\secp) = \nonnegl(\secp)$ by assumption
and because $\APILive_\epsilon$ occurs whenever $\Win$ occurs.

Moreover, we must also have 
$\Vrfy(\vk^1, \cert^1) = \ldots = \Vrfy(\vk^q, \cert^q) =
\top$ conditioned on $\win$.
We will exploit this fact to construct the following
reduction $\qB$, which breaks the collusion-resistant security of
two-superposition states (for $k = \secp$ and $N = 128q^2$).

\begin{description}
\item \underline{Execution of $\qB^\qA$ in
    $\expb{\qB}{two}{sup}(1^\secp, k, q, N)$}:

\begin{enumerate}
\item $\qB$ samples $(\mltt.\msk, f, \aux_f) \la \MLTT.\Setup(1^\secp, N, k)$
and sends $\aux_f$ to $\qA$.

\item For each $(i, j) \in [k] \times [q]$, $\qCh$ performs the
    following:
\begin{itemize}
    \item Sample $v_i^j, w_i^j \la [N]$ and $b_i^j \la \bit$.
\item Set $b \seteq b_i^j$ and construct the following state on
    register $\qreg{C_{i,j}}$:
    $$\sigma_i^j \seteq \frac{1}{\sqrt{2}}\ket{v_i^j}_{\qreg{C_{i,j}}} +
    (-1)^{b}\frac{1}{\sqrt{2}}\ket{w_i^j}_{\qreg{C_{i,j}}}$$
\end{itemize}
\item For each $i \in [k]$, $\qCh$ sets $\sigma_i \seteq \sigma_i^1 \otimes
    \ldots \otimes \sigma_i^q$ and sends $\sigma_i$ to $\qB$.

\item For each $(i, j) \in [k] \times [q]$, $\qB$ performs the
following:
\begin{itemize}
\item Initialize a register $\qreg{D_{i,j}}$ to $\ket{0\ldots0}$.
Apply the following map to the registers $\qreg{C_{i,j}},
\qreg{D_{i,j}}$:

$$\ket{u}_{\qreg{C_{i,j}}}\ket{z}_{\qreg{D_{i,j}}} \mapsto
\ket{u}_{\qreg{C_{i,j}}}\ket{z \xor
\MLTT.\KG(\mltt.\msk,i,u)}_{\qreg{D_{i,j}}}
$$

Let the resulting state be denoted as $\qsk_i^j$.
\end{itemize}

\item For each $j \in [q]$, $\qB$ sets $\qsk^j \seteq
(\qsk_i^j)_{i\in[k]}$ and sends $\qsk^j$ to $\qA$.

\item $\qA$ sends $(\cert^1, \ldots, \cert^q)$ and a program $\qP^*
= (U, \rho)$ to $\qB$.
\item For each $j \in [q]$, $\qB$ parses $\cert^j$ as $\cert^j =
(c_i^j, d_i^j)_{i \in [k]}$.

\item For each $(i, j) \in [k] \times [q]$, $\qB$ considers the
following function $g_i^j$:
\begin{description}
\item $\underline{g_i^j(v_i^j, w_i^j)}:$
\begin{itemize}
\item Compute $\sk_{i,v}^j \la \MLTT.\KG(\mltt.\msk, i, v_i^j)$.
\item Compute $\sk_{i,w}^j \la \MLTT.\KG(\mltt.\msk, i, w_i^j)$.
\item Output $c_i^j \cdot (v_i^j \xor w_i^j) \xor d_i^j \cdot
(\sk_{i,v}^j \xor \sk_{i,w}^j)$.
\end{itemize}
\end{description}

\item $\qB$ tests if $\qP^*$ is $\epsilon$-good wrt $(f, \qE, t)$.

\item $\qB$ runs $(\id_1^*, \ldots, \id_k^*)
\la \MLTT.\qTrace(\mltt.\msk, \qP^*, 0.9\epsilon)$.

\item For each $i \in [k]$, $\qB$ sends $(g_i^1, \ldots, g_i^q)$ and
$\id_i^*$ to $\qCh$.

\item For each $(i, j) \in [k] \times [q]$, let $Q_i^j \seteq
\{v_i^j, w_i^j\}$. For each $i\in[k]$, $\qCh$ performs the following:
\begin{itemize}
%\item Set $Q_i \seteq Q_i^1 \cup \ldots \cup Q_i^q$.
\item Check if there exists $j \in [q]$ such that $\id_i^* \in
Q^j_i$. If not, output $\bot$. Let $s \in [q]$ be such an index.
\item Check if $g_i^s(v_i^s, w_i^s) = b_i^s$ holds. If not, output $\bot$.
\end{itemize}
\item Output $\top$.
\end{enumerate}
\end{description}

Notice that the view of $\qA$ in the experiment is indistinguishable
from its view in the SKL experiment.
Observe now that the condition $g^j_i(v_i^j, w_i^j) = b_i^j$ holds for
every $(i, j) \in [k] \times [q]$, whenever $\win$ occurs.
This is because the algorithms $\Vrfy(\vk^1, \cert^1), \ldots,
\Vrfy(\vk^q, \cert^q)$ will check that for
every $(i, j) \in [k] \times [q]$, it holds that $b_i^j = (c_i^j \|
d_i^j) \cdot (v_i^j \| \sk_{i,v}^j \xor w_i^j \| \sk_{i,w}^j) =
c_i^j \cdot (v_i^j \xor w_i^j) \xor d_i^j \cdot (\sk_{i,v}^j \xor
\sk_{i,w}^j) = g_i^j(v_i^j, w_i^j)$. Moreover, when
$\GoodExt_\epsilon$ occurs, 
for every $i \in [k]$, $\id_i^* \in Q_i$
where $Q_i \seteq Q_i^1 \cup
\ldots \cup Q_i^q$. Recall that $k = \secp$ and $N=128q^2$ as per the
construction. Consequently, $\qB$ breaks the collusion-resistant
security of two-superposition states (\cref{thm:two-sup-par}), giving us a
contradiction. Therefore, $\win$ can only occur with $\negl(\secp)$ probability.
\end{proof}

\fuyuki{I did not fully understand the reason why this step does not go through if we adopt the existing definitional style where we require $\Pr[\mathsf{Good}]\ge\Pr[\mathsf{Live}]$.
Consider the following three events.
\begin{itemize}
\item $V$: All the deletion certificates output by $\qA$ is valid.
\item $L$: The quantum program output by $\qA$ is tested as $\gamma$-good.
\item $G$: We apply $\MLTT.\qTrace$ to the quantum program output by $\qA$ and it successfully traces. (We apply $\MLTT.\qTrace$ directly without apply $\gamma$-good test.)
\end{itemize}
We have $\Pr[V\land G]\ge\Pr[V\land L]+\negl(\secp)$ from the security of MLTT if we define its security so that it requires $\Pr[\mathsf{Good}]\ge\Pr[\mathsf{Live}]+\negl(\secp)$.
This can be proved by considering an MLTT adversary that simulates $\qA$ and outputs the quantum program output by $\qA$ only when the deletion certificate output by $\qA$ are valid.
I think the MLTT adversary can perform the verification of SKL.
We can also prove that $\Pr[V\land G]$ is negligible from the security of two superposition state, which also proves $\Pr[V\land L]$ is negligible.
Am I missing something?
}
\nikhil{This also works. I found the ($\API$ + 
$\qTrace$) more intuitive, so I decided to go with it.}

\begin{remark}
It is easy to see that if $\MLTT$ admits identity space $[2^\secp]$,
$\SKL$ satisfies unbounded standard-KLA security, based on Corollary \ref{cor:two-sup}.
\end{remark}

% !TEX root = main.tex

\section{Collusion-Resistant PRF-SKL from LWE}\label{sec:mlt-prf}

In this section, we will construct an MLTT scheme for the PRF
functionality. We refer to this primitive as a multi-level traceable
PRF (MLT-PRF). The scheme will allow for an arbitrary polynomial-size
identity space. Since our traceability definition (Definition
\ref{def:mltt-tr}) assumes
collusion-resistance by default, we do not specify it explicitly.
With the help of our compiler from the previous section, the MLT-PRF
implies a bounded collusion-resistant SKL scheme for the PRF
functionality (PRF-SKL).

The MLT-PRF is similar to the traceable PRF of Maitra and Wu
\cite{PKC:MaiWu22}, and relies on two building blocks: a fingerprinting
code, and a traceable PRF with identity space $\bit$. We define these in the
following subsection.

\subsection{Building Blocks}

\begin{definition}[Quantum-Secure Traceable PRF with Identity Space $\bit$]\label{def:tprf}

A quantum-secure traceable PRF (TPRF) with identity space $\bit$ consists of the
following algorithms. Let $\cX, \cY$ denote the domain and range of
the PRF respectively.

\begin{description}
\item $\Setup(1^\secp) \ra \msk:$ The setup algorithm takes a
security parameter as input and outputs a master secret key $\msk$.
\item $\KG(\msk, \id) \ra \sk:$ The key-generation algorithm takes a
master secret-key $\msk$ as input along with an identity $\id
\in \bit$. It outputs a secret-key $\sk$. We require that $\KG$ is
deterministic.\footnote{This is wlog, assuming the existence of a
post-quantum PRF. This is because a PRF key can be sampled as part of
$\msk$, which can be used to derive randomness corresponding to input $\id$.}

\item $\Eval(\sk, x) \ra y:$ The evaluation algorithm takes as input
a secret key $\sk$ (or $\msk$) and a value $x \in \cX$. It outputs a
value $y \in \cY$.

\item $\qTrace(\msk, \qP, \epsilon^*) \ra \id/\bot:$ The quantum tracing algorithm
takes the master secret-key as input, along with a quantum
pirate program $\qP$ and a parameter $\epsilon^*$ meant to be a lower
bound on the advantage of $\qP$. It outputs a traitor identity $\id$
or $\bot$.
\end{description}

\paragraph{Evaluation Correctness:} The following holds for every
$\id \in \bit$:

\[
\Pr\left[
\Eval(\msk, x) \neq \Eval(\sk, x)
 \ :
\begin{array}{rl}
 &\msk \la \Setup(1^\secp)\\
 &x \la \cX\\
 &\sk \la \KG(\msk, \id)
\end{array}
\right] \le \negl(\secp)
\]

\paragraph{Pseudo-randomness:} Let $\cR$ denote the set of all functions
with domain $\cX$ and range $\cY$. The following holds for all QPT
adversaries $\qA$ for $f, \msk$ sampled as $f \la \cR$ and
$\msk \la \Setup(1^\secp)$ respectively:

\[
\bigg\lvert
\Pr\left[
1 \la \qA^{f(\cdot)}(1^\secp)
\right] 
-
\Pr\left[
1 \la \qA^{\Eval(\msk, \cdot)}(1^\secp)
\right]
\bigg\rvert
\le \negl(\secp)
\]

\begin{definition}[Weak Pseudo-randomness]\label{def:wprf}
For $b \in \bit$, consider the following distributions
$\Dwprf[f]$ and $\Dchall[f]$:

\begin{description}
    \item $\underline{\Dwprf[f]}:$
\begin{itemize}
\item Sample $x \la \cX$. Compute $y \seteq f(x)$.
\item Output $(x, y)$.
\end{itemize}
\end{description}

\begin{description}
    \item $\underline{\Dchall[f]}:$
\begin{itemize}
\item Sample $x \la \cX$.
\item If $b = 1$, sample $y \la \cY$. Else, compute $y \seteq f(x)$.
\item Output $(x, y)$.
\end{itemize}
\end{description}

The following holds for every QPT adversary $\qA$: 

\[
\Pr\left[
    b \la \qA^{\Dwprf[f]}(1^\secp, x, y)
 \ :
\begin{array}{rl}
 &b \la \bit\\
 &\msk \la \Setup(1^\secp)\\
 &f \seteq \Eval(\msk, \cdot)\\
 &(x, y) \la \Dchall[f]
\end{array}
\right] \le \frac12 + \negl(\secp)
\]
\end{definition}

\begin{remark}
This notion of weak pseudo-randomness is equivalent to a notion where
polynomially many queries to $\Dchall$ are allowed. This can be shown
easily via a hybrid argument.
\end{remark}

\begin{remark}
While pseudo-randomness implies its weak variant, the latter is used
to determine which pirate programs are able to successfully evaluate
the PRF in the following traceability definition. Such a traceability
notion is inspired by previous works on traceable PRFs \cite{AC:GKWW21,EC:KitNis22}.
\end{remark}

\begin{definition}[Traceability] \label{def:tprf-tr}
For a TPRF $\TPRF$, consider the
experiment $\expa{\TPRF,\qA}{trace}(1^\secp)$ between a
challenger $\qCh$ and an adversary $\qA$. 

\begin{description}
\item $\underline{\expa{\TPRF,\qA}{trace}(1^\secp)}:$
\begin{enumerate}
\item $\qCh$ samples $\msk \la \Setup(1^\secp)$ and keys $\key,
\widetilde{\key} \la
\bit^\secp$ for a quantum-accessible PRF $\QPRF$.
\item $\qA$ sends $\id \in \bit$ to $\qCh$. $\qCh$ sends
$\sk \la \KG(\msk, \id)$ to $\qA$.

\item Consider the following distribution that takes a random tape $r$ as
input:

\begin{description}
\item $\underline{D[\key](r):}$
\begin{itemize}
\item Compute a pseudo-random bit $b$ and pseudo-random strings $x,
y_1$ using $\QPRF(\key, r)$.
\item Compute $y_0 \la \Eval(\msk, x)$.
\item Set $y \seteq y_b$ and output $(b, x, y)$.
\end{itemize}
\end{description}
\item $\qA$ is provided quantum access (on the random tape)
to the distributions $D[\key](\cdot)$ and
$D[\widetilde{\key}](\cdot)$.
\item $\qA$ outputs a quantum program $\qP^*$.
\end{enumerate}
\end{description}

Consider the following events:

\begin{description}
\item $\underline{\Live_{\epsilon}:}$
\begin{itemize}
\item Consider applying $\ProjImp(\cP_{D[\widetilde{\key}]})$ to
$\qP^*$ to get an outcome $p$, where $D[\widetilde{\key}]$ is the distribution
defined above.
\item The event is said to occur if $p \ge 1/2 + \epsilon$.
\end{itemize}

\item $\underline{\GoodTrace_{\epsilon}:}$ $\qTrace(\msk,
\qP^*, \epsilon)$ outputs $\id' \neq \bot$.
\item $\underline{\BadTrace_{\epsilon}:}$ $\qTrace(\msk,
\qP^*, \epsilon)$ outputs $\id' = 1 - \id$.
\end{description}

A TPRF scheme $\TPRF$ satisfies traceability if the following holds
for every $\qP^*$ output by every QPT $\qA$ and for every inverse polynomial $\epsilon$:

\begin{align}
    &\Pr[\GoodTrace_{\epsilon}] \ge \Pr[\Live_{\epsilon}] - \negl(\secp)\\
    &\Pr[\BadTrace_\epsilon] \le \negl(\secp)
\end{align}

\end{definition}

\end{definition}

The work of Kitagawa and Nishimaki \cite{EC:KitNis22} constructed a
quantum-secure watermarkable PRF (WMPRF) based on LWE with sub-exponential modulus. We show that their WMPRF is
also a TPRF as per our definition. Note that the difference is mainly
syntactic, except for the fact that our definition also provides the
adversary with quantum access to the distributions $D[\key]$ and
$D[\widetilde{\key}]$ on their random tapes.

\begin{theorem}\label{thm:tprf-lwe}
There exists a quantum-secure traceable PRF with identity space
$\bit$, based on the quantum hardness of LWE with sub-exponential
modulus.
\end{theorem}
\begin{proof}
We will show that the watermarkable PRF (WMPRF) based on LWE due to Kitagawa
and Nishimaki \cite{EC:KitNis22} implies a TPRF. A WMPRF $\WMPRF$ for
message space $\cM$, domain $\cX$ and range $\cY$ is a tuple of five
algorithms $(\Setup, \Gen, \Eval, \Mark, \qExtract)$.
$\Setup(1^\secp)$ outputs a public parameter $\pp$ and a secret
extraction-key $\xk$. $\Gen(\pp)$ outputs a PRF key $\prfk$ and a
public tag $\tau$. $\Eval(\prfk, x)$ outputs $y \in \cY$, which is
meant to be the PRF evaluation on input $x \in \cX$. $\Mark(\pp,
\prfk, \msg)$ outputs an evaluation circuit $\tlC$ that is marked with
$\msg \in \cM$. The quantum extraction algorithm $\qExtract(\xk, \tau,
\qC', \epsilon^*)$ takes $\xk, \tau$ as inputs along with a quantum
program $\qC'$ and a parameter $\epsilon^*$. It outputs $\msg' \in \cM
\cup \{\bot\}$. We now construct a TPRF $\TPRF = (\Setup, \KG, \Eval,
\qTrace)$ using a WMPRF $\WMPRF$ with message space $\bit$, domain
$\cX$ and range $\cY$ as follows:

\begin{description}
\item $\underline{\Setup(1^\secp):}$
\begin{itemize}
\item Execute $(\pp, \xk) \la \WMPRF.\Setup(1^\secp)$.
\item Execute $(\prfk, \tau) \la \WMPRF.\Gen(\pp)$.
\item Output $\msk \seteq (\pp, \xk, \prfk, \tau)$.
\end{itemize}

\item $\underline{\KG(\msk, \id):}$
\begin{itemize}
\item Parse $\msk = (\pp, \xk, \prfk, \tau)$.
\item Compute $\tlC \la \Mark(\pp, \prfk, \id)$.
\item Output $\sk \seteq \tlC$.
\end{itemize}

\item $\underline{\Eval(\sk', x):}$
\begin{itemize}
\item If $\sk'$ is of the form 
$\msk = (\pp, \xk, \prfk, \tau)$, 
output $y \la \WMPRF.\Eval\allowbreak(\prfk, x)$.
\item Otherwise, parse $\sk' = \widetilde{C}$ and output
$y = \widetilde{C}(x)$.
\end{itemize}

\item $\underline{\qTrace(\msk, \qP, \epsilon^*):}$
\begin{itemize}
\item Parse $\msk = (\pp, \xk, \prfk, \tau)$.
\item Output $\id' \la \qExtract(\xk, \tau, \qP, \epsilon^*)$.
\end{itemize}
\end{description}

It is easy to see that $\TPRF$ satisfies evaluation correctness and
pseudo-randomness by the analogous properties of evaluation
correctness and pseudo-randomness of $\WMPRF$. Consider now a
traceability notion that is similar to ours, except that the adversary
is not provided with access to $D[\key](\cdot),
D[\widetilde{\key}](\cdot)$. The unremovability notion of
\cite{EC:KitNis22} immediately implies that $\TPRF$ satisfies the
aforementioned traceability guarantee. This is because the difference
between these primitives and security notions is merely syntactic.

Next, we explain why providing access to the distributions $D[\key],
D[\widetilde{\key}]$ on the random tape does not give the adversary any additional
power. Notice that $D[\key](r)$ outputs samples of
the form $(0, x, \Eval(\msk, x))$ when $b=0$. Here, the value $x$ is
chosen to be pseudo-random, using $\QPRF(\key, r)$. However, by the
security of $\QPRF$ we can consider a computationally close
distribution $\widetilde{D}[\key]$ where $x$ is chosen at
random for each $r$. Consider now the distribution $D'[\key](r)$ that
is defined similar to $D[\key]$, except that it uses $\sk$ in place of
$\msk$. In other words, it outputs samples of the form $(0, x,
\Eval(\sk, x))$ when $b=0$. Once again, by the security of $\QPRF$, we
can consider a computationally close distribution
$\widetilde{D}'[\key]$ that samples $x$ at random. Observe now that
from the evaluation correctness property (Definition \ref{def:tprf}),
samples from $\widetilde{D}[\key]$ and $\widetilde{D}'[\key]$ are
statistically indistinguishable. We now recall the following lemma
(para-phrased) shown by Boneh and Zhandry:

\begin{lemma}\cite{EC:BonZha13}
Let $\cY$ and $\cZ$ be sets and for each $y \in \cY$, let $D_y$ and
$D'_y$ be distributions on $\cZ$ such that $\SD(D_y, D'_y) \le
\epsilon$.  Let $O:\cY \ra \cZ$ and $O':\cY \ra \cZ$ be functions
such that $O(y)$ outputs $z \gets D_y$ and $O'(y)$ outputs $z' \gets
D'_y$. Then, $O(y)$ and $O'(y)$ are $\epsilon'$-statistically
indistinguishable by quantum algorithms making $q$ superposition
oracle queries, such that $\epsilon' = \sqrt{8C_0q^3\epsilon}$ where
$C_0$ is a constant.
\end{lemma}

Using this lemma for the special case when $\cY$ is a singleton set,
we have that $\widetilde{D}[\key]$ and $\widetilde{D}'[\key]$ are
statistically indistinguishable with polynomially-many quantum
queries. This immediately gives us that $D[\key]$ and $D'[\key]$ are
computationally indistinguishable with polynomially-many quantum
queries. As a result, we can replace oracle access to $D[\key]$ with
$D'[\key]$. We can now apply the same argument to replace access to
$D[\widetilde{\key}]$ with $D'[\widetilde{\key}]$. Importantly, the
distributions $D'[\key]$ and $D'[\widetilde{\key}]$ do not provide any
additional power to the adversary, as it is already provided with
$\sk$ as part of the traceability game. Therefore, we can obtain
a $\TPRF$ as per our notion from the $\WMPRF$ of \cite{EC:KitNis22}.
\end{proof}

\fuyuki{I read Section 4 of \cite{TCC:Zhandry20} again and realized that we might need to be careful about the implication. Basically, the intuition is that they are equivalent since the adversary can check if a quantum program is good (live) or not by itself before outputting it. However, the condition is not generally satisfied in PRF setting. Testing goodness straightforwardly requires $\msk$ that the adversary does not have.}

\nikhil{I the changed definition to provide oracle access to 
WPRF-like distribution on random tape. Now the notions are equivalent.}

Next, we define the information-theoretic notion of fingerprinting
codes \cite{C:BonSha95}.

\begin{definition}[Fingerprinting Codes \cite{C:BonSha95}] A fingerprinting code 
$\FC$ consists of the following algorithms:

\begin{description}
\item $\Setup(1^\secp, N) \ra (\Gamma, \tk):$ The setup algorithm takes
a security parameter as input and an identity space size $N$. It outputs a codebook $\Gamma =
\{w_{\id}\}_{\id \in [N]}$ and a tracing key $\tk$. The
values $w_{\id}$ are called codewords, and are bit strings of length
$\ell$ (called the code-length).

\item $\tTrace(\tk, w^*) \ra \id^*:$ The tracing algorithm
takes the tracing key $\tk$ as input, along with a string $w^* \in
\bit^{\ell}$.
It outputs an identity $\id^* \in [N]$.

\end{description}

We define the traceability requirement of fingerprinting codes as
follows:

\begin{definition}[Traceability]\label{def:fc-tr} This notion is formalized
by the experiment $\expb{\FC,\qA}{fc}{trace}(1^\secp, N)$
between a challenger $\qCh$ and an adversary $\qA$. For $W \subseteq
[N]$, let the feasible set $F(W)$ be the set of all words $w \in
\bit^\ell$ satisfying the following:

\begin{itemize}
\item For each $i \in [\ell]$, there exists $\id \in W$ such that for
$w_\id \in \Gamma$, it holds that $w_{\id}[i] = w[i]$.
\end{itemize}

The experiment is defined as follows:

\begin{description}
\item $\underline{\expb{\FC,\qA}{fc}{trace}(1^\secp, N)}:$
\begin{enumerate}
\item $\qCh$ samples $(\Gamma, \tk) \la \Setup(1^\secp, N)$.
\item $\qA$ is allowed to make adaptive queries of the following
form: For a query input $\id \in [N]$, $\qCh$ responds with the codeword
$w_{\id} \in \Gamma$. Let $W \subseteq [N]$ be the set of queries
made by $\qA$.
\item $\qA$ outputs a string $w^*$. If $w^* \notin F(W)$, $\qCh$
 outputs $\bot$.
\item $\qCh$ runs $\id^* \la \tTrace(\tk, w^*)$. If
$\id^* \notin W$, $\qCh$ outputs $\top$. Else, it outputs $\bot$.
\end{enumerate}
\end{description}

A fingerprinting code $\FC$ satisfies traceability if the following holds
for every $\qA$ and $N = \poly(\secp)$:

$$\Pr[\expb{\FC,\qA}{fc}{trace}(1^\secp, N) \ra \top] \le \negl(\secp)$$

\end{definition}

\end{definition}

Next, we define our notion of a multi-level traceable PRF (MLT-PRF)
followed by providing its construction.

\subsection{Multi-Level Traceable PRF}\label{sec:mlt-prf-const}

\begin{definition}[Multi-Level Traceable PRF]\label{def:mlt-prf}
A multi-level traceable pseudo-random function (MLT-PRF) is an MLTT
scheme for the following application $(\qF, \qE, t)$. Let MLT-PRF
have domain $\cX$ and range $\cY$, and consist of the algorithms
$(\Setup, \KG, \Eval, \qTrace)$ as per the syntax of MLTT (Definition
\ref{def:mltt-cor}). Note that $\Setup$ samples a PRF key
$\prfk$, followed by setting $f \seteq \PRF(\prfk, \cdot)$ and $\aux_f \seteq
\bot$.

\begin{description}
\item $\underline{\qF(\qg, f, r)}:$
\begin{itemize}
\item Sample $x \la \cX$ where $\cX$ is the domain of the PRF $f$, using
the random tape $r$.
\item If $\qg(x) = f(x)$, output $1$. Else, output $0$.
\end{itemize}

\item $\underline{\qE(f, \qP, (r, \key))}$:
\begin{itemize}
\item Compute a pseudo-random bit $b$ and pseudo-random strings $x,
y_1$
using $\QPRF(\key, r)$ where $\QPRF$ is a quantum-accessible PRF.
\footnote{Note that $x, y_1$  must be computationally indistinguishable from values drawn
uniformly from $\cX, \cY$ respectively.}

\item Compute $y_0 \seteq f(x)$.
\item Set $y \seteq y_b$ and run $b' \la \qP(x, y)$.

\item Output $1$ if $b = b'$ and $0$ otherwise.
\end{itemize}

\item $\underline{t} \seteq \frac12$
\end{description}

\begin{remark}
Notice that $\qF$ accepts a quantum program $\qg$ as input, even
thought $\MLTPRF$ is a classical primitive. This is because such an
$(\qF, \qE, t)$ captures the SKL scenario, and the syntax of $\MLTT$
anyway restricts $\Eval$ to be classical.
\end{remark}

\begin{remark}
It is more natural to define $\qE$ which samples $b, x$ truly at
random directly using the random string $r$. However, we use the
current notion due to some technicalities in the proof of MLT-PRF.
Notice however, that for the end goal of PRF-SKL, both the definitions
of $\qE$ are equivalent, assuming that $\QPRF$ is a quantum-accessible
PRF. This is because the underlying distributions are computationally
indistinguishable. Consequently, Theorem \ref{thm:api} and Theorem
\ref{thm:pi-ind} ensure that the corresponding $\API$ measurements of
the $\epsilon$-good tests produce close outcomes.
\end{remark}

In addition to the properties of Traceability (Definition
\ref{def:mltt-tr}), Evaluation Correctness (Definition
\ref{def:mltt-cor}) and Deterministic Evaluation (Definition
\ref{def:mltt-cor}) which any MLTT scheme must satisfy, an MLT-PRF
must also satisfy pseudo-randomness:

\paragraph{Pseudo-randomness:} Let $\cR$ denote the set of all functions
with domain $\cX$ and range $\cY$. The following holds for all QPT
adversaries $\qA$, $N = \poly(\secp)$ and $k = \poly(\secp)$, for $g,
f$ sampled as $g \la \cR$ and $(\msk, f, \aux_f) \la \Setup(1^\secp,
N, k)$ respectively, the following holds:

\[
\bigg\lvert
\Pr\left[
1 \la \qA^{g(\cdot)}(1^\secp)
\right] 
-
\Pr\left[
1 \la \qA^{f(\cdot)}(1^\secp)
\right] 
\bigg\rvert
\le \negl(\secp)
\]
\end{definition}

\paragraph{Construction:} We will now construct an MLT-PRF $\MLTPRF$
using a TPRF $\TPRF = \TPRF.(\Setup, \KG, \allowbreak\Eval,
\qTrace)$ with identity space $\bit$, and a fingerprinting code $\FC = \FC.(\Setup, \tTrace)$ as follows:

\begin{description}
\item $\underline{\Setup(1^\secp, N, k):}$
\begin{itemize}
\item For $i \in [k]$, compute $(\fc.\Gamma_i, \fc.\tk_i) \la
\FC.\Setup(1^\secp, N)$.
\item For $(i,j) \in [k]\times[\ell]$, compute $\tprf.\msk_{i}^j \la
\TPRF.\Setup(1^\secp)$.
\item Set $\msk \seteq
    (\{\tprf.\msk_i^j\}_{(i,j)\in[k]\times[\ell]},
    \{\fc.\Gamma_i\}_{i\in[k]},
    \{\fc.\tk_i\}_{i\in[k]})$.
\item Set $f \seteq \bigoplus_{(i,j) \in [k] \times [\ell]}
\TPRF.\Eval(\tprf.\msk_i^j, \cdot)$ and $\aux_f \seteq \bot$.
\item Output $(\msk, f, \aux_f)$.
\end{itemize}
\item $\underline{\KG(\msk, i, \id):}$
\begin{itemize}
\item Parse $\msk =
    (\{\tprf.\msk_i^j\}_{(i,j)\in[k]\times[\ell]},
    \{\fc.\Gamma_i\}_{i\in[k]},
    \{\fc.\tk_i\}_{i\in[k]})$.
\item Parse $\fc.\Gamma_i = (w_{s})_{s \in [N]}$.
\item For $j \in [\ell]$, compute $\tprf.\sk_i^j \la \TPRF.\KG(\tprf.\msk_i^j,
    w_{\id}[j])$.
\item Output $\sk_i \seteq (\tprf.\sk_i^j)_{j \in [\ell]}$.
\end{itemize}
\item $\underline{\Eval(\sk_1, \ldots, \sk_k, x)}:$
\begin{itemize}
\item For each $i \in [k]$, parse $\sk_i = (\tprf.\sk_i^j)_{j\in[\ell]}$.
\item For each $(i, j) \in [k]\times[\ell]$, compute $y_i^j \la 
\TPRF.\Eval(\tprf.\sk_i^j, x)$.
\item Output $y \seteq \bigoplus_{(i,j)\in[k]\times[\ell]}y_i^j$.
\end{itemize}
\item $\underline{\qTrace(\msk, \qP^*, \epsilon^*)}:$

\begin{enumerate}
\item Parse $\msk =
    (\{\tprf.\msk_i^j\}_{(i,j)\in[k]\times[\ell]},
    \{\fc.\Gamma_i\}_{i\in[k]},
    \{\fc.\tk_i\}_{i\in[k]})$.

%\item For each $(i, j) \in [k] \times [\ell]$, consider the
%distribution $\cD_b^{i,j} \seteq \cD_b[\TPRF.\Eval(\tprf.\msk_i^j, \cdot)]$
%where $\cD_b$ is as defined in Definition \ref{}.

\item For each $(i, j) \in [k] \times [\ell]$, consider the
following algorithm $\qB_{i,j}$ that is provided with
$\qP^*$, and takes $(x, y)$ as input:

\begin{description}
    \item $\underline{\qB_{i,j}[\qP^*](x, y)}:$
\begin{itemize}
\item Compute $y' = y \xor \bigoplus_{(u,v) \neq
    (i,j)}\TPRF.\Eval(\tprf.\msk_u^v, x)$.

\item Execute $\qP^*(x, y')$.

\item Output $b'$, which is the value output by $\qP^*$.
\end{itemize}
\end{description}

\item Sample a key $\widetilde{\key} \la \bit^\secp$ for a quantum-accessible PRF
$\QPRF$.

\item For $(i, j) = (1,1)$ to $(k, \ell)$, let $\counts_i^j \seteq
    (i-1)\ell +
j$ and do the following:
\begin{enumerate}

\item Let $D[\widetilde{\key}]$ be the following distribution:

\begin{description}
    \item $\underline{D[\widetilde{\key}](r):}$
\begin{itemize}
\item Compute a pseudo-random bit $b$ and pseudo-random strings $x,
y_1$ using $\QPRF(\widetilde{\key}, r)$.
\item Compute $y_0 \la f(x)$.
\item Set $y \seteq y_b$ and output $(b, x, y)$.
\end{itemize}
\end{description}

Let $\{\cP_{b, x, y}\}_{b,x,y}$ denote the set of projective
measurements corresponding to running $\qP^*$ on input $(x, y)$
and outputting $1$ if its output $b' = b$ (outputting $0$ otherwise).
Apply $\Est \seteq \API^{\cP, D[\widetilde{\key}]}_{\epsilon', \delta}$ to $\qP^*$, where
$\epsilon' = 7\epsilon^*/9 \times 1/4k\ell$ and $\delta = 2^{-\secp}$.

\item Compute $w_i^j \la \TPRF.\qTrace(\tprf.\msk_i^j,
\qB_{i,j}[\qP^*](\cdot,\cdot), \widetilde{\epsilon})$ where the
parameter
$\widetilde{\epsilon} =
6\epsilon^*/9 - 2\epsilon'\counts_i^j$. Note that $\TPRF.\qTrace$
performs some measurement on $\qB_{i,j}$, and hence on $\qP^*$.
Consider the equivalent projective measurement $\widetilde{\Est}$
performed on $\qP^*$ that outputs $w_i^j$ (This can be achieved by
dilating the measurement performed by $\qTrace$ using a new ancilla
register each time).

\item Run the algorithm $\Repair^{\Est, \widetilde{\Est}}_{T, p,
s}$ on $\qP^*$ where $T = \frac{1}{\sqrt{\delta}}$ and
$s \seteq w_i^j$.
\end{enumerate}

\item For each $i \in [k]$, set $w_i^* \seteq w_i^1 \| \ldots \|
    w_i^{\ell}$. Compute $\id_i^* \la \FC.\tTrace(\fc.\tk_i, w_i^*)$.
\item Output $(\id_1^*, \ldots, \id_k^*)$.
\end{enumerate}
\end{description}

\paragraph{Evaluation Correctness:} This follows immediately from the
evaluation correctness of $\TPRF$ (Definition \ref{def:tprf}) and the
description of the $\Eval$ algorithm.

\paragraph{Deterministic Evaluation:} This follows from evaluation
correctness, which implies deterministic evaluation in the context of
PRFs.

\paragraph{Pseudorandomness:} Pseudorandomness also follows directly
from the pseudorandomness of $\TPRF$ (Definition \ref{def:tprf}) as $\Eval$
computes the XOR of the evaluations of $\TPRF$ instances.

\begin{theorem}
The MLT-PRF $\MLTPRF$ satisfies traceability (Definition
\ref{def:mltt-tr}), assuming $\QPRF$ is a quantum-accessible PRF,
$\TPRF$ satisfies traceability
(Definition \ref{def:tprf-tr}), and $\FC$ satisfies traceability
(Definition \ref{def:fc-tr}).
\end{theorem}

Since $\FC$ is known information-theoretically, $\QPRF$ is known from
OWFs, and $\TPRF$ is known
from LWE (Theorem \ref{thm:tprf-lwe}), we have the following:

\begin{corollary}\label{cor:mlt-prf-lwe}
$\MLTPRF$ satisfies traceability (Definition \ref{def:mltt-tr}), based on the
quantum hardness of LWE with sub-exponential modulus.
\end{corollary}

\begin{proof}
Assume for the sake of contradiction that there exists QPT $\qA$ that
breaks the traceability of $\MLTPRF$. Let $\APILive_\epsilon$ denote
the event that the $\epsilon$-good test wrt $(f, \qE, 1/2)$ passes
when applied to $\qP^*$ output by $\qA$ in the MLTT experiment. Let
$\GoodExt_\epsilon$
denote the event that $(\id_1^*, \ldots, \id_k^*)$ output by
$\qTrace(\msk, \qP^*, 0.9\epsilon)$ belongs to $Q_1 \times \ldots
\times Q_k$ where $\{Q_i\}_{i\in[k]}$ are the multi-sets defined in Definition
\ref{def:mltt-tr}. By assumption, we have $\Pr[\lnot \GoodExt_\epsilon \land
\APILive_\epsilon] = \nonnegl(\secp)$. This means
$\Pr[\APILive_\epsilon] = \nonnegl(\secp)$ and
$\Pr[\lnot \GoodExt_\epsilon \mid \APILive_\epsilon] =
\nonnegl(\secp)$.

Consider now the event $\BadCode_\epsilon$ that occurs when there
exists $s \in [k]$ such that $w_s^* \notin F(Q_s)$. Here, $w_s^*$
refers to the codeword computed by $\qTrace(\msk, \qP^*, 0.9\epsilon)$
and $F(Q_s)$ refers to the feasible set (Definition \ref{def:fc-tr})
of $Q_s$, wrt codebook $\Gamma_s$. Assume for the sake of
contradiction that $\Pr[\BadCode_\epsilon \land \APILive_\epsilon] =
\nonnegl(\secp)$. We will then construct the following reduction $\qR$
that breaks the security of the underlying TPRF $\TPRF$.

\begin{description}
\item \underline{Execution of $\qR^{\qA}$ in
    $\expa{\TPRF,\qR}{trace}(1^\secp):$}
\begin{enumerate}
\item $\qCh$ samples $\tprf.\msk \la \TPRF.\Setup(1^\secp)$ and
    $\QPRF$ keys $\key, \widetilde{\key} \la \bit^\secp$.
\item $\qR$ samples $\beta \la \bit$ and sends $\beta$ to $\qCh$. $\qCh$
sends $\tprf.\sk \la \TPRF.\KG(\allowbreak\tprf.\msk, \beta)$ to $\qR$.
\item $\qCh$ provides $\qR$ with quantum access to 
$\tprf.D[\key](\cdot), \tprf.D[\widetilde{\key}](\cdot)$ as in Definition \ref{def:tprf-tr}.
\item $\qR$ chooses $(c,d) \la [k] \times [\ell]$. It computes $\msk$
as in $\Setup(1^\secp, N, k)$, except that it sets $\tprf.\msk_c^d \seteq
\bot$. It then invokes $\qA$.

\item $\qA$ sends $q \in [N-1]$ to $\qR$.
\item For each $i \in [k] \times [q]$, $\qR$ samples $\id_i^j \la
[N]$. For each $i \in [k] \setminus \{c\}$, $\qR$ computes
$\{\sk_i^j\}_{j\in[q]}$ as in $\KG(\msk, i,
\id_i^j)$.
If for some $\id \in \{\id_c^j\}_{j\in[q]}$, it holds
that $w_\id[d] \neq \beta$ , then $\qR$ outputs
$\bot$. Otherwise, $\qR$ computes $\sk_c^j$ for each $j \in [q]$ as in
$\KG(\msk, c, \id_c^j)$, but using $\tprf.\sk$ instead of $\tprf.\sk_c^d$. Finally, $\qR$ sends $\{\id_i^j,
\sk_i^j\}_{(i,j) \in [k] \times [q]}$ to $\qA$. Define the
multi-sets $\{Q_i\}_{i\in[k]}$ where $Q_i \seteq
\{\id_i^j\}_{j\in[q]}$.

\item $\qA$ outputs a quantum program $\qP^*$.

\item $\qR$ applies the $\epsilon$-good test wrt $(f, \qE, 1/2)$ to
$\qP^*$ using quantum access to the distribution $\tprf.D[\key]$, where 
where $f$ is defined as $f(\cdot)
\seteq \TPRF.\Eval(\allowbreak\tprf.\msk, \cdot) \xor \bigoplus_{(i,j)\neq(c,d)}
\TPRF.\Eval(\tprf.\msk_i^j, \cdot)$. If the test fails, $\qR$ outputs
$\bot$.

\item $\qR$ executes Step 4. of $\qTrace(\msk, \qP^*, 0.9\epsilon)$ until just
before the iteration $(c, d)$. Then, it applies Step 4. (a) once.
These operations are performed using quantum access to
the distribution $\tprf.D[\widetilde{\key}]$.

\item Finally, $\qR$ outputs the program
$\qB^* \seteq \qB_{c,d}[\qP^*](\cdot,\cdot)$, where $\qB_{c, d}$ is as
defined in $\qTrace$.
\end{enumerate}
\end{description}

First, we note that in Step 8., $\qR$ applies the same $\epsilon$-good
test as the challenger would in the MLTT experiment. This is because
even though it doesn't have access to $\tprf.\msk$, it is provided
with quantum access to the distribution $\tprf.D[\key](\cdot)$ on its random
tape. Recall that this allows $\qR$ to obtain samples of the form
$(0, x, \TPRF.\Eval(\tprf.\msk, x))$ (in superposition), using which
it can compute samples of the form $(0, x, \Eval(\msk, x))$ (in
superposition). In other words, $\qR$ now has quantum access to the
distribution (on the random tape) that samples $(b, x, y)$  as in the
security predicate $\qE$ of MLT-PRF (Definition \ref{def:mlt-prf}).
Note that quantum access to this distribution is sufficient for $\qR$
to perform the corresponding $\API$ measurement, as apparent from the
procedure for $\API$ in \cite{TCC:Zhandry20}.

By a similar argument, we observe that in Step 9., $\qR$ performs
measurements that are equivalent to the corresponding ones from the
$\qTrace(\msk, \qP^*, 0.9\epsilon)$ algorithm, using access to the
distribution $\tprf.D[\widetilde{\key}]$. Note that the $\Repair$
procedure specified in $\qTrace$ can be performed as well, as it can
be performed with oracle access to the $\API$ measurement, which is
enabled by access to $\tprf.D[\widetilde{\key}]$.

Let $\BadBit_\epsilon'$ denote
the event that $w_c^d$ computed by the TPRF challenger in the
execution with
$\qR$ (in a way similar to $\qTrace(\msk,
\qP^*,\allowbreak
0.9\epsilon)$) satisfies $w_c^d \neq \beta$. Consider the
event $\BadBit_\epsilon$ that is defined analogous to
$\BadBit_\epsilon'$, except that it corresponds to the execution in
the MLTT experiment. Note that even though there is no $\qR$ in the
MLTT game, we can consider a hypothetical $\qR$ that simply guesses
$(c, d)$ uniformly in the MLTT game. Since the guess for $(c, d)$ is
uniform, we have that $\Pr[\BadBit_\epsilon \land \APILive_\epsilon] =
\frac{1}{k\ell}\cdot\Pr[\BadCode_\epsilon \land \APILive_\epsilon] =
\nonnegl(\secp)$.

Consider now the events $\BadCode_\epsilon', \APILive_\epsilon'$ that are analogous to 
$\BadCode_\epsilon, \allowbreak \APILive_\epsilon$, but are
defined wrt the execution of $\qR$. Let $\NoAbort'$ denote the event
that $\qR$ does not abort in Step 6. We begin by proving the following claim:

\begin{claim}\label{claim:noabort}
\begin{align}
\Pr[\BadBit_\epsilon \land \APILive_\epsilon] &= \nonnegl(\secp)\\
&\implies\\
\Pr[\NoAbort' \land \BadBit'_\epsilon \land \APILive_\epsilon'] &=
\nonnegl(\secp)
\end{align}
\end{claim}

\begin{proof}
Consider the set $Q_c =
\{\id_c^j\}_{j\in[q]}$.
We first argue that
there exists some $u \in [\ell]$ s.t. $\id_c^j[u] = \id_c^1[u]$ for each $j
\in [q]$ with probability $1 - \negl(\secp)$. Suppose not. Then, with
some probability $\nonnegl(\secp)$, it holds that every binary string
lies in $F(Q_c)$, which is the feasible set of $Q_c$. Consider now
an adversary $\qD$ that participates in the traceability game for the
fingerprinting code $\FC$. The adversary $\qD$ simply makes $q$ uniformly
random identity queries, where $q \in [N-1]$. Then, it outputs a
random binary string $w^*$ as its codeword. Clearly, conditioned on
the above bad event (which happens with $\nonnegl(\secp)$
probability), $w^*$ provides no information and yet lies in the
feasible set. As a result, tracing it fails with some non-negligible
probability, thereby breaking the traceability of the fingerprinting
code $\FC$.

Next, observe that the reduction $\qR$ guesses the index $d$ and its
value $\beta$ uniformly at random. Consequently, $\qR$ doesn't abort
with probability $(1 - \negl(\secp))\cdot(1/2\ell)$, i.e., the event
$\NoAbort'$ occurs with $\nonnegl(\secp)$ probability. Now, we define
an event $\NoAbort$ similar to $\NoAbort'$, but corresponding to the
MLTT experiment. Note that even though there is no $\qR$ in the MLTT
game, we can consider a hypothetical abort check being performed
wrt the identities sampled by the MLTT challenger. Based on the fact
that conditioned on $\NoAbort'$, the view of $\qA$ is identical to its
view in the MLTT experiment, we have that
$\Pr[\BadBit_\epsilon' \land \APILive_\epsilon' \mid \NoAbort'] =
\Pr[\BadBit_\epsilon \land \APILive_\epsilon \mid \NoAbort]$. 

Observe
that the events $\NoAbort$ and $\BadBit_\epsilon \land
\APILive_\epsilon$ are independent, as the occurrence of the former
only depends on whether the correct ``slot'' is guessed. Hence, we
have $\Pr[\BadBit_\epsilon' \land \APILive_\epsilon' \land \NoAbort']
= \Pr[\NoAbort']\cdot\Pr[\BadBit_\epsilon \land \APILive_\epsilon] =
\nonnegl(\secp)\cdot\Pr[\BadBit_\epsilon \land \APILive_\epsilon]$.
\end{proof}

We have from the above claim that $\Pr[\BadBit'_\epsilon \land
\APILive'_\epsilon \land \NoAbort'] = \nonnegl(\secp)$, due to our
assumption that $\Pr[\BadBit_\epsilon \land \APILive_\epsilon] =
\nonnegl(\secp)$.

Let $\cP' =
\{P'_{b,x,y}\}_{b,x,y}$ be the set of projective measurements
corresponding to running $\qB^*$ on input $x$ and outputting $1$ if
$\qB^*$ outputs $b$ (outputting $0$ otherwise). Let $\epsilon_c^d
\seteq 0.6\epsilon - 2\times0.7\epsilon\cdot\counts_c^d/4k\ell$ and
let $\alpha$ be a placeholder for $\epsilon_c^d$. Consider the event
$\Live_\alpha$ that corresponds to the measurement
$\ProjImp(\cP'_{D'})$ applied to $\qB^*$ outputting a probability
greater than $1/2 + \alpha$, where $D' \seteq \tprf.D[\widetilde{\key}]$.
We now prove that conditioned on
$\NoAbort' \land \APILive'_\epsilon$, $\Live_\alpha$ occurs with probability close to $1$.

\begin{remark}
In the rest of the proof, when we refer to some probability $p$ output by
a measurement and claim that $p > m$ for some $m$, we mean to say this
holds with $1 - \negl(\secp)$ probability.
\end{remark}

\begin{claim}
$\Pr[\Live_\alpha = 1 \mid \NoAbort' \land \APILive'_\epsilon] \ge 1 -
\negl(\secp)$
\end{claim}
\begin{proof}
Let $\cP = \{P_{b,x,y}\}_{b,x,y}$ be the set of projective
measurements corresponding to running $\qP^*$ on input $x$ and
outputting $1$ if $\qP^*$ outputs $b$ (outputting $0$ otherwise). Let
$D[\key]$ be the distribution that samples $b, x, y_1$ as
pseudo-random values using $\QPRF(\key, r)$ and outputs
$(b, x, y)$, where $y \seteq y_b$ and $y_0 \la f(x)$. Let $D \seteq
D[\key]$.
Notice that if $\ProjImp(\cP_D)$ is performed after Step 8, the
outcome would be greater than $1/2 + 0.8\epsilon$ by Theorem \ref{thm:api},
since the $\epsilon$-good test (Definition \ref{def:good}) is
performed in Step 8. Next, consider the distribution
$D[\widetilde{\key}]$ that is defined similar to $D[\key]$ but using
$\widetilde{\key}$ sampled by the TPRF challenger. Let $\widetilde{D}
\seteq D[\widetilde{\key}]$.  Observe that if
$\ProjImp(\cP_{\widetilde{D}})$ is applied after Step 8., the output
would be greater than $1/2 + 0.7\epsilon$ due to the fact that $D,
\widetilde{D}$ are computationally indistinguishable (by security of
$\QPRF$), and by Theorem \ref{thm:pi-ind} (with parameter $\epsilon' =
0.1\epsilon$).

Consider now the probability $p^{(1,1)}$ output by the $\API$
measurement in iteration $(1,1)$ of Step 4. (a) performed in Step 9.
of the reduction. Notice that this measurement is just the $\API$
analogue of the measurement $\ProjImp(\cP_{\widetilde{D}})$. Hence, by
applying Theorem \ref{thm:api}, we get
$p^{(1,1)} \ge 1/2 + 0.7\epsilon - 0.7\epsilon/4k\ell$. We can ignore the
outcome of the measurement in Step 4. (b). Since $\Repair$ is
performed in Step 4. (c), by Theorem \ref{thm:repair}, we have that
$p^{(1,2)} \ge 1/2 + 0.7 \epsilon - 0.7\epsilon/4k\ell -
2\times0.7\epsilon/4k\ell$, where $p^{(1,2)}$ denotes the output of the
$\API$ measurement in iteration $(1,2)$ performed in Step 9. of the
reduction. By an inductive argument, we have that $p^{(c,d)} \ge 1/2 +
0.7\epsilon - 0.7\epsilon/4k\ell -
2\times0.7\epsilon(\counts_c^d - 1)/4k\ell$, where $\counts_c^d =
(c-1)\ell + d$. Now, we need to argue about the probability output by
$\ProjImp(\cP'_{D'})$ after Step 10. Consider first the distribution
$D''$ that samples $(b, x, y) \la D'$ and outputs $(b, x, y \xor
\bigoplus_{(i,j)\neq (c,d)}\TPRF.\Eval(\tprf.\msk_i^j, x))$. Notice
that $\ProjImp(\cP'_{D'})$ is equivalent to the measurement
$\ProjImp(\cP_{D''})$, since running $\qB^*$ on outputs of $D'$ is
equivalent to running $\qP^*$ on outputs of $D''$. Furthermore, notice
that $D''$ and $\widetilde{D}$ are computationally indistinguishable.
This means $\ProjImp(\cP_{D''})$ and $\ProjImp(\cP_{\widetilde{D}})$
produce outcomes which are $0.1\epsilon$ close, by Theorem
\ref{thm:pi-ind} (with parameter $\epsilon' = 0.1\epsilon$). Since
applying $\ProjImp(\cP_{\widetilde{D}})$ to $\qP^*$ after Step 9.
gives a probability greater than $1/2 + 0.7\epsilon -
2\times0.7\epsilon\cdot\counts_c^d/4k\ell$ (by Theorem \ref{thm:api}), applying
$\ProjImp(\cP'_{D'})$ to $\qB^*$ outputs $p \ge 1/2 + 0.6\epsilon -
2\times0.7\epsilon\cdot\counts_c^d/4k\ell$. Therefore, $p \ge 1/2 +
\alpha$ with overwhelming probability.
\end{proof}

Let $\Event$ be the event $\NoAbort' \land
\APILive_\epsilon'$. From
Theorem \ref{thm:tprf-lwe}, we have that $\Pr[\GoodTrace_\alpha \mid
\Event] \ge \Pr[\Live_\alpha \mid \Event] - \negl(\secp)$ and
$\Pr[\BadTrace_\alpha \mid \Event] \le \negl(\secp)$, since $\alpha =
0.6\epsilon - 2\times 0.7\epsilon\cdot \counts_c^d/4k\ell \ge
0.6\epsilon - 0.35\epsilon = 0.25\epsilon = \poly(\secp)$. Note that
this is because $\qR$ itself performs $\APILive_\epsilon'$ and that
the probabilities in the statement of Theorem \ref{thm:tprf-lwe}
consider only the randomness generated by measurements on the pirate
program, and not the randomness of the adversary $\qA$. Consequently,
we have $\Pr[\GoodTrace_\alpha \mid \Event] \ge 1 - \negl(\secp)$ and
$\Pr[\BadTrace_\alpha \mid \Event] \le \negl(\secp)$, which
contradicts $\Pr[\BadBit'_\epsilon \land \Event] = \nonnegl(\secp)$.
This is because $\BadBit'_\epsilon$ corresponds to $w_c^d \neq \beta$,
where $w_c^d \la \TPRF.\qTrace(\tprf.\msk, \qB^*, \alpha)$. As a
result, it cannot be possible that $\Pr[\BadCode_\epsilon \land
\APILive_\epsilon] = \nonnegl(\secp)$.

Hence, we can now move on
to the case when $\Pr[\BadID_\epsilon \land \APILive_\epsilon] =
\nonnegl(\secp)$, where $\BadID_\epsilon$ is the event that there
exists some $s \in [k]$ such that $w_s^* \in F(Q_s)$ but $\id_s^*
\notin Q_s$, in the MLTT game. In this case, we can break the security of the fingerprinting code
$\FC$ using the following reduction $\qS$:

\begin{description}
\item
\underline{Execution of $\qS$ in Experiment
$\expb{\FC,\qS}{fc}{trace}(1^\secp, N):$}
\begin{enumerate}
\item $\qCh$ samples $(\Gamma, \tk) \la \FC.\Setup(1^\secp, N)$.
\item $\qS$ samples $(\msk, f, \aux_f) \la \Setup(1^\secp, N, k)$ and sends
$\aux_f$ to $\qA$.
\item $\qS$ samples $s \la [k]$.
\item For each $i \in [k] \times [q]$, $\qS$ samples $\id_i^j \la
[N]$. For $i \neq s$, $\qS$ computes $\sk_i^j \la \KG(\msk, i,
\id_i^j)$ for each $j\in[q]$. For $i = s$, $\qS$ computes $\sk_s^j$
for each $j\in[q]$ as in $\KG$, but by first querying $\id_s^1,
\ldots, \id_s^j$ to $\qCh$. $\qS$ sends
$\{\sk_i^j\}_{(i,j)\in[k]\times[q])}$ to $\qA$.
Let $W \seteq \{\id_s^j\}_{j\in[q]}$.

\item $\qA$ outputs a quantum program $\qP^*$.

\item $\qS$ applies the $\epsilon$-good test wrt $(f, \qE, 1/2)$ to
$\qP^*$. Then, $\qS$ applies $\qTrace(\msk, \qP^*,\allowbreak 0.9\epsilon)$ until
the beginning of Step 5. of the algorithm to obtain $w_1^*, \ldots,
w_k^*$.

\item $\qS$ sends $w_s^*$ to $\qCh$. If $w_s^* \notin F(W)$, $\qCh$
outputs $\bot$.

\item $\qCh$ runs $\id^* \la \FC.\tTrace(\tk, w^*_s)$. If 
$\id^* \notin W$, $\qCh$ outputs $\top$. Else, it outputs $\bot$.
\end{enumerate}
\end{description}

It is easy to see that the view of $\qA$ in the reduction is identical
to its view in the MLT-PRF game. By assumption, we have
$\Pr[\BadID_\epsilon] = \nonnegl(\secp)$. Since $\qS$ guesses an index
uniformly at random, it wins the game with probability
$1/k\cdot(\nonnegl(\secp)) = \nonnegl(\secp)$. This breaks the
traceability of $\FC$ (Definition \ref{def:fc-tr}), a contradiction.
Hence, it cannot be possible that $\Pr[\BadID_\epsilon \land
\APILive_\epsilon] = \nonnegl(\secp)$.

Let us summarize what we showed until now. We started by assuming that
$\Pr[\lnot \GoodExt_\epsilon \land \APILive_\epsilon] \ge
\nonnegl(\secp)$. Notice that $\lnot \GoodExt_\epsilon$ only occurs if
either $\BadCode_\epsilon$ occurs, or if $\BadID_\epsilon$ occurs.
Hence, we have two cases: either $\Pr[\BadCode_\epsilon \land
\APILive_\epsilon] = \nonnegl(\secp)$ or $\Pr[\BadID_\epsilon \land
\APILive_\epsilon] = \nonnegl(\secp)$. We showed that when the former
is true, we arrive at a contradiction by utilizing the security of
$\TPRF$. When the latter condition is true, we arrive at a
contradiction based on the traceability of $\FC$. Consequently, it
must be that $\Pr[\lnot \GoodExt_\epsilon \land \APILive_\epsilon] \le
\negl(\secp)$, i.e., $\MLTPRF$ satisfies traceability (Definition
\ref{def:mltt-tr}).
\end{proof}

\subsection{Bounded Collusion-Resistant PRF-SKL}\label{sec:bcr-prf-skl}

Using our compiler from Section \ref{sec:const} and the
MLT-PRF scheme from Section \ref{sec:mlt-prf-const}, we obtain a bounded
collusion-resistant PRF-SKL scheme. First, we define a PRF-SKL scheme
to be an SKL scheme (Definition \ref{def:skl}) with algorithms
$(\Setup, \qKG, \qEval, \qDel, \Vrfy)$ 
for the following
cryptographic application $(\qF, \widetilde{\qE}, t)$. Note that
$\Setup$ samples a PRF key $\prfk$, followed by setting $f
\seteq \PRF(\prfk, \cdot)$ and $\aux_f \seteq \bot$.

\begin{description}
\item $\underline{\qF(\qg, f, r)}:$
\begin{itemize}
\item Sample $x \la \cX$ where $\cX$ is the domain of the PRF $f$, using
the random tape $r$.
\item If $\qg(x) = f(x)$, output $1$. Else, output $0$.
\end{itemize}

\item $\underline{\widetilde{\qE}(f, \qP, (r, \key))}$:
\begin{itemize}
\item Ignore $\key$.
\item Sample $x \la \cX$ where $\cX$ is the domain of the PRF $f$,
using the random tape $r$. Sample $b \la \bit$ also using $r$.
\item If $b=0$, compute $y \seteq f(x)$. Otherwise, sample $y \la
\cY$ using $r$, where $\cY$ is the range of the PRF $f$.

\item Run $b' \la \qP(x, y)$.
\item Output $1$ if $b = b'$ and $0$ otherwise.
\end{itemize}

\item $\underline{t} \seteq \frac12$
\end{description}

In Corollary \ref{cor:mlt-prf-lwe}, we obtained an MLT-PRF, which is
an MLTT scheme for the application $(\qF, \qE, t)$ defined in
\ref{def:mlt-prf}. Using Theorem \ref{thm:const}, we obtain an SKL
scheme for the functionality $(\qF, \qE, t)$. Importantly, this SKL
scheme immediately implies an SKL scheme for the functionality $(\qF,
\widetilde{\qE}, t)$. This can be seen from the fact that the values
$(b, x, y)$ which determine the success probability of $\qP$, are
sampled in $\qE, \widetilde{\qE}$ from distributions which are
computationally indistinguishable. Note that this holds because of the
security of the quantum-accessible PRF $\QPRF$ utilized by $\qE$.
Consequently, the corresponding API measurements must produce similar
outcomes based on Theorem \ref{thm:api} and Theorem \ref{thm:pi-ind}.
As a result, we have the following theorem:

\begin{theorem}\label{thm:bcr-prf-skl_lwe}
Assuming the quantum hardness of LWE with sub-exponential modulus, for
every collusion-bound $q = \poly(\secp)$, there exists a PRF-SKL
scheme satisfying $q$-bounded standard-KLA security (Definition
\ref{def:std-kla}).
\end{theorem}

\begin{remark}
Observe that the predicate $\qE$ only tests $\qP$ using a single
challenge input $(x, y)$. However, weak pseudo-randomness security
should provide polynomially many challenge inputs to $\qP$. Note that
if $\qP$ is given oracle access to $\Dwprf$ (Definition
\ref{def:wprf}), then the single-challenge and multi-challenge notions
are equivalent due to a hybrid argument.  Importantly, we do not have
to consider access to $\Dwprf$ in the SKL security game. This is
because the SKL adversary $\qA$ can easily construct a program
$\widetilde{\qP}$ that receives sufficiently many samples from
$\Dwprf$, pre-computed using some leased-key $\qsk$.  The pirate
$\widetilde{\qP}$ can then simulate $\qP$ that expects such oracle
access.
\end{remark}

\begin{remark}
By a similar argument, testing on a single challenge input is also
sufficient for MLT-PRF. In this context, we can rely on the fact that
$\qA$ can provide $\widetilde{\qP}$ with multiple samples from
$\Dwprf$ without affecting the traceability. This is because with
overwhelming probability, all the evaluations are independent of the
identities of the generating keys, by the evaluation correctness
guarantee.
\end{remark}

% !TEX root = main.tex

\section{Verification Oracle Resilience from Tokenized MAC}\label{sec:vo}

Until now, we have focussed on obtaining standard-KLA security,
with either bounded or unbounded collusion resistance. We will now
show a black-box compiler that transforms any standard-KLA secure
scheme into one with verification oracle (VO-KLA) security. The
compiler requires a single ingredient; a uniquely quantum primitive
called tokenized MAC. Since tokenized MACs are known from OWFs
\cite{BSS21}, the cryptographic overhead of the compiler is
minimal.

\subsection{Tokenized MACs}

A tokenized MAC scheme $\TMac$ is a tuple of 4 algorithms $(\KG, \TG,
\qSign, \Vrfy)$, that are described as follows:

\begin{description}
\item $\KG(1^\secp) \ra \sk$: The key-generation algorithm takes the
security parameter as input and outputs a secret-key $\sk$.

\item $\TG(\sk) \ra \qtk$: The token-generation takes a secret-key as
input and outputs a quantum ``token'' state $\qtk$.

\item $\qSign(\qtk, \msg) \ra \sigma:$ The quantum signing algorithm takes as
input a quantum token $\qtk$ and a message $\msg \in \bit^\ell$. It
outputs a classical signature $\sigma$.

\item $\Vrfy(\sk, \sigma, \msg) \ra \top / \bot$: The classical verification algorithm
takes as input a secret key $\sk$, a signature $\sigma$ and a message
$\msg$. It outputs $\top$ or $\bot$.

A tokenized MAC scheme must satisfy the following correctness and
security requirements:

\paragraph{Correctness:} For every $\msg \in \bit^\ell$, the following
holds:

\[
\Pr\left[
\Vrfy(\sk, \sigma, \msg) \ra \bot
 \ :
\begin{array}{rl}
 &\sk \la \KG(1^\secp)\\
 &\qtk \la \TG(\sk)\\
 & \sigma \la \qSign(\qtk, \msg) \\
\end{array}
\right] \le \negl(\secp).
\]

\paragraph{Unforgeability:} The following holds for every QPT
adversary $\qA$ with classical oracle access to the verification
algorithm $\Vrfy(\sk, \cdot, \cdot)$:

\[
\Pr\left[
\begin{array}{rl}
 &\;\;\; \msg_1 \neq \msg_2\\
 &\land \; \Vrfy(\sk, \msg_1, \sigma_1) \ra \top\\
 &\land \; \Vrfy(\sk, \msg_2, \sigma_2) \ra \top\\
\end{array}
 \ :
\begin{array}{rl}
 &\sk \la \KG(1^\secp)\\
 &\qtk \la \TG(\sk)\\
 & (\msg_i, \sigma_i)_{i\in[2]} \la \qA^{\Vrfy(\sk, \cdot, \cdot)}(\qtk) \\
\end{array}
\right] \le \negl(\secp).
\]
\end{description}

\begin{theorem}[\cite{BSS21}]
Tokenized MACs exist, assuming OWFs exist.
\end{theorem}

\subsection{The Compiler}

Let $\widetilde{\SKL} = \widetilde{\SKL}.(\Setup, \widetilde{\qKG},
\widetilde{\qEval}, \widetilde{\qDel}, \widetilde{\Vrfy})$ be an SKL
scheme for application $(\qF, \qE, t)$, satisfying standard-KLA
security (Definition \ref{def:std-kla}). Let $\TMac =
\TMac.(\allowbreak\KG,
\TG, \qSign, \Vrfy)$ be a tokenized MAC scheme.  Then, the construction
$\SKL = \SKL.(\Setup, \allowbreak\qKG, \qEval, \qDel, \Vrfy)$ satisfies VO-KLA
security (Definition \ref{def:vo-kla}). The scheme is described as follows:

\begin{description}
\item $\underline{\qKG(\msk)}:$
\begin{itemize}
\item Generate $(\widetilde{\qsk}, \widetilde{\vk}) \la \widetilde{\qKG}(\msk)$.
\item Generate $\tmac.\sk \la \TMac.\KG(1^\secp)$ and $\tmac.\qtk \la
    \TMac.\TG(\tmac.\sk)$.
\item Output $\qsk \seteq (\widetilde{\qsk}, \tmac.\qtk)$ and $\vk
    \seteq (\widetilde{\vk}, \tmac.\sk)$.
\end{itemize}

\item $\underline{\qEval(\qsk, x)}:$
\begin{itemize}
\item Parse $\qsk$ as $\qsk = (\widetilde{\qsk}, \tmac.\qtk)$.
\item Output $y \la \widetilde{\qEval}(\widetilde{\qsk}, x)$.
\end{itemize}

\item $\underline{\qDel(\qsk)}:$
\begin{itemize}
\item Parse $\qsk$ as $\qsk = (\widetilde{\qsk}, \tmac.\qtk)$.
\item Compute $\widetilde{\cert} \la \widetilde{\qDel}(\widetilde{\qsk})$.
\item Compute $\tmac.\sigma \la \TMac.\qSign(\tmac.
   \qtk, \widetilde{\cert})$.
\item Output $\cert \seteq (\widetilde{\cert}, \tmac.\sigma)$
\end{itemize}

\item $\underline{\Vrfy(\vk, \cert)}:$
\begin{itemize}
\item Parse $\vk$ as $\vk = (\widetilde{\vk}, \tmac.\sk)$ and $\cert$
as $\cert = (\widetilde{\cert}, \tmac.\sigma)$.
\item If $\TMac.\Vrfy(\tmac.\sk, \tmac.\sigma, \widetilde{\cert}) =
    \top \land \widetilde{\Vrfy}(\widetilde{\vk}, \widetilde{\cert}) =
    \top$, output $\top$. Else, output $\bot$.
\end{itemize}
\end{description}

\paragraph{Evaluation Correctness:} This follows directly from the
evaluation correctness of $\widetilde{\SKL}$, as the TMAC part is not
involved in the algorithm $\qEval$.

\paragraph{Verification Correctness:} From the correctness of $\TMac$,
we have that a signature $\tmac.\sigma$ produced using a valid token
$\tmac.\qtk$ on any (message) certificate $\widetilde{\cert}$ will be
verified successfully. Hence, verification correctness follows from
that of the underlying SKL scheme $\widetilde{\SKL}$.

\begin{theorem}[VO-Resilience]\label{thm:vo-res} If the SKL scheme $\widetilde{\SKL}$
for application $(\qF, \qE, t)$ satisfies
($q$-bounded/unbounded) standard-KLA security, then the scheme
$\SKL$ for $(\qF, \qE, t)$ satisfies ($q$-bounded/unbounded)
VO-KLA security, assuming the TMAC scheme $\TMac$ satisfies
unforgeability.
\end{theorem}

\begin{proof}

\fuyuki{Can't we simply use the difference lemma? Let $\Hyb_0$ and $\Hyb_1$ be experiments. Let $\Event_b$ for $b\in\bit$ is an event defined in $\Hyb_b$. Suppose we have $\Pr[1\gets\Hyb_0 \land \lnot \Event_0]=\Pr[1\gets\Hyb_1 \land \lnot \Event_1]$ and $\Pr[\Event_0]=\Pr[\Event_1]$. Then. the difference lemma says that we have $\abs{\Pr[1\gets\Hyb_0]-\Pr[1\gets\Hyb_1]}\le\Pr[\Event_0]=\Pr[\Event_1]$. By using it, I thought the proof can be simplified. We set $\Hyb_0$ as $\expb{\SKL,\qA}{vo}{kla}(1^\secp, \qF, \qE, \gamma)$.
$\Hyb_1$ is the same as $\Hyb_0$ except the behavior of the verification oracle. For each $j\in[q]$, let $(j,\cert_j^*=(\widetilde{\cert}_j^*,\tmac.\sigma_j^*))$ be the first query with respect to $j$ such that $\tmac.\sigma_j^*$ is valid. Then, in $\Hyb_1$, after $(j,\cert_j^*=(\widetilde{\cert}_j^*,\tmac.\sigma_j^*))$ is queried, given $(j,\widetilde{\cert},\tmac.\sigma)$, the verification oracle returns $\top$ only if $\widetilde{\cert}=\widetilde{\cert}_j^*$ and $\tmac.\sigma_j^*$ is valid. (This means after $(j,\cert_j^*=(\widetilde{\cert}_j^*,\tmac.\sigma_j^*))$ is queried, the verification oracle does not need to check the validity of the deletion certificates with respect to $j$.) 
We also define $\Event_b$ as the event that in $\Hyb_b$, for some $j\in[q]$, after $(j,\cert_j^*=(\widetilde{\cert}_j^*,\tmac.\sigma_j^*))$ is queried, $\qA$ queries $(j,\widetilde{\cert},\tmac.\sigma)$ such that $\widetilde{\cert}\ne\widetilde{\cert}_j^*$ and $\tmac.\sigma$ is valid. Then, we have $\Pr[1\gets\Hyb_0 \land \lnot \Event_0]=\Pr[1\gets\Hyb_1 \land \lnot \Event_1]$ and $\Pr[\Event_0]=\Pr[\Event_1]$, since $\Hyb_0$ and $\Hyb_1$ are identical unless $\qA$ makes ``a bad query''.
We can prove that $\Pr[\Event_0]=\Pr[\Event_1]=\negl(\secp)$ from the security of $\TMac$.
We can also prove that $\Pr[1\gets\Hyb_1]=\negl$ from the security of $\widetilde{\SKL}$, since in $\Hyb_1$, the verification oracle (thus the reduction) needs to check the validity of only a single deletion certificate $\widetilde{\cert}_j^*$ for each $j\in[q]$.
What do you think?
}

\nikhil{In this proof strategy, it is unclear to me how $\Pr[1 \la \Hyb_0 \land \lnot E_0] = \Pr[1 \la \Hyb_1 \land \lnot E_1]$. The problem is that even the first query $\qA$ makes can be a ``bad'' query, which means $\tmac.\sigma$ is valid, but the actual certificate $\widetilde{\cert}$ is invalid. In this case, $\Hyb_0$ will respond with $\bot$ while $\Hyb_1$ responds with $\top$. This is because $\Hyb_1$ should not actually perform verification of $\widetilde{\cert}$ but should just check $\tmac.\sigma$. Otherwise the reduction to $\widetilde{\SKL}$ will not work, because it doesn't have the verification keys. Note that the reduction to $\widetilde{\SKL}$ cannot check the validity of even a single deletion certificate $\widetilde{\cert}$. This means that the view of $\qA$ is already different in the hybrids after the first query, even before it outputs any $\TMac$ forgery.}
\fuyuki{
I thought ``$1\gets\Hyb_b \land \lnot E_b$'' implies that $\widetilde{\cert}_j^*$ is valid for every $j\in[q]$ in $\Hyb_b$ (otherwise, $\Hyb_b$ never outputs $1$), thus we do not have to consider the situation you mentioned when we focus on the event ``$1\gets\Hyb_b \land \lnot E_b$''.
More concretely, let $F_b$ be the event that $\widetilde{\cert}_j^*$ is valid for every $j\in[q]$ in $\Hyb_b$.
Then, we have $\Pr[1 \la \Hyb_b \land \lnot E_b] = \Pr[1 \la \Hyb_b \land \lnot E_b \land F_b]$, since $\Pr[1 \la \Hyb_b \land \lnot E_b \land \lnot F_b]=0$.
We also have $\Pr[1 \la \Hyb_0 \land \lnot E_0 \land F_0]=\Pr[1 \la \Hyb_1 \land \lnot E_1 \land F_1]$. (The event $F_b$ excludes the case you mentioned.)
As a result, we obtain $\Pr[1 \la \Hyb_0 \land \lnot E_0] = \Pr[1 \la \Hyb_1 \land \lnot E_1]$.
What do you think?
}
\fuyuki{I realized that $\Pr[\Event_0]$ might be different from $\Pr[\Event_1]$, but we still have $\Pr[\Event_0]=\negl(\secp)$ and $\Pr[\Event_1]=\negl(\secp)$,which would be sufficient, since
\begin{align}
&\abs{\Pr[1\gets\Hyb_0]-\Pr[1\gets\Hyb_1]}\\
&\le\abs{\Pr[1\gets\Hyb_0\land E_0]-\Pr[1\gets\Hyb_1\land E_1]}
+\abs{\Pr[1\gets\Hyb_0\land \lnot E_0]-\Pr[1\gets\Hyb_1\land \lnot E_1]}\\
&=\abs{\Pr[1\gets\Hyb_0\land E_0]-\Pr[1\gets\Hyb_1\land E_1]}\\
&=\abs{\Pr[E_0]\cdot\Pr[1\gets\Hyb_0 | E_0]-\Pr[E_1]\cdot\Pr[1\gets\Hyb_1| E_1]}
=\negl(\secp).
\end{align}
}
\nikhil{Got it. Changed the proof.}

Assume that $\SKL$ does not satisfy unbounded VO-KLA security (the
case of $q$-bounded security is similar). Then,
for some $\epsilon = 1/\poly(\secp)$,
there exists a QPT attacker $\qA$ that wins with $\nonnegl(\secp)$ 
probability in the experiment
$\expb{\SKL,\qA}{vo}{kla}(1^\secp, \qF, \qE, t, \epsilon)$.
Consider now the following hybrid experiments:

\begin{description}
\item \underline{$\Hyb_0$:} This is the same as the experiment 
$\expb{\SKL,\qA}{vo}{kla}(1^\secp, \qF, \qE, t, \epsilon)$, which
executes as follows:

\begin{enumerate}
\item $\qCh$ samples $(\msk, f, \aux_f) \la \Setup(1^\secp, \bot)$ and
sends $\aux_f$ to $\qA$.
\item $\qA$ sends $q = \poly(\secp)$ to $\qCh$.
\item For each $i \in [q]$, $\qCh$ computes $(\qsk_i, \vk_i) \la
    \qKG(\msk)$. It sends $(\qsk_1, \ldots, \qsk_q)$ to $\qA$.
\item For each $i\in[q]$, $\qCh$ sets $V_i \seteq \bot$.
\item When $\qA$ makes a query $(j, \cert)$ to $\Oracle{\Vrfy}$,
$\qCh$ performs the following:
\begin{itemize}
\item Compute $d \la \Vrfy(\vk_j, \cert)$.
\item If $V_j = \bot$, set $V_j \seteq d$. Return $d$.
\end{itemize}
\item $\qA$ outputs a quantum program $\qP^* = (U, \rho)$ to $\qCh$.
\item If $V_i = \top$ for each $i \in [q]$ and $\qP^*$ is tested to
be $\epsilon$-good wrt $(f, \qE, t)$, then $\qCh$ outputs $\top$.
Else, it outputs $\bot$.
\end{enumerate}

\item \underline{$\Hyb_1$:} This is similar to $\Hyb_0$, with the
following differences colored in $\textcolor{red}{red}$.

\begin{enumerate}
\item $\qCh$ samples $(\msk, f, \aux_f) \la \Setup(1^\secp, \bot)$ and
sends $\aux_f$ to $\qA$.
\item $\qA$ sends $q = \poly(\secp)$ to $\qCh$.
\item For each $i \in [q]$, $\qCh$ computes $(\qsk_i, \vk_i) \la
\qKG(\msk)$. It sends $(\qsk_1, \ldots, \qsk_q)$ to $\qA$.
\textcolor{red}{For each $i\in[q]$, parse $\vk_i = (\widetilde{\vk}_i,
\tmac.\sk_i)$}.

\item \textcolor{red}{For each $i\in[q]$, $\qCh$ sets
$\widetilde{\cert}_i \seteq \bot$.}
\item When $\qA$ makes a query $(j, \cert)$ to $\Oracle{\Vrfy}$,
$\qCh$ parses $\cert = (\widetilde{\cert}, \tmac.\sigma)$ and
performs the following, where $\{\tmac.\sk_j\}_{j\in[q]}$ are
generated as part of $\{\vk_j\}_{j\in[q]}$.

\begin{itemize}
\item \textcolor{red}{Check if $\TMac.\Vrfy(\tmac.\sk_j,
\widetilde{\cert}, \tmac.\sigma) = \top$. If so, set
$\widetilde{\cert}_j \seteq
\widetilde{\cert}$ and output $\top$. Else, output $\bot$.}
\end{itemize}

\item $\qA$ outputs a quantum program $\qP^* = (U, \rho)$ to $\qCh$.
\item \textcolor{red}{If
$\widetilde{\Vrfy}(\widetilde{\vk}_i, \widetilde{\cert_i}) = \top$ for each
$i\in[q]$} and $\qP^*$ is tested to
be $\epsilon$-good wrt $(f, \qE, t)$, then $\qCh$ outputs $\top$.
Else, it outputs $\bot$.
\end{enumerate}

\end{description}

Now, let $\Event_0$ be the event that $\qA$ makes queries $(j,
\cert = (\widetilde{\cert}, \tmac.\sigma))$ and $(j, \cert' =
(\widetilde{\cert}', \tmac.\sigma'))$ such that
$\TMac.\Vrfy(\tmac.\sk_j, \widetilde{\cert}, \tmac.\sigma) =
\TMac.\Vrfy(\tmac.\sk_j, \widetilde{\cert}', \tmac.\sigma') = \top$
and $\widetilde{\cert} \neq \widetilde{\cert}'$
in $\Hyb_0$. Let $\Event_1$ be defined similarly, but corresponding
to $\Hyb_1$. It is easy to see that $\Pr[\Event_0] = \negl(\secp)$ and
$\Pr[\Event_1] = \negl(\secp)$, as otherwise, the security of $\TMac$ is
compromised.

Observe now that $\Pr[\Hyb_0 \ra \top  \land \lnot \Event_0] =
\Pr[\Hyb_1 \ra \top \land \lnot \Event_1]$. This is because the
only time $\qA$ obtains a different response from $\Hyb_0$ and
$\Hyb_1$ is when it makes a ``bad'' query $(j, \cert =
(\widetilde{\cert}, \tmac.\sigma))$ such that
$\widetilde{\Vrfy}(\widetilde{\vk}_j, \widetilde{\cert}) = \bot$ and
$\TMac.\Vrfy(\tmac.\sk_j, \widetilde{\cert}, \tmac.\sigma) = \top$.
However, when the events $\lnot \Event_0, \lnot \Event_1$ occur in
$\Hyb_0, \Hyb_1$ respectively, it is clear that $\qA$ would cause
both these
hybrids to output $\bot$, assuming it makes such a ``bad'' query at
any point. Consequently, we have that:

\begin{align}
&\abs{\Pr[\Hyb_0 \ra \top]-\Pr[\Hyb_1 \ra \top]}\\
&\le\abs{\Pr[\Hyb_0 \ra \top\land \Event_0]-\Pr[\Hyb_1 \ra \top\land
\Event_1]}\\
&+\abs{\Pr[\Hyb_0 \ra \top\land \lnot \Event_0]-\Pr[\Hyb_1 \ra
\top\land \lnot \Event_1]}\\
&=\abs{\Pr[\Hyb_0 \ra \top \land \Event_0]-\Pr[\Hyb_1 \ra \top\land
\Event_1]}\\
&=\abs{\Pr[\Event_0]\cdot\Pr[\Hyb_0 \ra \top|
\Event_0]-\Pr[\Event_1]\cdot\Pr[\Hyb_1 \ra \top| \Event_1]}
=\negl(\secp).
\end{align}

Now, by the assumption that $\Pr[\Hyb_0 \ra \top] =
\nonnegl(\secp)$, we have that $\Pr[\Hyb_1 \ra \top] =
\nonnegl(\secp)$. We now show the following reduction $\qR$ that
simulates $\Hyb_1$ for $\qA$ and breaks the standard-KLA security of
$\widetilde{\SKL}$.

\begin{description}
\item \underline{Execution of $\qR^\qA$ in
$\expb{\widetilde{\SKL},\qR}{std}{kla}(1^\secp, \qF, \qE, t,
\epsilon):$}

\begin{enumerate}
\item $\qCh$ samples $(\msk, f, \aux_f) \la
\Setup(1^\secp, \bot)$ and sends $\aux_f$ to $\qR$.
\item $\qR$ sends $\aux_f$ to $\qA$. $\qA$ sends $q = \poly(\secp)$ to
$\qR$. $\qR$ sends $q$ to $\qCh$.
\item For each $i \in [q]$, $\qCh$ computes $(\widetilde{\qsk}_i,
\widetilde{\vk}_i) \la \widetilde{\qKG}(\msk)$. It sends
$(\widetilde{\qsk}_i)_{i\in[q]}$ to $\qR$.
\item For each $i \in [q]$, $\qR$ computes $\tmac.\sk_i \la
\TMac.\KG(1^\secp)$ and $\tmac.\qtk_i \la \TMac.\TG(\tmac.\sk_i)$.
It then sets $\qsk_i \seteq (\widetilde{\qsk}_i, \tmac.\qtk_i)$. It
sends $(\qsk_i)_{i\in[q]}$ to $\qA$.

\item For every $i \in [q]$, $\qR$ initializes $\widetilde{\cert}_i
\seteq \bot$.

\item When $\qA$ makes a query $(j, \cert)$, $\qR$ parses $\cert =
(\widetilde{\cert}, \tmac.\sigma)$. If $\TMac.\Vrfy(\tmac.\sk_j,
\tmac.\sigma, \widetilde{\cert}) = \top$, $\qR$ sets
$\widetilde{\cert}_j \seteq
\widetilde{\cert}$ and responds with $\top$. Else, it responds with $\bot$.

\item $\qA$ outputs a quantum program $\qP^*$ to $\qR$. $\qR$
sends $\widetilde{\cert}_1, \ldots, \widetilde{\cert}_q$ to $\qCh$ along with $\qP^*$.

\item If for each $i \in [q]$, it holds that
$\widetilde{\Vrfy}(\widetilde{\vk}_i,
\widetilde{\cert}_i) = \top$ and $\qP^*$ is tested to be $\epsilon$-good wrt
$(f, \qE, t)$, then $\qCh$ outputs $\top$. Else, it outputs $\bot$.

\end{enumerate}
\end{description}

Observe that the view of $\qA$ is identical in the hybrid $\Hyb_1$
and the reduction $\qR$. This means
that with non-negligible
probability, $\widetilde{\Vrfy}(\widetilde{\vk}_i, \widetilde{\cert}_i) =
\top$ for each $i \in [q]$ and $\qP^*$ is $\epsilon$-good wrt $(f,
\qE, t)$. This breaks the standard-KLA security of $\widetilde{\SKL}$.
\end{proof}

Finally, from Theorem \ref{thm:const} and Theorem \ref{thm:vo-res}, we
have the following corollary:

\begin{corollary}\label{cor:mltt_vo-kla}
If there exists an MLTT scheme for application $(\qF, \qE, t)$
satisfying traceability (Definition \ref{def:mltt-tr}), there exists
an SKL scheme for $(\qF, \qE, t)$ satisfying $q$-bounded VO-KLA
security (Definition \ref{def:vo-kla}).
\end{corollary}

Likewise, from Theorem \ref{thm:bcr-prf-skl_lwe} and Theorem \ref{thm:vo-res}, we have

\begin{corollary}
Assuming the quantum hardness of LWE with sub-exponential modulus, for
every collusion-bound $q = \poly(\secp)$, there exists a PRF-SKL
scheme satisfying $q$-bounded VO-KLA security (Definition
\ref{def:vo-kla}).
\end{corollary}

\section{Unbounded Collusion-Resistant SKL for Signatures}\label{sec:sig}

In this section, we construct unbounded collusion-resistant digital
signatures with secure key leasing (DS-SKL). We rely on a DS-SKL
scheme that is secure given a single leased key, along with a
classical post-quantum digital signature scheme. Since the former
building block is known from the SIS assumption
(\cite{EC:KitMorYam25}) and the latter is implied by it, we obtain our
result assuming SIS holds. Note that we focus on the more natural
notion of DS-SKL with static signing keys as in \cite{EC:KitMorYam25},
and unlike in \cite{TQC:MorPorYam24}.

\subsection{Preparation}

\begin{definition}[DS-SKL]\label{def:ds-skl}
An SKL scheme for digital signatures (DS-SKL) consists of algorithms
$(\Setup, \qKG, \qEval, \qDel, \allowbreak \Vrfy)$ as per the syntax of SKL
(Definition \ref{def:skl}) along with algorithms $(\Sign, \SigVrfy)$.
Note that $\Setup$ samples a pair of keys $\ssk, \svk$ which are the
signature signing key and the signature verification key respectively.
Then,
$f$ is
interpreted as $f \seteq \Sign(\ssk, \cdot) \| \svk$ and $\aux_f$ as
$\aux_f \seteq \svk$. The correctness and security guarantees are
described by the tuple $(\qF, \qE, t)$ which are specified as
follows:

\begin{description}
\item $\underline{\qF(\qg, f, r)}:$
\begin{itemize}
\item Obtain $\svk$ from $f$.
\item Sample $\msg$ uniformly from the message space using the
randomness $r$.
\item If $\SigVrfy(\svk, \msg, \qg(\msg)) = 1$, output $1$. Else, output $0$.
\end{itemize}

\item $\underline{\qE(f, \qP, (r, \key))}$:
\begin{itemize}
\item Ignore $\key$.
\item Obtain $\svk$ from $f$.
\item Sample $\msg$ uniformly from the message space using the
randomness $r$.
\item If $\SigVrfy(\svk, \msg, \qP(1^\secp, \msg)) = 1$, output 1.
Else, output 0.
\end{itemize}

\item $\underline{t} \seteq 0$.
\end{description}
\end{definition}

We now define a classical digital signature scheme as follows:

\begin{definition}[Digital Signature Scheme] A digital signature
scheme consists of algorithms $(\Gen, \Sign, \Vrfy)$ which are
specified as follows:

\begin{description}
\item $\Gen(1^\secp) \ra (\sk, \vk):$ The key generation algorithm
outputs a secret signing key $\sk$ and a public verification key $\vk$.
\item $\Sign(\sk, \msg) \ra \sigma:$ The signing algorithm takes the signing key
$\sk$ and a message $\msg \in \cM$ as inputs. It generates a signature
$\sigma$ as output.
\item $\Vrfy(\vk, \msg, \sigma) \ra \top / \bot:$ The verification algorithm takes the
verification key $\vk$, a message $\msg$ and a signature $\sigma$ as
inputs. It outputs $\top$ (valid) or $\bot$ (invalid).
\end{description}

\paragraph{Correctness:}

The correctness requires the following to hold for all $\msg \in \cM$
for a message space $\cM$:

\[
\Pr\left[
\Vrfy(\vk, \msg, \sigma) = \bot
 \ :
\begin{array}{rl}
 &(\sk, \vk) \la \Gen(1^\secp)\\
 &\sigma \la \Sign(\sk, \msg)
\end{array}
\right] \le \negl(\secp).
\]
\end{definition}

\paragraph{Unforgeability:}

The unforgeability security guarantees that the following holds for
all QPT adversaries $\qA$, where $Q$ is the set of messages queried by
$\qA$ to the oracle $\Sign(\sk, \cdot)$:

\[
\Pr\left[
\Vrfy(\vk, \msg, \sigma) = \top \land \msg \notin Q
 \ :
\begin{array}{rl}
 &(\sk, \vk) \la \Gen(1^\secp)\\
 &(\msg, \sigma) \la \qA^{\Sign(\sk, \cdot)}(1^\secp, \vk)
\end{array}
\right] \le \negl(\secp).
\]

\subsection{The Compiler}

The DS-SKL scheme $\DSSKL$ uses the DS-SKL scheme $\widetilde{\DSSKL}
= (\widetilde{\Setup}, \widetilde{\qKG}, \allowbreak\widetilde{\qEval},
\widetilde{\qDel}, \widetilde{\Vrfy}, \widetilde{\Sign},
\widetilde{\SigVrfy})$ and the digital signature
scheme $\DSig = (\Gen, \Sign, \allowbreak\Vrfy)$ as building blocks. It consists
of $(\Setup, \qKG, \qEval, \qDel, \Vrfy)$ and the algorithms $(\Sign,
\SigVrfy)$:

\begin{description}
\item $\underline{\Setup(1^\secp, \bot):}$
\begin{itemize}
\item Sample $(\sig.\sk, \sig.\vk) \la \DSig.\Gen(1^\secp)$.
\item Output $\big(\msk \seteq (\sig.\sk, \sig.\vk), f \seteq
\DSig.\Sign(\sig.\sk, \cdot) \| \sig.\vk, \aux_f \seteq
\sig.\vk\big)$.
\end{itemize}

\item $\underline{\qKG(\msk)}:$

\begin{itemize}
\item Parse $\msk$ as $\msk = (\sig.\sk, \sig.\vk)$.
\item Sample $(\widetilde{\msk}, \widetilde{f},
    \widetilde{\aux}_{\widetilde{f}})
\la \widetilde{\Setup}(1^\secp, \bot)$.
\item Sample $(\widetilde{\qsk}, \widetilde{\vk}) \la
    \widetilde{\qKG}(\widetilde{\msk})$.
\item Obtain $\widetilde{\svk}$ from $\widetilde{f}$ and compute
$\sig.\sigma \la \DSig.\Sign(\sig.\sk, \widetilde{\svk})$.

\item Set $\vk \seteq \widetilde{\vk}$ and $\qsk \seteq
(\widetilde{\qsk}, \widetilde{\svk}, \sig.\sigma)$.
\item Output $(\qsk, \vk)$.
\end{itemize}

\item $\underline{\qEval(\qsk, \msg)}:$
\begin{itemize}
\item Parse $\qsk = (\widetilde{\qsk}, \widetilde{\svk},
\sig.\sigma)$.
\item Compute $\widetilde{\sigma} \la
\widetilde{\qEval}(\widetilde{\qsk}, \msg)$.
\item Output $\sigma' \seteq (\widetilde{\sigma}, \widetilde{\svk}, \sig.\sigma)$.
\end{itemize}

\item $\underline{\qDel(\qsk)}:$
\begin{itemize}
\item Parse $\qsk = (\widetilde{\qsk}, \widetilde{\svk},
\sig.\sigma)$.
\item Output $\widetilde{\cert} \la
\widetilde{\qDel}(\widetilde{\qsk})$.
\end{itemize}

\item $\underline{\Vrfy(\vk, \cert)}:$
\begin{itemize}
\item Parse $\vk = \widetilde{\vk}$.
\item Output $\widetilde{\Vrfy}(\widetilde{\vk}, \cert)$.
\end{itemize}

\item $\underline{\Sign(\sig.\sk, \msg)}:$
\begin{itemize}
\item Output $\sig.\sigma' \la \DSig.\Sign(\sig.\sk, \msg)$.
\end{itemize}

\item $\underline{\SigVrfy(\sig.\vk, \msg, \sigma')}:$
\begin{itemize}
\item If $\DSig.\Vrfy(\sig.\vk, \msg, \sigma') = \top$, output $\top$.
\item Otherwise, parse $\sigma' = (\widetilde{\sigma},
    \widetilde{\svk}, \sig.\sigma)$. If $\DSig.\Vrfy(\sig.\vk,
    \widetilde{\svk}, \sig.\sigma) = \top \land
    \widetilde{\SigVrfy}(\widetilde{\svk}, \msg, \widetilde{\sigma}) =
    \top$, output $\top$.
\item Else, output $\bot$.
\end{itemize}

\end{description}

\paragraph{Evaluation Correctness:}
Observe that $\qEval(\qsk, \msg)$ outputs $\sigma' =
(\widetilde{\sigma}, \widetilde{\svk}, \sig.\sigma)$ where $\qsk =
(\widetilde{\qsk}, \widetilde{\svk}, \sig.\sigma)$. Notice that $\qKG$
generates $\sig.\sigma \la \DSig.\Sign(\sig.\sk, \widetilde{\svk})$.
Consequently, $\DSig.\Vrfy(\sig.\vk,\allowbreak \widetilde{\svk},
\sig.\sigma) = \top$ holds with overwhelming probability by the
correctness of $\DSig$. Moreover, $\widetilde{\SigVrfy}(\widetilde{\svk},
\msg, \widetilde{\sigma})\allowbreak = \top$ holds with overwhelming
probability by the correctness of $\widetilde{\DSSKL}$, since
$\widetilde{\sigma}$ is generated as $\widetilde{\sigma} \la
\widetilde{\Eval}(\widetilde{\qsk}, \msg)$. Consequently,
$\SigVrfy(\sig.\vk, \msg, \sigma')$ outputs $\top$ with overwhelming
probability. By the correctness of $\widetilde{\DSSKL}$, we also have that $\widetilde{\qsk}$ is almost undisturbed. Hence, the post-evaluation state of $\DSSKL$ is close in trace distance to $\qsk$.

\begin{theorem}\label{thm:ds-skl-comp}
Let $\widetilde{\DSSKL}$ be a DS-SKL scheme satisfying $1$-bounded
standard-KLA security and $\DSig$ be a classical post-quantum digital signature scheme. Then, the construction $\DSSKL$ satisfies unbounded
standard-KLA security (Definition \ref{def:std-kla}).
\end{theorem}

\begin{proof}
Let $\qA$ be a QPT adversary that succeeds with non-negligible
probability in $\expb{\DSSKL,\qA}{std}{kla}(1^\secp,
\qF, \qE, t, \epsilon)$ for some $\epsilon = 1/\poly(\secp)$. Let
$\qA$ obtain keys $\qsk_1, \ldots, \qsk_q$. By assumption, $\qA$
outputs a quantum program $\qP^*$ such that when provided a uniformly
random message $\msg$, $\qP^*$ outputs a valid signature $\sigma'$
with non-negligible probability.

For each $i \in [q]$, let $\qsk_i = (\widetilde{\qsk}_i,
\widetilde{\svk}_i, \sig.\sigma_i)$, where $\sig.\sigma_i \la
\DSig.\Sign(\sig.\sk, \allowbreak\widetilde{\svk}_i)$. Consider first the case
that $\qP^*$ outputs $\sigma'$ which satisfies $\DSig.\Vrfy(\sig.\vk,
\msg, \allowbreak\sigma') = \top$ (the first accept condition of
$\SigVrfy(\sig.\vk, \msg, \cdot)$). Let $Q =
\{\widetilde{\svk}_i\}_{i\in[q]}$. In this case,
$\qP^*$ can be used to break unforgeability of $\DSig$ in a
straightforward manner, assuming that $\msg \notin Q$. Since $\msg$ is
chosen uniformly at random from a super-polynomial size space, this
holds with overwhelming probability. Now, consider the case when
$\qP^*$ outputs $\sigma' = (\widetilde{\sigma}, \widetilde{\svk},
\sig.\sigma)$ where $\DSig.\Vrfy(\sig.\vk, \widetilde{\svk},
\sig.\sigma) = \top$ holds but $\widetilde{\svk} \notin Q$. Once
again, by the unforgeability of $\DSig$, this is infeasible. As a
result, we can consider the case where $\qP^*$ outputs $\sigma' =
(\widetilde{\sigma}, \widetilde{\svk}, \sig.\sigma)$ for some
$\widetilde{\svk} \in Q$. In this case, we show the following
reduction $\qR$ which breaks $1$-bounded standard KLA security of
$\widetilde{\DSSKL}$:

\begin{description}
\item \underline{Execution of $\qR$ in 
$\expb{\widetilde{\DSSKL},\qR}{std}{kla}(1^\secp, 1, \qF, \qE, t,
\epsilon)$:}
\begin{enumerate}
    \item $\qCh$ samples $(\widetilde{\msk}, \widetilde{f},
        \widetilde{\aux}_{\widetilde{f}}) \la
        \widetilde{\Setup}(1^\secp, \bot)$ and sends
        $\widetilde{\aux}_{\widetilde{f}}$ to $\qR$.

    \item $\qCh$ computes $(\qsk^*, \vk^*) \la
        \widetilde{\qKG}(\widetilde{\msk})$ and sends $\qsk^*$ to $\qR$.

    \item $\qR$ samples $(\sig.\sk, \sig.\vk) \la
        \DSig.\Gen(1^\secp)$ and computes $\msk, f, \aux_f$ as per
        $\Setup(1^\secp, \bot)$. It sends $\aux_f$ to $\qA$.
    \item $\qA$ sends $q = \poly(\secp)$ to $\qR$.
    \item $\qR$ picks a random index $j \in [q]$. For each $i \in [q]
        \setminus \{j\}$, $\qR$ computes $(\qsk_i, \vk_i) \la
        \qKG(\msk)$.
    \item $\qR$ obtains $\svk^*$ from
        $\widetilde{\aux}_{\widetilde{f}}$ and computes $\sig.\sigma^*
        \la \DSig.\Sign(\sig.\sk, \svk^*)$. It sets $\qsk_j \seteq (\qsk^*, \svk^*, \sig.\sigma^*)$.
        Finally, it sends $(\qsk_i)_{i\in[q]}$ to $\qA$.
    \item $\qA$ sends $(\cert_1, \ldots, \cert_q)$ and a program
        $\qP^* = (U, \rho)$ to $\qR$.
    \item $\qR$ outputs $\cert_j$ along with $\qP^*$ to $\qCh$.

    \item $\qCh$ tests if $\widetilde{\Vrfy}(\vk^*, \cert_j) = \top$
        and if $\qP^*$ is $\epsilon$-good wrt $(\widetilde{f},
        \qE, t)$. If these hold, then it outputs $\top$, and
        otherwise outputs $\bot$.
\end{enumerate}
\end{description}

Notice that the view of $\qA$ as run by $\qR$ is identical to its view
in $\expb{\DSSKL,\qA}{std}{kla}(1^\secp, \allowbreak\qF, \qE, t, \epsilon)$. By assumption, 
when
$\qCh$ provides $\msg$ to $\qP^*$ as part of the
$\epsilon$-good test, $\qP^*$ must output $\sigma' = (\widetilde{\sigma},
\widetilde{\svk}, \sig.\sigma)$ such that $\widetilde{\svk} =
\svk^*$ and $\widetilde{\SigVrfy}(\widetilde{\svk}, \msg,
\widetilde{\sigma}) = \top$ with non-negligible probability. This is because
the index $j$ is chosen by $\qR$ uniformly at random from $[q]$.
Consequently, $\qP^*$ would pass the $\epsilon$-good test run by $\qCh$.
Clearly, $\widetilde{\Vrfy}(\vk^*, \cert_j) = \top$ also holds simultaneously with non-negligible probability by assumption.
Hence, $\qR$ ends up breaking the $1$-bounded
standard-KLA security of $\widetilde{\DSSKL}$.
\end{proof}

We now state a theorem from prior work that essentially achieves
$1$-bounded standard-KLA secure DS-SKL based on the SIS assumption.

\begin{theorem}[\cite{EC:KitMorYam25}]\label{thm:sis_ds-skl}
Assuming the SIS assumption holds, there exists a DS-SKL scheme with 1-bounded standard-KLA security.
\end{theorem}

Now, from Theorem \ref{thm:sis_ds-skl}, Theorem \ref{thm:ds-skl-comp}
and Theorem \ref{thm:vo-res}, we have the following corollary:

\begin{corollary}
There exists a DS-SKL scheme satisfying unbounded VO-KLA security
(Definition \ref{def:vo-kla}), given that the SIS assumption holds.
\end{corollary}

\ifnum\anonymous=1
\else
% \section*{Acknowledgement}

\fi

	\ifnum\llncs=1
\bibliographystyle{splncs04}
\bibliography{bib/abbrev3,bib/crypto,bib/siamcomp_jacm,bib/other}
	\else
\bibliographystyle{alpha} 
\bibliography{bib/abbrev3,bib/crypto,bib/siamcomp_jacm,bib/other}
	\fi

\ifnum\cameraready=0
	\ifnum\llncs=0
	%%%%%% Full version region %%%%%%
	\appendix
%%% appendix files here
% \input{more-related}
% \input{omitted-prelim}
% \input{proof-two-sup}
% \input{proof-const}
% \input{mlt-prf}
% \input{vo}
% \input{sig}
%\input{more-related}
%\input{omitted-prelim}
%\input{proof-two-sup}
%\input{proof-const}

\else
%%%%% LNCS-style submission version region %%%%%
	\newpage
        \AddToHook{cmd/appendix/before}{\gdef\theHsection{\Alph{section}}}
	 	\appendix
	 	\setcounter{page}{1}
 	{
	\noindent
 	\begin{center}
	{\Large SUPPLEMENTARY MATERIAL}
	\end{center}
 	}
	\setcounter{tocdepth}{2}
	\ifnum\noaux=1
 	\else
% {\color{red}{We attached the full version of this paper as a separated file (auxiliary supplemental material) for readability. It is available from the program committee members.}}
	\fi

% !TEX root = main.tex

\section{More Related Work}\label{sec:more-related}

\subsubsection*{Collusion-Resistant Copy Protection.}
The first collusion resistant copy protection schemes in the plain
model were shown in the work of Liu et al. \cite{TCC:LLQZ22}. They
constructed copy protection schemes for PKE, PRFs and digital
signatures that are $k \ra k+1$ secure, i.e., an adversary receiving
$k = \poly(\secp)$ many copies cannot produce $k+1$ copies.
Importantly, their scheme is bounded collusion-resistant, i.e., the
parameter sizes grow linearly with the collusion-bound $k$. In the work of
{\c C}akan and Goyal~\cite{TCC:CakGoy24}, unbounded collusion-resistant copy protection
schemes were constructed in the plain model for PKE, PKFE, PRFs, and
digital signatures. Although these schemes imply SKL with similar
collusion-resistance guarantees (See the discussion in
\cite{EC:AKNYY23}), they all rely on iO. Since, copy protection is known
to imply public-key quantum money in general, achieving it with weaker than iO
assumptions is a major open problem.

\subsubsection*{Quantum-Secure Traitor Tracing.}

Traitor tracing in the quantum setting was first explored in the work
of Zhandry \cite{TCC:Zhandry20}. The work identifies several
challenges of dealing with quantum adversaries that output quantum
pirate programs. Firstly, it is not possible to know the success
probability of a quantum pirate until it is measured. This is because
the pirate may be in a superposition of ``successful'' and
``unsuccessful'' pirates.  Moreover, a measurement may disturb the
pirate and render it useless. Consequently, the work presented
workarounds for such definitional issues. Additionally, estimating the
success probability is also challenging, as it requires testing the
adversary on several samples from a distribution. The work shows an
efficient quantum procedure for this task by building on the work of
Marriott and Watrous \cite{CC:MarWat05}. Furthermore, the work shows a useful
property: consecutive estimations of the adversary's success
probability on computationally close distributions produce similar
outcomes. These quantum tools were leveraged to extend classical
private linear broadcast encryption (PLBE) based tracing schemes
\cite{EC:BonSahWat06} to the quantum setting. In a followup work by Zhandry
\cite{C:Zhandry23}, a new quantum rewinding technique due to Chiesa et
al. \cite{FOCS:CMSZ21} was utilized to further expand the kind of
probability estimations that can be performed without destroying the
pirate. The work utilized this
technique to extend several more classical tracing
schemes for PKE to the quantum setting, including several collusion-resistant ones.

The work of Kitagawa and Nishimaki \cite{EC:KitNis22} explores the setting of
watermarking in the quantum setting, which is a notion similar to that
of traitor tracing. Specifically, they showed a watermarkable PRF
secure against quantum adversaries (that output quantum pirate
programs), based on the hardness of LWE. We utilize this PRF as a
building block in our MLTT construction for PRFs (Section
\ref{sec:mlt-prf}).  Recently, Kitagawa and Nishimaki
\cite{EPRINT:KitNis25}
constructed a watermarkable digital signature scheme that is secure
against quantum adversaries, in the setting of white-box
traitor-tracing introduced by Zhandry \cite{C:Zhandry21}.

\subsubsection*{Certified Deletion.}

The notion of certified deletion for encryption was introduced in the
work of Broadbent and Islam \cite{TCC:BroIsl20}. This primitive allows
the generation of quantum ciphertexts, which can be provably deleted
by presenting a classical certificate. After deletion, even if the
secret-key is revealed, one cannot learn the contents of the
ciphertext they once held. Observe that the difference between this
notion and SKL is that here, it is access to secret data that is being
``revoked'', rather than the ability to evaluate some function.
Following this work, several other works have studied certified
deletion for different primitives
\cite{AC:HMNY21,ITCS:Poremba23,C:BarKhu23,EC:HKMNPY24,EC:BGKMRR24} and
with publicly-verifiable deletion
\cite{AC:HMNY21,EC:BGKMRR24,TCC:KitNisYam23,TCC:BKMPW23}. Recently,
the work of Ananth, Mutreja and Poremba \cite{EPRINT:AnaMutPor24}
introduced multi-copy revocable encryption.  This is a notion similar
to certified deletion but guarantees security in a setting where
multiple copies of the quantum ciphertext are provided to the
adversary. Note that this is in contrast to the collusion-resistant
setting we consider, where multiple i.i.d leased-keys are provided,
instead of identical copies of the same quantum state.

% !TEX root = main.tex

\section{Omitted Preliminaries}\label{sec:omitted-prelim}

\begin{definition}[Positive Operator-Valued Measure (POVM)]
A POVM is a set of Hermitian positive semi-definite matrices
$\cM = \{M_i\}_{i\in\cI}$
s.t. $\sum_{i\in\cI}M_i =
\mathbb{I}$ holds. Applying $\cM$ to a state $\qstate{q}$ produces
outcome $i \in \cI$ with probability $\Trace(\qstate{q}M_i)$. Let
$\cM(\ket{\psi})$ denote the distribution of the output of applying
$\cM$ to $\ket{\psi}$.
\end{definition}

\begin{definition}[Quantum Measurement] A quantum measurement is a set
of matrices $\cE = \{E_i\}_{i\in\cI}$ satisfying
$\sum_{i\in\cI}E^\dagger_iE_i = \mathbb{I}$. Applying $\cE$ to a state
$\qstate{q}$ produces outcome $i$ with probability
$p_i \seteq \Tr(\qstate{q}E_i^\dagger E_i)$ with the corresponding
post-measurement state being $E_i \qstate{q} E_i^\dagger / p_i$.
For any quantum state $\qstate{q}$, let $\cE(\qstate{q})$ denote the
distribution of the outcome of applying $\cE$ to $\qstate{q}$. For any
states $\qstate{q_0}$ and $\qstate{q_1}$, the statistical distance
between $\cE(\qstate{q_0})$ and $\cE(\qstate{q_1})$ is bounded above
by $\TD(\qstate{q_0}, \qstate{q_1})$.
\end{definition}

\begin{definition}[Projective Measurement/POVM]
A quantum measurement $\cE = \{E_i\}_{i\in\cI}$ is a projective
measurement if for each $i \in \cI$, $E_i$ is a projector. A binary
projective measurement is of the form $\cE = \{\Pi, \mathbb{I} -
\Pi\}$ where $\Pi$ corresponds to the outcome $0$ and $\mathbb{I} -
\Pi$ to the outcome $1$. Likewise, a POVM $\cM = \{M_i\}_{i\in\cI}$ is
projective if for each $i \in \cI$, $M_i$ is a projector.
\end{definition}

\begin{lemma}[Gentle Measurement \cite{Winter99}]
\label{lma:gentle} Let $\qstate{q}$ be a
quantum state and $(\Pi, \mathbb{I} - \Pi)$ be a binary projective
measurement such that $\Tr(\Pi\qstate{q}) \ge 1- \delta$. Let the
post-measurement state on applying
$(\Pi, \mathbb{I} - \Pi)$ and conditioning on the first outcome be:
$\qstate{q}' = \Pi \qstate{q} \Pi/\Trace(\Pi\qstate{q})$. Then, we
have that
$\TD(\qstate{q}, \qstate{q}') \le 2\sqrt{\delta}$.
\end{lemma}

\begin{theorem}[Indistinguishability of Projective Implementation \cite{TCC:Zhandry20}]\label{thm:pi-ind}
Let $\qstate{q}$ be an efficiently constructible and possibly mixed
state, and $D_0, D_1$ be computationally indistinguishable
distributions. Then, for any inverse polynomial $\epsilon$ and any
function $\delta$, the following holds:
$$\Delta_{\epsilon}\big(\ProjImp(\cP_{D_0})(\qstate{q}),
\ProjImp(\cP_{D_1})(\qstate{q})\big) \le \delta$$
\end{theorem}

\begin{remark}
As noted in prior work \cite{EC:KitNis22}, for Theorem
\ref{thm:pi-ind} to hold, the indistinguishability of $D_0, D_1$
needs to hold for distinguishers that can efficiently construct
$\qstate{q}$.
\end{remark}

\begin{theorem}[State Repair \cite{FOCS:CMSZ21}]\label{thm:repair}
Let $\cE$ be a projective measurement on register $\cH$ with outcomes
in set $S$. Let $\cM$ be an $(\epsilon, \delta)$-almost projective
measurement on $\cH$. Consider parameters $T \in \mathbb{N}, s \in S,
p \in [0,1]$. Then, there exists a procedure
$\Repair_{T,p,s}^{\cM,\cE}$ which is to be applied on $\cH$ as follows:

\begin{itemize}
\item Apply $\cM$ to obtain $p \in [0,1]$.
\item Next, apply $\cE$ to obtain $s \in S$.
\item Then, apply $\Repair_{T,p,s}^{\cM,\cE}$.
\item Finally, apply $\cM$ again to obtain $p' \in [0,1]$.
\end{itemize}

Then, it is guaranteed that:

$$\Pr[\lvert p - p'\rvert > 2\epsilon] \le |S|\cdot\delta + |S|/T + 4\sqrt{\delta}$$

\end{theorem}

% !TEX root = main.tex

\section{Proof of Theorem \ref{thm:two-sup}}\label{sec:proof-two-sup}

\begin{proof}
We will consider the following sequence of hybrids:

\begin{description}
\item $\underline{\Hyb_0(1^\secp)}:$ This is the same as the
experiment $\expb{\qA}{two}{sup}(1^\secp, q, N)$:

\begin{enumerate}
\item For each $i \in [q]$, $\qCh$ performs the following:
\begin{itemize}
    \item Sample $x_i^0, x_i^1 \la [N]$ such that $x_i^0 \neq x_i^1$ 
and $b_i \la \bit$. Define $Q_i$ as $Q_i \seteq \{x_i^0, 
x_i^1\}$.
\item Let $b$ be a placeholder for $b_i$. Construct the following on register $\qreg{A_i}$:
    $$\sigma_i \seteq \frac{1}{\sqrt{2}}\ket{x_i^0} +
    (-1)^{b}\frac{1}{\sqrt{2}}\ket{x_i^1}$$
\end{itemize}
\item $\qCh$ sends the registers $(\qreg{A_1}, \ldots, \qreg{A_q})$
to $\qA$.
\item $\qA$ sends $(g_1, \ldots, g_q)$ and a value $\msg$ to $\qCh$.
\item $\qCh$ checks if there exists $i \in [q]$ such that $\msg \in
Q_i$. If not, it outputs $\bot$. Let $s$ be such an index.
\item If $g_s(x_s^0, x_s^1) = b_s$ holds, $\qCh$ outputs
$\top$. Else, it outputs $\bot$.
\end{enumerate}

\item $\underline{\Hyb_1(1^\secp)}:$ This is similar to
$\Hyb_0(1^\secp)$, with the following differences colored in red:

\begin{enumerate}
\item For each $i \in [q]$, $\qCh$ performs the following:
\begin{itemize}
\item \textcolor{red}{Sample $x_i^0, x_i^1 \la [N]$ such that $x_i^0 \neq x_i^1$. Define $Q_i$ as $Q_i \seteq \{x_i^0, 
        x_i^1\}$.}

    \item \textcolor{red}{$\qCh$ constructs the following state $\sigma_i$ on registers
$\qreg{C_i}$ and $\qreg{A_i}$.
$$\sigma_i \seteq \frac12
\sum_{c \in
\bit}\ket{c}_{\qreg{C_i}}\otimes \big(\ket{x_i^0} +
(-1)^{c}\ket{x_i^1}\big)_{\qreg{A_i}}$$}
\end{itemize}

\item $\qCh$ sends the registers $(\qreg{A_1}, \ldots, \qreg{A_q})$
to $\qA$.
\item $\qA$ sends $(g_1, \ldots, g_q)$ and a value $\msg$ to $\qCh$.
\item $\qCh$ checks if there exists $i \in [q]$ such that $\msg \in
Q_i$. If not, it outputs $\bot$. Let $s$ be such an index.
\item \textcolor{red}{$\qCh$ measures register $\qreg{C_s}$ in the
computational basis to obtain outcome $c_s$.}
\item If $g_s(x_s^0, x_s^1) = \textcolor{red}{c_s}$ holds, $\qCh$ outputs
$\top$. Else, it outputs $\bot$.
\end{enumerate}

\item $\underline{\Hyb_2(1^\secp)}:$ This is similar to
$\Hyb_1(1^\secp)$, with the following differences colored in red:

\begin{enumerate}
\item For each $i \in [q]$, $\qCh$ performs the following:
\begin{itemize}
\item \textcolor{black}{Sample $x_i^0, x_i^1 \la [N]$ such that $x_i^0 \neq x_i^1$. Define $Q_i$ as $Q_i \seteq \{x_i^0, 
        x_i^1\}$.}

    \item \textcolor{black}{$\qCh$ constructs the following state $\sigma_i$ on registers
$\qreg{C_i}$ and $\qreg{A_i}$.
$$\sigma_i \seteq \frac12
\sum_{c \in
\bit}\ket{c}_{\qreg{C_i}}\otimes \big(\ket{x_i^0} +
(-1)^{c}\ket{x_i^1}\big)_{\qreg{A_i}}$$}
\end{itemize}

\item $\qCh$ sends the registers $(\qreg{A_1}, \ldots, \qreg{A_q})$
to $\qA$.
\item $\qA$ sends $(g_1, \ldots, g_q)$ and a value $\msg$ to $\qCh$.
\item $\qCh$ checks if there exists $i \in [q]$ such that $\msg \in
Q_i$. If not, it outputs $\bot$. Let $s$ be such an index.

\item \textcolor{red}{Let $\chec[x_s^0, x_s^1, \msg]$ be a function such that
$\chec[x_s^0, x_s^1, \msg](u) = \top$ if $\msg = x_s^u$ and $\bot$ otherwise.
Apply the following map to the register $\qreg{C_s}$ in the Hadamard
basis, and an ancilla register $\qreg{OUT}$ initialized to $\ket{0}$:
$$
\ket{u}_{\qreg{C_s}}\ket{w}_{\qreg{OUT}} \mapsto
\ket{u}_{\qreg{C_s}}\ket{w \xor \chec[x_s^0, x_s^1, \msg](u)}_{\qreg{OUT}}$$
}

\item \textcolor{red}{$\qCh$ measures register $\qreg{OUT}$ in the computational
basis to get outcome $\out$. If $\out = \bot$, output $\bot$.}

\item \textcolor{black}{$\qCh$ measures register $\qreg{C_s}$ in the
computational basis to obtain outcome $c_s$.}
\item If $g_s(x_s^0, x_s^1) = \textcolor{black}{c_s}$ holds, $\qCh$ outputs
$\top$. Else, it outputs $\bot$.
\end{enumerate}

\item $\underline{\Hyb_3(1^\secp)}:$ This is similar to
$\Hyb_2(1^\secp)$, with the following differences colored in red:

\begin{enumerate}
\item For each $i \in [q]$, $\qCh$ performs the following:
\begin{itemize}
\item \textcolor{black}{Sample $x_i^0, x_i^1 \la [N]$ such that $x_i^0 \neq x_i^1$. Define $Q_i$ as $Q_i \seteq \{x_i^0, 
        x_i^1\}$.}

    \item \textcolor{black}{$\qCh$ constructs the following state $\sigma_i$ on registers
$\qreg{C_i}$ and $\qreg{A_i}$.
$$\sigma_i \seteq \frac12
\sum_{c \in
\bit}\ket{c}_{\qreg{C_i}}\otimes \big(\ket{x_i^0} +
(-1)^{c}\ket{x_i^1}\big)_{\qreg{A_i}}$$}
\end{itemize}

\item $\qCh$ sends the registers $(\qreg{A_1}, \ldots, \qreg{A_q})$
to $\qA$.
\item $\qA$ sends $(g_1, \ldots, g_q)$ and a value $\msg$ to $\qCh$.
\item $\qCh$ checks if there exists $i \in [q]$ such that $\msg \in
Q_i$. If not, it outputs $\bot$. Let $s$ be such an index.

\item \textcolor{red}{$\qCh$ measures registers $\qreg{C_1}, \ldots,
\qreg{C_q}$ in the Hadamard basis to get outcomes $c_1' \ldots, c_s'$
respectively. For $u \seteq c_s'$, if $\msg \neq x_s^u$, $\qCh$
outputs $\bot$.}

\item \textcolor{black}{$\qCh$ measures register $\qreg{C_s}$ in the
computational basis to obtain outcome $c_s$.}
\item If $g_s(x_s^0, x_s^1) = \textcolor{black}{c_s}$ holds, $\qCh$ outputs
$\top$. Else, it outputs $\bot$.
\end{enumerate}

\item $\underline{\Hyb_4(1^\secp)}:$ This is similar to
$\Hyb_3(1^\secp)$, with the following differences colored in red:

\begin{enumerate}
\item For each $i \in [q]$, $\qCh$ performs the following:
\begin{itemize}
\item \textcolor{black}{Sample $x_i^0, x_i^1 \la [N]$ such that $x_i^0 \neq x_i^1$. Define $Q_i$ as $Q_i \seteq \{x_i^0, 
        x_i^1\}$.}

    \item \textcolor{black}{$\qCh$ constructs the following state $\sigma_i$ on registers
$\qreg{C_i}$ and $\qreg{A_i}$.
$$\sigma_i \seteq \frac12
\sum_{c \in
\bit}\ket{c}_{\qreg{C_i}}\otimes \big(\ket{x_i^0} +
(-1)^{c}\ket{x_i^1}\big)_{\qreg{A_i}}$$}
\end{itemize}

\item \textcolor{red}{$\qCh$ measures registers $\qreg{C_1}, \ldots,
        \qreg{C_q}$ in
the Hadamard basis to get outcomes $c_1', \ldots, c_q'$.}

\item $\qCh$ sends the registers $(\qreg{A_1}, \ldots, \qreg{A_q})$
to $\qA$.
\item $\qA$ sends $(g_1, \ldots, g_q)$ and a value $\msg$ to $\qCh$.
\item $\qCh$ checks if there exists $i \in [q]$ such that $\msg \in
Q_i$. If not, it outputs $\bot$. Let $s$ be such an index.

\item \textcolor{red}{For $u \seteq c_s'$, if $\msg \neq x_s^u$, $\qCh$
outputs $\bot$.}

\item \textcolor{black}{$\qCh$ measures register $\qreg{C_s}$ in the
computational basis to obtain outcome $c_s$.}
\item If $g_s(x_s^0, x_s^1) = \textcolor{black}{c_s}$ holds, $\qCh$ outputs
$\top$. Else, it outputs $\bot$.
\end{enumerate}

\end{description}

\begin{remark}
The statements of the following claims are meant for arbitrary
$t\in\mathbb{N}$ and appropriately chosen $N = O(t^2q^2)$.
\end{remark}

\begin{claim}
The probability that $\qCh$ outputs $\bot$ in Step 6. of $\Hyb_4$ is
at most $1/4t^2$.
\end{claim}
\begin{proof}
Notice that for $i \in [q]$, the states $\sigma_i$ are of the
following form:

$$\sigma_i = \frac12
\sum_{c\in\bit}\ket{c}_{\qreg{C_i}}\otimes\big(\ket{x_i^0} +
    (-1)^c\ket{x_i^1}\big) =
    \frac{1}{\sqrt2}\big(\ket{+}_{\qreg{C_i}}\ket{x_i^0}_{\qreg{A_i}} +
    \ket{-}_{\qreg{C_i}}\ket{x_i^1}_{\qreg{A_i}}\big)$$

Since $\qCh$ measures the registers $\qreg{C_1}, \ldots, \qreg{C_q}$
in the Hadamard basis to get $c_1', \ldots, c_q'$, when $\qA$
receives the registers $\qreg{A_1}, \ldots, \qreg{A_q}$, the states
are of the following form:

$$\ket{c_1'}_{\qreg{C_1}}\ket{x_1^{c_1'}}_{\qreg{A_1}}, \ldots,
\ket{c_q'}_{\qreg{C_q}}\ket{x_q^{c_q'}}_{\qreg{A_q}}$$

Now, let $Q$ be a multi-set such that $Q \seteq Q_1 \cup \ldots \cup
Q_q$. Let $\Event$ be the event that all the elements of $Q$ are
distinct. It is easy to see from a birthday bound analysis that if
$N = O(q^2 t^2)$, we can
ensure that the probability $\Event$ does \emph{not} occur is bounded by
$1/8t^2$. Consider now the set $S \seteq
Q \setminus \{x_1^{c_1'}, \ldots, x_q^{c_q'}\}$. Conditioned on
$\Event$ occurring, $S$ is a uniformly random subset of $[N]$ of size
$q$. Notice that for $\qA$ to cause an abort, it must produce $\msg$
such that $\msg \in S$. However, the only information $\qA$ has are
the values $x_1^{c_1'}, \ldots, x_q^{c_q'}$ which are
independent of the values in $S$.
Therefore, we have that $\Pr[\msg \in S | \Event]
\le q/N$, which is less than $1/8t^2$ for appropriate $N =
O(q^2t^2)$. Consequently, the probability that $\qCh$ outputs $\bot$
in Step 6. of $\Hyb_4$ is bounded by $1/8t^2 + 1/8t^2 = 1/4t^2$.
\fuyuki{If I understand correctly, this claim want to bound the probability that the experiment aborts as the last sentence of the proof says, not the probability that the experiment outputs $\bot$ as the statement of the claims says. This is confusing. We should not use the same symbol to indicate the adversary's failure and the experiment's (artificial) abort.}
\nikhil{Changed last sentence of the proof to match the claim.}
\end{proof}

\begin{claim}
The probability that $\qCh$ outputs $\bot$ in Step 5. of $\Hyb_3$ is
at most $1/4t^2$.
\end{claim}
\begin{proof}
This follows from the fact that operations performed on different
registers commute. Hence, the Hadamard measurements on $\qreg{C_1},
\ldots, \qreg{C_q}$ can be performed after $\qA$ submits its response,
without altering the probability of $\bot$ from $\Hyb_4$.
\end{proof}

\begin{claim}
The probability that $\qCh$ outputs $\bot$ in Step 6. of $\Hyb_2$ is
at most $1/4t^2$.
\end{claim}
\begin{proof}
This follows directly because the same check as the one introduced in
$\Hyb_3$ is applied coherently in $\Hyb_2$.
\end{proof}

\begin{claim}
Probability $\qCh$ outputs $\top$ in $\Hyb_2$ is at most $1/2$.
\end{claim}
\begin{proof}
Notice that the measurement on register $\qreg{OUT}$ necessarily collapses
$\qreg{C_s}$ in the Hadamard basis. Consequently, the following
computational basis measurement on $\qreg{C_s}$ yields a truly random
bit $c_s$ that is independent of $x_s^0, x_s^1$ and $g_s$.
\end{proof}

\begin{claim}
Probability $\qCh$ outputs $\top$ in $\Hyb_1$ is at most $1/2 + 1/t$.
\end{claim}
\begin{proof}
The only difference between $\Hyb_1$ and $\Hyb_2$ is the measurement
on register $\qreg{OUT}$ in $\Hyb_2$. Since we argued that this
measurement outputs $\bot$ with probability at most $1/4t^2$,
it follows from the gentle measurement lemma (Lemma
\ref{lma:gentle}) that these hybrids are $1/t$ close in
trace distance.
\end{proof}

\fuyuki{I did not see the reason why we need $\Hyb_2$. Why can't we go from $\Hyb_1$ to $\Hyb_3$ directly? (If I remember correctly, there isn't something like $\Hyb_2$ in the proof of BKM+.)}
\nikhil{I thought it was easier to understand this way. In $\Hyb_2$,
since the measurement is performed coherently, it would only end up
collapsing $\qreg{C_s}$. Also, I think the gentle measurement argument
relating $\Hyb_1$ and $\Hyb_2$ makes more sense due to the coherent
measurement.}

Notice that in $\Hyb_1$, if the measurement on register $\qreg{C_s}$
were performed at the beginning itself, then it is equivalent to
$\Hyb_0$. Since operations on different registers commute, this
measurement can be performed after the response of $\qA$. This
proves that for every $q, t \in \mathbb{N}$,
$\Pr[\expb{\qA}{two}{sup}(1^\secp, q, N) \ra \top] \le 1/2 +
1/t$ for some $N = O(t^2q^2)$. It is also easy to see that for every
$q \in \mathbb{N}$ and $N=128q^2$,
$\Pr[\expb{\qA}{two}{sup}(1^\secp, q, N) \ra \top] \le 3/4$.
\end{proof}

% traceability. !TEX root = main.tex

\section{Proof of Theorem \ref{thm:const}}\label{sec:proof-const}

\begin{proof}
Let $\qA$ be a QPT adversary in
$\expb{\SKL,\qA}{std}{kla}(1^\secp, q, \qF,\qE,t,\epsilon)$ for some
$\epsilon = 1/\poly(\secp)$ and let $\win$
denote the event that $\qA$ wins the SKL experiment. Assume $\win$
occurs with $\nonnegl(\secp)$ probability. Let $(\qsk^1, \vk^1),
\ldots, (\qsk^q, \vk^q)$ denote the leased secret-key and
verification-key pairs sampled by $\qKG$ in the experiment. For each
$j \in [q]$, we have $\qsk^j = (\qsk^j_i)_{i\in[k]}$ and $\vk^j =
(\vk_i^j)_{i\in[k]}$ by construction, where $\vk_i^j = (v_i^j,
w_i^j, \sk_{i,v}^j, \sk_{i,w}^j, b_i^j)$. For each $(i, j) \in [k]
\times [q]$, define the set $Q_i^j = \{v_i^j, w_i^j\}$. For each $i
\in [k]$, let $Q_i \seteq Q_i^1 \cup \ldots \cup Q_i^q$. Recall that
$\qA$ outputs certificates $\cert^1, \ldots, \cert^q$ and a quantum
program $\qP^*$.

Let $(\id_1^*, \ldots, \id_k^*)$ be computed as $(\id_1^*, \ldots,
\id_k^*) \la \MLTT.\qTrace(\mltt.\msk, \qP^*, \allowbreak0.9\epsilon)$ and
$\GoodExt_\epsilon$ denote the event that
$(\id_1^*, \ldots, \id_k^*) \in Q_1 \times \ldots \times Q_k$. Let
$\APILive_\epsilon$ denote the event that the $\epsilon$-good test wrt
$(f, \qE, t)$ outputs 1 when applied to $\qP^*$.

Assume
for contradiction that $\Pr[\lnot \GoodExt_\epsilon \mid
\APILive_\epsilon] =
\nonnegl(\secp)$. We will then construct the following reduction
algorithm $\qR^\qA$ that breaks the traceability of the MLTT scheme
$\MLTT$ for $N = 128q^2$ and $k = \secp$:

\begin{description}
\item \underline{Execution of $\qR^\qA$ in
    $\expb{\MLTT,\qR}{multi}{trace}(1^\secp,\qF,\qE,t,\epsilon, N, k)$}:

\begin{enumerate}
\item $\qCh$ samples $(\mltt.\msk, f, \aux_f) \la \MLTT.\Setup(1^\secp, N, k)$ and
sends $\aux_f$ to $\qR$. $\qR$ sends $\aux_f$ to $\qA$.
\item $\qR$ sends $2q$ to $\qCh$.
\item For each $(i, j) \in [k] \times [2q]$, $\qCh$ samples $\id_i^j
\la [N]$ and computes $\sk_i^j \la \KG(\mltt.\msk, i, \id_i^j)$. It sends
$\{\id_i^j, \sk_i^j\}_{(i,j)\in [k] \times [2q]}$ to $\qR$.
Define $Q_i$ to be the multi-set $Q_i \seteq
\{\id_i^j\}_{j \in [2q]}$.

\item For each $(i, j) \in [k] \times [q]$, $\qR$ does the
following:
\begin{itemize}
    \item Sample $b_i^j \la \bit$. Set $v_i^j \seteq \id_i^j$ and
        $w_i^j \seteq \id_i^{q+j}$.
\item Compute $\sk_{i,v}^j \seteq \sk_i^j$ and
    $\sk_{i,w}^j \seteq \sk_i^{q+j}$.
\item Compute the state $\rho_i^{j}$ on registers $\qreg{C_i^j,
D_i^j}$ similar to that of $\SKL.\qKG$.
\end{itemize}
\item For each $j \in [q]$, $\qR$ computes $\qsk^j \seteq
(\rho^j_i)_{i\in[k]}$ and $\vk^j \seteq (\vk^j_i)_{i\in[k]}$ where
$\vk_i^j \seteq (v_i^j, w_i^j, \sk_{i,v}^j, \sk_{i,w}^j, b_i^j)$.
Then, it sends $\qsk^1, \ldots, \qsk^q$ to $\qA$.
\item $\qA$ sends $(\cert^1, \ldots, \cert^q)$ and a quantum program $\qP^* =
(U, \rho)$ to $\qR$.\ryo{Shouldn't $\qR$ check
$(\cert^1,\ldots,\cert^q)$ are valid or not?}\nikhil{I believe its not
needed. As long as $\Pr[\lnot \GoodExt \land \APILive] = \nonnegl(\secp)$,
security of $\MLTT$ is broken.}
\item $\qR$ outputs the quantum program $\qP^*$. \nikhil{$\qR$ doesn't
need to check if $\qP^*$ is good either.}
\item $\qCh$ tests if $\qP^*$ is $\epsilon$-good wrt $(f, \qE, t)$. If not,
    it outputs $\bot$.
\item $\qCh$ runs $(\id_1^*, \ldots, \id_k^*) \la
    \MLTT.\qTrace(\mltt.\msk,
\qP^*, 0.9\epsilon)$.

\item If $(\id_1^*, \ldots,
\id_k^*) \in Q_1 \times \ldots \times Q_k$, $\qCh$ outputs $\bot$.
Else, it outputs $\top$.
\end{enumerate}
\end{description}

Observe that the view of $\qA$ in $\qR$ is indistinguishable from its
view in the standard-KLA experiment for $\SKL$.
Since we have $\Pr[\Win] = \nonnegl(\secp)$ and that
$\APILive_\epsilon$ occurs when $\Win$ occurs, we have
$\Pr[\APILive_\epsilon] = \nonnegl(\secp)$. This means $\Pr[\lnot
\GoodExt_\epsilon \land \APILive_\epsilon] = \nonnegl(\secp)$, which breaks
the security of $\MLTT$.

Now, assume that $\Pr[\lnot \GoodExt_\epsilon \mid
\APILive_\epsilon] \le \negl(\secp)$, which means that
$\Pr[\GoodExt_\epsilon \mid \APILive_\epsilon] \ge 1 -
\negl(\secp)$. We have that
$\Pr[\GoodExt_\epsilon \land \Win] =
\Pr[\Win]\cdot\Pr[\GoodExt_\epsilon \mid \Win] =
\nonnegl(\secp)\cdot \nonnegl(\secp) = \nonnegl(\secp)$ by assumption
and because $\APILive_\epsilon$ occurs whenever $\Win$ occurs.

Moreover, we must also have 
$\Vrfy(\vk^1, \cert^1) = \ldots = \Vrfy(\vk^q, \cert^q) =
\top$ conditioned on $\win$.
We will exploit this fact to construct the following
reduction $\qB$, which breaks the collusion-resistant security of
two-superposition states (for $k = \secp$ and $N = 128q^2$).

\begin{description}
\item \underline{Execution of $\qB^\qA$ in
    $\expb{\qB}{two}{sup}(1^\secp, k, q, N)$}:

\begin{enumerate}
\item $\qB$ samples $(\mltt.\msk, f, \aux_f) \la \MLTT.\Setup(1^\secp, N, k)$
and sends $\aux_f$ to $\qA$.

\item For each $(i, j) \in [k] \times [q]$, $\qCh$ performs the
    following:
\begin{itemize}
    \item Sample $v_i^j, w_i^j \la [N]$ and $b_i^j \la \bit$.
\item Set $b \seteq b_i^j$ and construct the following state on
    register $\qreg{C_{i,j}}$:
    $$\sigma_i^j \seteq \frac{1}{\sqrt{2}}\ket{v_i^j}_{\qreg{C_{i,j}}} +
    (-1)^{b}\frac{1}{\sqrt{2}}\ket{w_i^j}_{\qreg{C_{i,j}}}$$
\end{itemize}
\item For each $i \in [k]$, $\qCh$ sets $\sigma_i \seteq \sigma_i^1 \otimes
    \ldots \otimes \sigma_i^q$ and sends $\sigma_i$ to $\qB$.

\item For each $(i, j) \in [k] \times [q]$, $\qB$ performs the
following:
\begin{itemize}
\item Initialize a register $\qreg{D_{i,j}}$ to $\ket{0\ldots0}$.
Apply the following map to the registers $\qreg{C_{i,j}},
\qreg{D_{i,j}}$:

$$\ket{u}_{\qreg{C_{i,j}}}\ket{z}_{\qreg{D_{i,j}}} \mapsto
\ket{u}_{\qreg{C_{i,j}}}\ket{z \xor
\MLTT.\KG(\mltt.\msk,i,u)}_{\qreg{D_{i,j}}}
$$

Let the resulting state be denoted as $\qsk_i^j$.
\end{itemize}

\item For each $j \in [q]$, $\qB$ sets $\qsk^j \seteq
(\qsk_i^j)_{i\in[k]}$ and sends $\qsk^j$ to $\qA$.

\item $\qA$ sends $(\cert^1, \ldots, \cert^q)$ and a program $\qP^*
= (U, \rho)$ to $\qB$.
\item For each $j \in [q]$, $\qB$ parses $\cert^j$ as $\cert^j =
(c_i^j, d_i^j)_{i \in [k]}$.

\item For each $(i, j) \in [k] \times [q]$, $\qB$ considers the
following function $g_i^j$:
\begin{description}
\item $\underline{g_i^j(v_i^j, w_i^j)}:$
\begin{itemize}
\item Compute $\sk_{i,v}^j \la \MLTT.\KG(\mltt.\msk, i, v_i^j)$.
\item Compute $\sk_{i,w}^j \la \MLTT.\KG(\mltt.\msk, i, w_i^j)$.
\item Output $c_i^j \cdot (v_i^j \xor w_i^j) \xor d_i^j \cdot
(\sk_{i,v}^j \xor \sk_{i,w}^j)$.
\end{itemize}
\end{description}

\item $\qB$ tests if $\qP^*$ is $\epsilon$-good wrt $(f, \qE, t)$.

\item $\qB$ runs $(\id_1^*, \ldots, \id_k^*)
\la \MLTT.\qTrace(\mltt.\msk, \qP^*, 0.9\epsilon)$.

\item For each $i \in [k]$, $\qB$ sends $(g_i^1, \ldots, g_i^q)$ and
$\id_i^*$ to $\qCh$.

\item For each $(i, j) \in [k] \times [q]$, let $Q_i^j \seteq
\{v_i^j, w_i^j\}$. For each $i\in[k]$, $\qCh$ performs the following:
\begin{itemize}
%\item Set $Q_i \seteq Q_i^1 \cup \ldots \cup Q_i^q$.
\item Check if there exists $j \in [q]$ such that $\id_i^* \in
Q^j_i$. If not, output $\bot$. Let $s \in [q]$ be such an index.
\item Check if $g_i^s(v_i^s, w_i^s) = b_i^s$ holds. If not, output $\bot$.
\end{itemize}
\item Output $\top$.
\end{enumerate}
\end{description}

Notice that the view of $\qA$ in the experiment is indistinguishable
from its view in the SKL experiment.
Observe now that the condition $g^j_i(v_i^j, w_i^j) = b_i^j$ holds for
every $(i, j) \in [k] \times [q]$, whenever $\win$ occurs.
This is because the algorithms $\Vrfy(\vk^1, \cert^1), \ldots,
\Vrfy(\vk^q, \cert^q)$ will check that for
every $(i, j) \in [k] \times [q]$, it holds that $b_i^j = (c_i^j \|
d_i^j) \cdot (v_i^j \| \sk_{i,v}^j \xor w_i^j \| \sk_{i,w}^j) =
c_i^j \cdot (v_i^j \xor w_i^j) \xor d_i^j \cdot (\sk_{i,v}^j \xor
\sk_{i,w}^j) = g_i^j(v_i^j, w_i^j)$. Moreover, when
$\GoodExt_\epsilon$ occurs, 
for every $i \in [k]$, $\id_i^* \in Q_i$
where $Q_i \seteq Q_i^1 \cup
\ldots \cup Q_i^q$. Recall that $k = \secp$ and $N=128q^2$ as per the
construction. Consequently, $\qB$ breaks the collusion-resistant
security of two-superposition states (\cref{thm:two-sup-par}), giving us a
contradiction. Therefore, $\win$ can only occur with $\negl(\secp)$ probability.
\end{proof}

	\setcounter{tocdepth}{2}
	\tableofcontents

\fi
\else
	%%%%%%% Camera-ready region (no appendix) %%%%%%
\fi

\end{document}